\documentclass[12pt]{amsart}

\usepackage{latexsym, amsmath, amssymb,amsthm, natbib,verbatim, xy}
\xyoption{all}

\newtheorem{innercustomthm}{Theorem}
\newenvironment{customthm}[1]
  {\renewcommand\theinnercustomthm{#1}\innercustomthm}
  {\endinnercustomthm}

\usepackage{mathtools}

\usepackage{tgpagella}

\usepackage[left=1in,top=1in,right=1in,bottom=1in]{geometry}
\usepackage{setspace,color}
\usepackage{graphicx}
\usepackage{enumerate}
\usepackage[multiple]{footmisc}
\usepackage[leqno]{amsmath}
\usepackage{subfig}
\usepackage{hyperref}
\usepackage{xcolor}

\usepackage{etoolbox}
\patchcmd{\section}{\scshape}{\bfseries}{}{}
\makeatletter
\renewcommand{\@secnumfont}{\bfseries}
\makeatother

\newtheorem{theorem}{Theorem}
\newtheorem{corollary}{Corollary}
\newtheorem{definition}{Definition}
\newtheorem{lemma}{Lemma}
\newtheorem{proposition}{Proposition}
\newtheorem{claim}{Claim}

\theoremstyle{remark}

\newtheorem{example}{Example}

\def\D{\mathcal{D}}

\def\X{\mathcal{X}}    
\def\S{\mathcal{S}}   \def\C{\mathcal{C}}
    \def\T{\mathcal{T}}

\def\t{\tilde t}
\def\d{\tilde d}
\def\c{\tilde c}
\def\xii{\hat \xi}
\def\xij{\bar \xi}

\def\im{initial matching}

\newcommand{\df}[1]{\textbf{\textit{#1}}}

\newcommand{\abs}[1]{\left| #1 \right|}
\newcommand{\wt}[1]{\widetilde{#1}}

\newcommand{\chd}{Ch_d}

\newcommand{\ieh}[1]{{\color{purple} IEH: #1 }}
\newcommand{\fk}[1]{{\color{red} FK: #1 }}

\newcommand{\mby}[1]{{\color{blue} MBY: #1 }}

\begin{document}

\title[Interdistrict School Choice]{Interdistrict School Choice: A Theory of Student Assignment$^{\dagger}$}

\author[Hafalir, Kojima, and Yenmez]{Isa E. Hafalir \and Fuhito Kojima  \and M. Bumin Yenmez$^{*}$}

\thanks{
We thank Mehmet Ekmekci, Haluk Ergin, Yuichiro Kamada, Kazuo Murota, Tayfun S\"{o}nmez, Utku \"{U}nver, Rakesh Vohra, and the audiences at various seminars and conferences. Lukas Bolte, Ye Rin Kang, and especially Kevin Li provided superb research assistance. Kojima acknowledges financial support from the National Research Foundation through its Global Research Network Grant (NRF-2016S1A2A2912564).
Hafalir is affiliated with the UTS Business School, University of Technology Sydney, Sydney, Australia; Kojima is with the Department of Economics, Stanford University, 579 Serra Mall, Stanford, CA, 94305; Yenmez is with the Department of Economics, Boston College, 140 Commonwealth Ave, Chestnut Hill, MA, 02467. Emails: \texttt{isa.hafalir@uts.edu.au}, \texttt{fkojima@stanford.edu},
\texttt{bumin.yenmez@bc.edu}.}


\begin{abstract}
Interdistrict school choice programs---where a student can be assigned to a school outside of her district---are widespread in the US, yet the market-design literature has not considered such programs. We introduce a model of interdistrict school choice and present two mechanisms that produce \emph{stable} or \emph{efficient} assignments. We consider three categories of policy goals on
assignments and identify when the mechanisms can achieve them. By introducing a novel framework of interdistrict school choice, we provide a new avenue of research in market design.
\end{abstract}

\date{\today, First draft: July 15, 2017}

\maketitle



\section{Introduction}
School choice is a program that uses preferences of children and their parents over public schools to assign children to schools.
It has expanded rapidly in the United States and many other countries in the last few decades. Growing popularity and interest in school choice stimulated research in market design,  which has not only studied this problem in the abstract, but also contributed to designing specific assignment mechanisms.\footnote{\label{fn:SchoolChoiceImplementation}See \citet{abdul05,abdulka05b,abdulkadiroglu/pathak/roth:09} for details of the implementation of these new school choice procedures in New York and Boston.}

Existing market-design research about school choice is, however, limited to \textit{intradistrict} choice, where each student is
assigned to a school only in her own district. In other words, the literature has not studied \textit{interdistrict} choice,
where a student can be assigned to a school outside of her district. This is a severe limitation for at least two reasons.
First, interdistrict school choice is widespread: some form of it is practiced in 43 U.S.
states.\footnote{See \url{http://ecs.force.com/mbdata/mbquest4e?rep=OE1705}, accessed on July 14, 2017.} Second, as we illustrate in detail below, many policy goals in school choice impose constraints across districts in reality, but the existing literature assumes away such constraints. This omission limits our ability to analyze these policies of interest.

In this paper, we propose a model of interdistrict school choice. Our paper  builds upon matching models
in the tradition of \citet{gale62}.\footnote{We use the terms \emph{assignment} and \emph{matching} interchangeably for the rest of
the paper.} We study mechanisms and interdistrict admissions
rules to assign students to schools under which a
variety of policy goals can be established, an approach similar to the intradistrict school choice literature \citep{abdulson03}. In our setting, however, policy goals are defined on the district level---or sometimes even over multiple districts---rather than the individual school level, placing our model outside of the standard setting.  To facilitate the analysis in this setting, we model the problem as matching with contracts \citep{hatfi04}  between students and districts in which a contract specifies the particular school within the district that the student attends.\footnote{One might suspect that an interdistrict school choice problem can readily be reduced to an intradistrict problem by relabeling a district as a school. This is not the case because, among other things, which school within a district a student is matched with matters for that student's welfare.}

Following the school choice literature, we begin our analysis by considering \emph{stability} (we also consider efficiency, as explained later). To define stability in our  framework, we assume that each district is endowed with an admissions rule represented by a choice function over sets of contracts. 
We focus our attention on the student-proposing deferred-acceptance mechanism (SPDA) of \cite{gale62}. In our setting, this mechanism is not only stable but also strategy-proof---i.e., it renders truthtelling a weakly dominant strategy for each student.


In this context, we formalize a number of important policy goals. The first is \emph{individual rationality} in the sense that
every student is matched with a weakly more preferred school than  the school she is initially matched with
(in the absence of interdistrict school choice). This is an important requirement, because if an interdistrict school choice
program harms students, then public opposition is expected and the program may not be sustainable. The second policy is what we call
\emph{the balanced-exchange policy}: The number of students that each district receives from the other districts must be the same as the number of students that it sends to the others. Balanced exchange is also highly
desired by school districts in practice. This is because each district's funding depends on the number
of students that it serves and, therefore, if
the balanced-exchange policy is not satisfied, then some districts may lose funding, possibly making the interdistrict
school choice program impossible. For each of these policy goals, we identify the necessary and sufficient condition for achieving that goal under SPDA as a restriction on district admissions rules.

Last, but not least, we also consider a requirement that there be enough student diversity in each district.
In fact, diversity appears to be the main motivation for many interdistrict school choice programs.\footnote{We refer to
\citet{Wells09} for a review and discussion of interdistrict integration programs.}
To put this into context, we note that the lack of diversity is prevalent under intradistrict school choice programs
even though they often seek diversity by controlled-choice constraints.\footnote{Examples
of controlled school choice include Boston before 1999, Cambridge, Columbus, and Minneapolis. See \cite{abdulson03} for details of these programs as well as analysis of controlled school choice.} This is perhaps unsurprising given that only  residents of the given district can participate in intradistrict school choice and there is often severe residential segregation. In fact, a number of studies such as \cite{rivkin94} and \cite{clotfelter1999public,clotfelter11} attribute the majority---as high as 80 percent for some data and measure---of racial and ethnic segregation in public schools to disparities between school districts rather than within school districts. Given this concern, many interdistrict choice programs explicitly list achieving diversity as their main goal.

A case in point is the \emph{Achievement and Integration (AI) Program} of the Minnesota Department of Education (MDE).
Introduced in 2013, the AI program incentivizes school districts for integration. A district is required to
participate in this program if the proportion of a racial group in the district is considerably higher than
that in a neighboring district. In particular, every year the MDE commissioner analyzes fall enrollment data
from every district and, when a district and one of its adjoining districts have a difference of 20 percent or
higher in the proportion of any group of enrolled \emph{protected students} (American Indian, Asian or Pacific Islander,
Hispanic, Black, not of Hispanic origin, and White, not of Hispanic origin),  the district with the higher percentage
is required to be in the AI program.\footnote{In Minnesota's AI program, if the difference in the proportion of
protected students at a school is 20 percent or higher than a school in the same district, the school
with the higher percentage is considered a \emph{racially identifiable school} (RIS) and districts with RIS schools
also need to participate in the AI program. In this paper, we focus on diversity issues across districts rather
than within districts. Diversity problems within districts are studied in the controlled school choice literature
that we discuss below.} In the 2015-16 school year, more than 120 school districts participated in this program
(Figure \ref{fig:AI Figure}, taken from MDE's website, shows school districts in the Minneapolis-Saint Paul
metro area that take part in this program).
\begin{figure}[htb]
 \centering
 \includegraphics[scale=0.12]{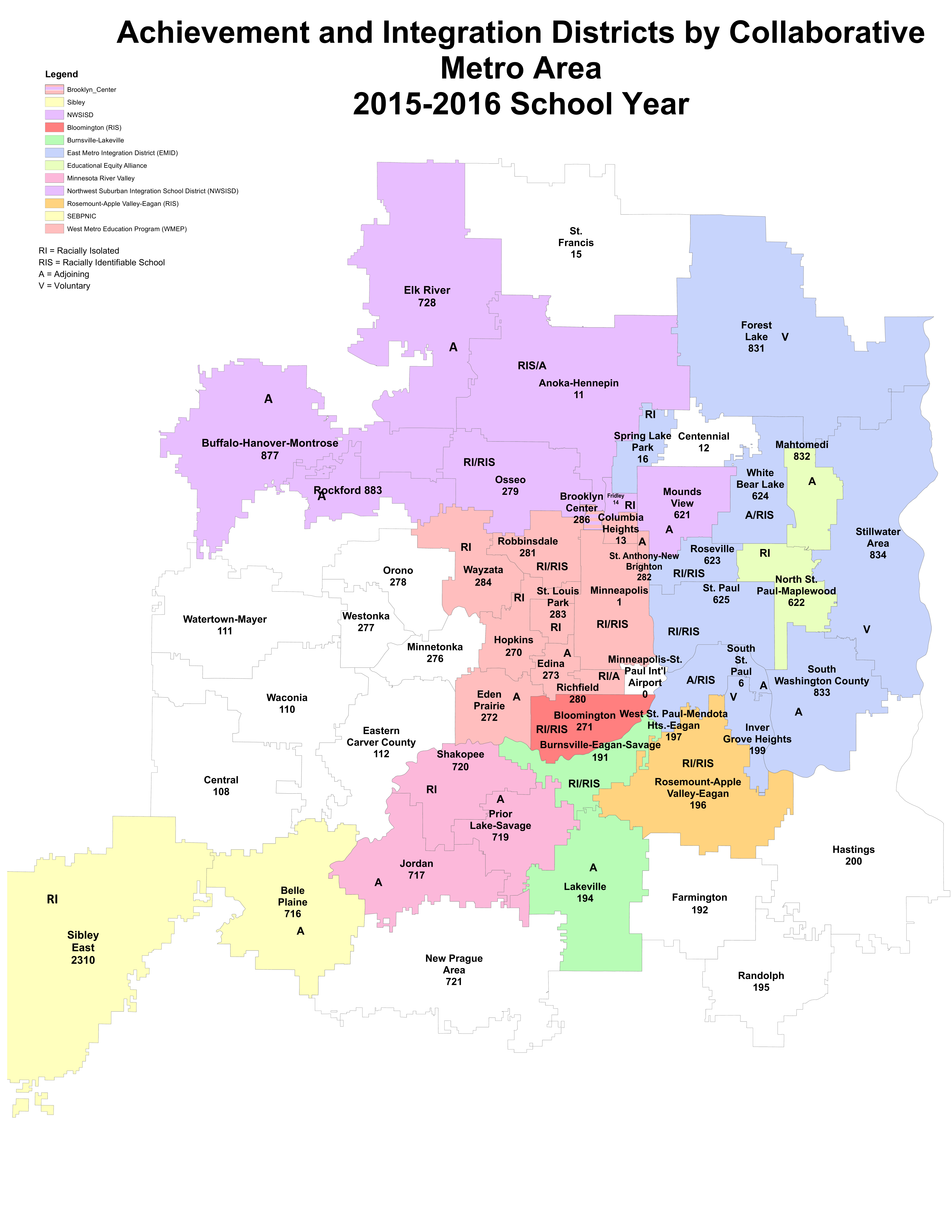}
\caption{Minnesota-Saint Paul metro area school districts participating in the AI program. The districts with the same color are adjoining districts that exchange students with one another.}
\label{fig:AI Figure}
\end{figure}

Motivated by Minnesota's AI program, we consider a policy goal requiring that  the difference in the proportions of each student type
across districts be within a given bound. Then, we provide a necessary and sufficient condition
for SPDA to satisfy the diversity policy. The condition provided is one on district admissions rules that have a structure of
type-specific ceilings, an analogue of the class of choice rules analyzed by \citet{abdulson03} and \citet{ehayeyi14} in the context of a more standard intradistrict school-choice problem.

Next, we turn our attention to efficiency. Given that the distributional policy goals work as constraints on matchings, we use the concept of \emph{constrained efficiency}. We say that a matching is constrained efficient if it satisfies the policy goal and is not Pareto dominated by any matching that satisfies the same policy goal. In addition, we require individual rationality and strategy-proofness.\footnote{Without individual rationality, all the other desired properties can be attained by a serial dictatorship.}
We first demonstrate an impossibility result; when the diversity policy is given as type-specific ceilings at the district level, there is no mechanism that satisfies the policy goal, constrained efficiency, individual rationality, and strategy-proofness.
By contrast, a version of the top trading cycles mechanism (TTC) of \cite{shasca74} satisfies these properties when the policy goal satisfies M-convexity,
a concept in discrete mathematics \citep{Murota:SIAM:2003}. We proceed to show that the balanced-exchange policy and an alternative form of diversity policy---type-specific ceilings at the individual school level instead of at the district level---are M-convex,
so TTC satisfies the desired properties for these policies. The same conclusion holds even when both of these policy goals are imposed simultaneously.

We also consider the case when there is a policy function that measures how well a matching
satisfies the policy goal. For example, diversity of a matching can be measured as its distance to an ideal distribution
of students. We show that TTC
satisfies the same desirable properties when the policy function satisfies \emph{pseudo M-concavity}, a notion
of concavity for discrete functions that we introduce. Furthermore, we show that there is an equivalence between two approaches based on the M-convexity of the policy set and the pseudo M-concavity of the policy function.
Therefore, both results can naturally be applied in different settings depending on how the policy goals are stated.

\subsection*{Related Literature}
Our paper is closely related to the controlled school choice literature that studies student diversity in schools in a given district. \citet{abdulson03} introduce a policy that imposes type-specific ceilings on each school. This policy has been analyzed by \citet{abdulkadirouglu2005college}, \citet{Ergin06}, and \citet{koj12}, among others. More accommodating policies using reserves rather than type-specific ceilings have been proposed and analyzed by
\citet{hayeyi13} and \citet{ehayeyi14}. The latter paper finds difficulties associated with hard floor constraints, an issue further analyzed by \citet{fragiadakis/iwasaki/troyan/ueda/yokoo:12} and \citet{frapet17}.\footnote{In addition to the works discussed above,
recent studies on controlled school choice and other two-sided matching problems with diversity concerns include \citet{westkamp10}, \citet{echyen12}, \citet{sonmez_rotc2011}, \citet{komson12}, \citet{dur_boston}, \citet{dur16}, and \citet{ngvoh17}.} In addition to sharing the motivation of achieving diversity, our paper is related to this literature in that we extend the type-specific reserve and ceiling constraints to district admissions rules. In contrast to this literature, however, our policy goals are imposed on districts rather than individual schools, which makes our model and analysis different from the existing ones.

The feature of our paper that imposes constraints on sets of schools (i.e., districts), rather than individual schools, is shared by several recent studies in matching with constraints. \citet{kamakoji-basic} study a model where the number of doctors who can be matched with hospitals in each region has an upper bound constraint. Variations and generalizations of this problem are studied by  \citet{goto2014improving,goto:16},
 \citet{biro:tcs:2010}, and
 \citet{kamakoji-concepts,kamakoji-iff}, among others.  While sharing the broad interest in constraints, these papers are different from ours in at least two major respects. First, they do not assume a set of hospitals is endowed with a well-defined choice function, while each school district has a choice function in our model. Second, the policy issues studied in these papers and those studied in ours are different given differences in the intended applications. These differences render our analysis distinct from those of the other papers, with none of their results implying ours and vice versa.

One of the notable features of our model is that district admissions rules do not necessarily satisfy the standard assumptions in the literature, such as \emph{substitutability}, which guarantee the existence of a stable matching. In fact, even a seemingly reasonable district admissions rule may violate substitutability because a district can choose at most one contract associated with the same student---namely just one contract representing one school that the student can attend. Rather, we make weaker assumptions following the approach of \cite{hatkom14}. This issue is playing an increasingly prominent role in matching with contracts literature; for example, in matching with constraints \citep{kamakoji-basic}, college admissions \citep{aygtur16,yen14}, and postgraduate admissions \citep{harosh16}, to name just a few.

Our analysis of Pareto efficient mechanisms is related to a small but rapidly growing literature that uses discrete
optimization techniques for matching problems. Closest to ours is \citet{suzuki17},  who show that a version of
TTC satisfies desirable properties if the constraint satisfies M-convexity.\footnote{See
\citet{kurata2016pareto} for an earlier work on TTC in a more specialized setting involving floor constraints at individual schools.}
Our analysis on efficiency builds upon and generalizes theirs. While the use of discrete
convexity concepts for studying efficient object allocation is still rare, it has been utilized
in an increasing number of matching problems such as two-sided matching with possibly
bounded transfer \citep{Fujishige:2006,Fujishige:2007}, matching with substitutable choice functions \citep{murota:metr:2013},
matching with constraints \citep{kojima-tamura-yokoo}, and trading networks \citep{candogan2016competitive}.

There is also a recent literature on segmented matching markets in a given district. \cite{mantur14} study a setting where
different clearinghouses can be coordinated, but not integrated in a centralized clearinghouse, and show how a stable
matching can be achieved. In a similar setting, \cite{durkes18} study sequential mechanisms
and show that these mechanisms lack desired properties. In another work, \cite{ekyen14} study the incentives
of a school to join a centralized clearinghouse. In contrast to these papers, we study which interdistrict school choice policies
can be achieved when districts are integrated.

At a high level, the present paper is part of research in resource allocation under constraints.
Real-life auction problems often feature constraints \citep{milgrom2009assignment}, and a great deal of attention was paid to cope with complex constraints in a recent FCC auction for spectrum allocation \citep{milgrom2014deferred}. Auction and exchange markets under constraints are analyzed by \citet{BLM04}, \citet{GPZ18}, and \citet{KSY18}. Handling constraints is also a subject of a series of papers on probabilistic assignment mechanisms \citep{bckm:10,che2013generalized,pycia2015decomposing,akbarpour2017approximate,nguyen2016assignment}. Closer to ours are \citet{dur2015two} and \citet{dur2018erasmus}. They consider the balance of incoming and outgoing members---a requirement that we also analyze---while modeling exchanges of members of different institutions under constraints. Although the differences in the model primitives and exact constraints make it impossible to directly compare their studies with ours, these papers and ours clearly share  broad interests in designing mechanisms under constraints.

The rest of the paper is organized as follows. Section \ref{sec:model} introduces the model. In Sections \ref{sec:stability} and \ref{sec:TTC}, we study when the policy goals can be satisfied together with stability and constrained efficiency, respectively. Section \ref{sec:conclusion} concludes. Additional results, examples, and omitted proofs are presented in the Appendix.

\section{Model}\label{sec:model}

In this section, we introduce our concepts and notation.

\subsection{Preliminary Definitions}
There exist finite sets of students $\S$, districts $\D$, and schools $\C$. Each student $s$ and school $c$ has 
a home district denoted by $d(s)$ and $d(c)$, respectively. Each student $s$ has a type $\tau(s)$ that can represent different aspects of the student such as the gender, race, socioeconomic status, etc. The set of all types is finite and denoted by $\T$. Each school $c$ has a capacity $q_c$, which is the maximum number of students that the school can enroll. There exist at least two school districts with one or more schools. For each district $d$, $k_d$ is the number of students whose home district is $d$. In each district, schools have sufficiently large capacities to accommodate all
students from the district, i.e., for every district $d$, $k_d\leq \sum_{c : d(c)=d} q_c$. For each type $t$,
$k^t$ is the number of type-$t$ students.

We model interdistrict school choice as a matching problem between students and districts. However, merely identifying the district with which a student is matched leaves the specific school she is enrolled in unspecified. To specify which school within a district the student is matched with, we use the notion of contracts: A contract $x=(s,d,c)$ specifies a student $s$, a district $d$, and a school $c$ within this district, i.e., $d(c)=d$.\footnote{For ease of exposition, a contract will sometimes be denoted by a pair $(s,c)$ with the understanding that the district associated with the contract is the home district of school $c$.} For any contract $x$, let $s(x)$, $d(x)$, and $c(x)$ denote the student, district, and school associated with this contract, respectively. Let $\X \equiv
\{(s,d,c) | d(c)=d\}$ denote the set of all contracts. For any set of contracts $X$, let $X_s$ denote the set of all
contracts in $X$ associated with student $s$, i.e., $X_s=\{x\in X| s(x)=s\}$. Similarly, let $X_d$ and $X_c$ denote the sets of all
contracts in $X$ associated with district $d$ and school $c$, respectively.

Each district $d$ has an \df{admissions rule} that is represented by a choice function $\chd$. Given a set of contracts $X$, the district chooses a subset of contracts associated with itself, i.e., $\chd (X)=\chd(X_d) \subseteq X_d$.

Each student $s$ has a strict preference order $P_s$ over all schools and the outside option of being unmatched, which is denoted by $\emptyset$. Likewise, $P_s$ is also used to rank contracts associated with $s$. 
Furthermore, we assume that
the outside option is the least preferred outcome, so for every contract $x$ associated with $s$, $x \mathrel{P_s} \emptyset$. The corresponding weak order is denoted by $R_s$. More precisely, for any two contracts $x,y$ associated with $s$, $x \mathrel{R_s} y$ if $x \mathrel{P_s} y$ or $x=y$.

A \df{matching} is a set of contracts. A matching $X$ is \df{feasible for students} if there exists at most one contract associated with every student in $X$. 
A matching $X$ is \df{feasible} if it is feasible for students and the number of contracts associated with every school in $X$ is at most its capacity, i.e., for any $c\in \C$, $\abs{X_c}\leq q_c$. 
We assume that there exists a feasible \df{\im} $\tilde{X}$ such that every student has exactly one contract.\footnote{In Appendix \ref{sec:cen}, we also consider the case when the initial matching for each district is constructed using student preferences and district admissions rules.} For any student $s$, if $\tilde{X}_s=\{(s,d,c)\}$  for some district $d$ and school $c$, then $c$ is called the \df{initial school} of $s$.

A \df{problem} is a tuple $(\S,\D,\C,\T,\{d(s),\tau(s),P_s\}_{s\in \S},\{Ch_{d}\}_{d\in \D}, \{d(c),q_c\}_{c\in \C},\tilde{X})$.
In what follows, we assume that all the components of a problem are publicly known except for student preferences. Therefore, we sometimes refer to a problem by the student preference profile which we denote as $P_{\S}$. The preference profile of a subset of students $S\subseteq \S$ is denoted by $P_S$.

\subsection{Properties of Admissions Rules}

A district admissions rule $\chd$ is \df{feasible} if it always chooses a feasible matching. It is \df{acceptant} if, for any contract $x$
associated with district $d$ and matching $X$ that is feasible for students; and if $x$ is rejected from $X$, then  at $\chd(X)$, either
\begin{itemize}
\item the number of students assigned to school $c(x)$ is equal to $q_{c(x)}$, or
\item the number of students assigned to district $d$ is at least $k_{d}$.
\end{itemize}

In words, when a district admissions rule is acceptant, a contract $x=(s,d,c)$ can be rejected by district $d$ from a set which is feasible for students only if either the capacity of school $c$ is filled or district $d$ has accepted at least $k_d$ students. Equivalently, if neither of these two conditions is satisfied, then the district has to accept the student. Throughout the paper, we assume that admissions rules are feasible and acceptant.\footnote{In Section \ref{sec:diversity}, we assume a weaker notion of acceptance when the admissions rule limits the number of students of each type that the district can accept.}

A district admissions rule satisfies \df{substitutability} if, whenever a contract is chosen from a set, it is also chosen from any subset containing that contract \citep{kelso82,roth84}. More formally, a district admissions rule $Ch_d$ satisfies substitutability  if, for every $x \in X\subseteq Y \subseteq \X$ with $x\in Ch_d(Y)$, it must be that $x\in Ch_d(X)$. A district admissions rule satisfies \df{the law of aggregate demand} (LAD) if the number of contracts chosen from a set is weakly greater than that of any of its subsets
\citep{hatfi04}. Mathematically, a district admissions rule $Ch_d$ satisfies LAD if, for every $X\subseteq Y \subseteq \X$, $\abs{Ch_d(X)}\leq \abs{Ch_d(Y)}$.\footnote{\cite{alkan02} and \cite{alkan03} introduce related monotonicity conditions.}
A \df{completion} of a district admissions rule $Ch_d$ is another admissions rule $Ch'_d$ such that for every matching $X$ either $Ch'_d(X)$ is equal to $Ch_d(X)$ or it is not feasible for students \citep{hatkom14}. Throughout the paper, we assume that district admissions rules have completions that satisfy substitutability and LAD.\footnote{\cite{hatkoj10} introduce other notions of weak substitutability.} In Appendix \ref{sec:examples}, we provide classes of district admissions rules that satisfy our assumptions.

\subsection{Matching Properties, Policy Goals, and Mechanisms}

A feasible matching $X$ satisfies \df{individual rationality} if every student weakly prefers her outcome in $X$ to
her initial school, i.e., for every student $s$, $X_s  \mathrel{R_s} \tilde{X}_s$.

A \df{distribution} $\xi \in \mathbb Z_+^{|\C|\times |\T|}$ is a vector such that the entry for school $c$ and type $t$ is denoted by $\xi_c^t$. The entry $\xi_c^t$ is interpreted as the number of type-$t$ students in school $c$ at $\xi$. Furthermore, let $\xi_d^t  \equiv \sum_{c:d(c)=d} \xi_c^t$, which is interpreted as the number of type-$t$ students in district $d$ at $\xi$. Likewise, for any feasible matching $X$,  the \df{distribution associated with $X$} is $\xi(X)$ whose $c,t$ entry $\xi_c^t(X)$ is the number of type-$t$ students assigned to school $c$ at $X$. Similarly, $\xi_d^t(X)$ denotes the number of type-$t$ students assigned to district $d$ at $X$.

We represent a distributional policy goal $\Xi$ as a set of distributions. The policy
that each student is matched without assigning any school more students than its capacity
is denoted by $\Xi^0$, i.e., $\Xi^0 \equiv \{\xi | \sum_{c,t}{\xi_c^t}=\sum_d k_d \text{ and }  q_c \geq \sum_t \xi_c^t \text{ for all } c\}$. A matching $X$ \df{satisfies the policy goal} $\Xi$ if the distribution associated with $X$ is in $\Xi$.

A feasible matching $X$ \df{Pareto dominates} another feasible matching $Y$ if every student weakly prefers
her outcome in $X$ to her outcome in $Y$ and at least one student strictly prefers the former to the latter.
Given a distributional policy goal, a feasible matching $X$ that satisfies the policy goal satisfies \df{constrained efficiency}
if there exists no feasible matching that satisfies the policy goal and Pareto dominates $X$.

A matching $X$ is \df{stable} if it is feasible and
\begin{itemize}
\item districts would choose all contracts assigned to them, i.e., $\chd(X)=X_d$ for every district $d$, and
\item there exist no student $s$ and no district $d$ who would like to match with each other, i.e., there exists no contract $x=(s,d,c) \notin X$ such that $x \mathrel{P_s} X_s$ and $x\in \chd(X\cup \{x\})$.
\end{itemize}

Stability was introduced by \cite{gale62} for the college admissions problem. In the context of assigning students to public schools, it is viewed as a fairness notion \citep{abdulson03}.

A \df{mechanism} $\phi$ takes a profile of student preferences as input and produces a feasible matching. The outcome for student $s$ at the reported preference profile $P_{\S}$ under mechanism $\phi$ is denoted as $\phi_s(P_{\S})$. A mechanism $\phi$ satisfies \df{strategy-proofness} if no student can misreport her preferences and get a strictly more preferred contract. More formally, for every student $s$ and preference profile $P_{\S}$,  there exists no preference $P'_s$ such that $\phi_s(P'_s,P_{\S\setminus \{s\}}) \mathrel{P_s} \phi_s(P_{\S})$. For any property on matchings, a mechanism satisfies the property if, for every preference profile, the matching produced by the mechanism satisfies the property.

\section{Achieving Policy Goals with Stable Outcomes}\label{sec:stability}
To achieve stable matchings with desirable properties, we use a generalization of the deferred-acceptance
algorithm of \cite{gale62}.\medskip

\paragraph{\textbf{Student-Proposing Deferred Acceptance Algorithm}}
\begin{description}
  \item[Step 1] Each student $s$ proposes a contract $(s,d,c)$ to district $d$ where $c$ is her most preferred school. Let $X_d^1$ denote the set of contracts proposed to district $d$. District $d$ tentatively accepts contracts in $\chd(X_d^1)$ and permanently rejects the rest. If there are no rejections, then stop and return $\cup_{d \in \mathcal D} \chd(X_d^1)$ as the outcome.
  \item[Step $\mathbf{n}$ ($\mathbf{n>1}$)] Each student $s$ whose contract was rejected in Step $n-1$ proposes a contract $(s,d,c)$ to district $d$ where $c$ is her next preferred school. If there is no such school, then the student does not make any proposals. Let $X_d^n$ denote the union of the set of contracts that were tentatively accepted by district $d$ in Step $n-1$ and the set of contracts that were proposed to district $d$ in Step $n$. District $d$ tentatively accepts contracts in $\chd(X_d^n)$ and permanently rejects the rest. If there are no rejections, then stop and return $\cup_{d \in \mathcal D} \chd(X_d^n)$.
\end{description}

The student-proposing deferred acceptance mechanism (SPDA) takes a profile of student preferences as input
and produces the outcome of this algorithm at the reported student preference profile.
When district admissions rules have completions that satisfy substitutability and LAD, SPDA is stable and strategy-proof \citep{hatkom14}. Therefore, when we analyze SPDA, we assume that students report their preferences truthfully.

We illustrate SPDA using the following example. We come back to this example later to study the effects of interdistrict school choice.

\begin{example}\label{ex:simple}
Consider a problem with two school districts, $d_1$ and $d_2$. District $d_1$ has school $c_{1}$ with capacity one and school $c_{2}$ with capacity two. District $d_2$ has school $c_{3}$ with capacity two. There are four students: students $s_{1}$ and $s_{2}$ are from district $d_{1}$, whereas students $s_{3}$ and $s_{4}$ are from district $d_{2}$. The \im{} is $\{(s_1,c_1),(s_2,c_2),(s_3,c_3),(s_4,c_3)\}$.

Given any set of contacts, district $d_1$ chooses students who have contracts with school $c_1$ first and then chooses from the remaining students who have contracts with school $c_2$. For school $c_1$, the district prioritizes students in the order $s_3\succ s_4 \succ s_1 \succ s_2$ and chooses one applicant if there is any. For school $c_2$, the district prioritizes students according to the order $s_1 \succ s_2 \succ s_3 \succ s_4$ and chooses as many applicants as possible without going over the school's capacity while ignoring the contracts of the students who have already been accepted at school $c_1$. 
 Likewise, district $d_2$ prioritizes students according to the order $s_3 \succ s_4 \succ s_1 \succ s_2$ and chooses as many applicants as possible without going over  the capacity of school $c_3$. 
These admissions rules are feasible and acceptant, and they have completions that satisfy substitutability and LAD.\footnote{In Appendix \ref{section:genex}, we provide a general class of admissions rules, including this one as a special case. We show that these admissions rules are feasible and acceptant, and they have completions that satisfy substitutability and LAD.}
In addition, student preferences are given by the following table,
\[
\begin{tabular}
[c]{llll}
\underline{$P_{s_{1}}$} & \underline{$P_{s_{2}}$} & \underline{$P_{s_{3}}$} &
\underline{$P_{s_{4}}$}\\
$c_{1}$ & $c_{3}$ & $c_{1}$ & $c_{2}$\\
$c_{2}$ & $c_{1}$ & $c_{2}$ & $c_{1}$\\
$c_{3}$ & $c_{2}$ & $c_{3}$ & $c_{3}$
\end{tabular}
\]
\smallskip

\noindent
which means that, for instance, student $s_1$ prefers $c_1$ to $c_2$ to $c_3$.

In this problem, SPDA runs as follows. At the first step, student $s_1$ proposes to district $d_1$ with contract $(s_1,c_1)$, student $s_2$ proposes to district $d_2$ with contract $(s_2,c_3)$, student $s_3$ proposes to district $d_1$ with contract $(s_3,c_1)$, and student $s_4$ proposes to district $d_1$ with contract $(s_4,c_2)$. District $d_1$ first considers contracts associated with school $c_1$, $(s_1,c_1)$ and $(s_3,c_1)$, and tentatively accepts $(s_3,c_1)$ while rejecting $(s_1,c_1)$ because student $s_3$ has a higher priority than student $s_1$ at school $c_1$. Then district $d_1$ considers contracts of the remaining students associated with school $c_2$. In this case, there is only one such contract, $(s_4,c_2)$, which is tentatively accepted. District $d_2$ considers contract $(s_2,c_3)$ and tentatively accepts it. The tentative matching is $\{(s_2,c_3), (s_3,c_1), (s_4,c_2)\}$. Since there is a rejection, the algorithm
proceeds to the next step.

At the second step, student $s_1$ proposes to district $d_1$ with contract $(s_1,c_2)$. District $d_1$ first considers contract $(s_3,c_1)$ and tentatively accepts it. Then district $d_1$ considers contracts $(s_1,c_2)$ and $(s_4,c_2)$ and tentatively accepts them both. District $d_2$ does not have any new contracts, so tentatively accepts $(s_2,c_3)$. Since there is no rejection, the
algorithm stops. The outcome of SPDA is $\{(s_1,c_2), (s_2,c_3), (s_3,c_1), (s_4,c_2)\}$.
\qed
\end{example}

In the rest of this section, we formalize three policy goals and characterize 
 conditions under which SPDA satisfies them.

\subsection{Individual Rationality}\label{sec:ir}
In our context, individual rationality requires that every student 
 is matched with a weakly more preferred school than her initial school. As a result, SPDA does not necessarily satisfy individual rationality even though each student is either unmatched or matched with a school that is more preferred than being unmatched.

If individual rationality is violated so that some students prefer their initial schools to the outcome of SPDA, then there may
be public opposition that harm interdistrict school choice efforts. For this reason, individual rationality is a desirable property for
policymakers. The following condition proves to play a crucial role for achieving this property.

\begin{definition}
A district admissions rule $\chd$ \df{respects the \im} if, for any student $s$ whose initial school $c$ is in district $d$ and matching $X$ that is feasible for students, $(s,d,c) \in X$ implies $(s,d,c) \in \chd(X)$.
\end{definition}

When a district's admissions rule respects the \im, it has to admit those contracts
associated with itself in which students apply to their initial schools from every matching that is feasible for students.
The following result shows that this is exactly the condition for SPDA to satisfy individual rationality.

\begin{theorem}\label{thm:welfimprove2}
SPDA satisfies individual rationality if, and only if, each district's admissions rule respects the \im.
\end{theorem}

The intuition for the ``if'' part of this theorem is simple. When district admissions rules respect the initial matching,
no student is matched with a school which is strictly less preferred than her initial school under SPDA because she is guaranteed to be accepted by that school if she applies to it. For the ``only if'' part of the theorem, we construct a specific student preference profile
such that SPDA assigns one student a strictly less preferred school than her initial school whenever there exists one district with an admissions rule that does not respect the initial matching.

In the next example, we illustrate SPDA with district admissions rules that respect the initial matching.

\begin{example}
Consider the problem in Example \ref{ex:simple}. Recall that in this problem, the outcome of SPDA is $\{(s_1,c_2), (s_2,c_3), (s_3,c_1), (s_4,c_2)\}$. This matching is not individually rational because student $s_1$ prefers her initial school $c_1$ to school $c_2$ that she is matched with. This observation is consistent with Theorem \ref{thm:welfimprove2} because the admissions rule of district $d_1$ does not respect the initial matching. In particular, $Ch_{d_1}(\{(s_1,c_1), (s_3,c_1)\})=\{(s_3,c_1)\}$, so student $s_1$ is rejected from a matching that is feasible for students and includes the contract with her initial school.

Now modify the priority ranking of district $d_1$ at school $c_1$ so that $s_1 \succ s_2 \succ s_3 \succ s_4$ but, otherwise, keep the construction of the district admissions rules and student preferences the same as before. With this change, district admissions rules respect the initial matching because each student is accepted when she applies to the district with her initial school.\footnote{In Appendix  \ref{ex:respecting}, we construct a class of district admissions rules that includes this admissions rule as a special case. These admissions rules are feasible and acceptant, and have completions that satisfy substitutability and LAD. Furthermore, they also respect the initial matching.} In particular, the proposal of student $s_1$ to district $d_1$ with her initial school $c_1$ is always accepted. With this modification, it is easy to check that the outcome of SPDA is $\{(s_1,c_1), (s_2,c_3), (s_3,c_2), (s_4,c_2)\}$. This matching satisfies individual rationality.
\qed
\end{example}

In some school districts, each student gets a priority at her neighborhood school, as in this example. In the absence of other types of priorities, neighborhood priority guarantees that SPDA satisfies individual rationality.


\subsection{Balanced Exchange}\label{sec:bal}
For interdistrict school choice, maintaining a balance of students incoming from and outgoing to other districts is important.
To formalize this idea, we say that a mechanism satisfies the \df{balanced-exchange} policy if the number of students that a
district gets from the other districts and the number of students that the district sends to the others are the same for
every district and for every profile of student preferences. Since district choice rules are acceptant and students
prefer every school to the outside option of being unmatched, every student is matched with a school under SPDA. Therefore,
for SPDA, this policy is equivalent to the requirement that the number of students assigned to a district must be
equal to the number of students from that district.

The balanced-exchange policy is important  because the funding that a district gets depends on the number of students it serves. Therefore, an interdistrict school choice program may not be sustainable if SPDA does not satisfy the balanced-exchange policy. For achieving this policy goal, the following condition on admissions rules proves important.

\begin{definition}
A matching $X$ is \df{rationed} if, for every district $d$, it does not assign strictly more students to
the district than the number of students whose home district is $d$. A district admissions rule is
\df{rationed} if it chooses a rationed matching from any matching that is feasible for students.
\end{definition}

When a district admissions rule is rationed, the district does not accept strictly more students than the
number of students from the district at any matching that is feasible for students. The result below establishes
that this property is exactly the condition to guarantee that SPDA satisfies the balanced-exchange policy.

\begin{theorem}\label{thm:balance}
SPDA satisfies the balanced-exchange policy if, and only if, each district's admissions rule is rationed.
\end{theorem}

To obtain the intuition for this result, consider a student. Acceptance requires that a district can reject all
contracts of this student only when the number of students assigned to the district is at least as large as the
number of students from that district. As a result, all students are guaranteed to be matched. In addition,
when district admissions rules are rationed, a district cannot accept more students than
the number of students from the district. These two facts together imply that the number of students assigned
to a district in SPDA is equal to the number of students from that district. Therefore, SPDA satisfies the
balanced-exchange policy when each district's admissions rule is rationed. Conversely, when there exists one
district with an admissions rule that fails to be rationed, then we can construct student preferences such
that this district is matched with strictly more students than the number of students from the district
in SPDA, which means that the outcome does not satisfy the balanced-exchange policy.

Now we illustrate SPDA when district admissions rules are rationed.

\begin{example}\label{ex:balance}
Consider the problem in Example \ref{ex:simple}. Recall that in this problem, the SPDA outcome is $\{(s_1,c_2), (s_2,c_3), (s_3,c_1), (s_4,c_2)\}$. Since there are three students matched with district $d_1$ and there are only two students from that district, SPDA does not satisfy the balanced-exchange policy. This is consistent with Theorem \ref{thm:balance} because the admissions rule of district $d_1$ is not rationed. In particular, $Ch_{d_1}(\{(s_1,c_2), (s_3,c_1), (s_4,c_2)\})=\{(s_1,c_2), (s_3,c_1), (s_4,c_2)\}$, so district $d_1$ accepts strictly more students than the number of students from there given a matching that is feasible for students.

Suppose that we modify the admissions rule of district $d_1$ as follows. If the district chooses a contract associated with school $c_1$, then at most one contract associated with school $c_2$ is chosen. Therefore, the district never chooses more than two contracts, which is the number of students from there. Therefore, the updated admissions rule is rationed.\footnote{In Appendix \ref{ex:rationed}, we construct a class of rationed district admissions rules that includes this admissions rule as a special case. These admissions rules are feasible and acceptant, and they have completions that satisfy substitutability and LAD.} With this change, it is easy to check that the SPDA outcome is $\{(s_1,c_2), (s_2,c_3), (s_3,c_1), (s_4,c_3)\}$, which satisfies the balanced-exchange policy.
\qed
\end{example}

An implication of Theorems \ref{thm:welfimprove2} and \ref{thm:balance} is that SPDA is guaranteed to satisfy
individual rationality and the balanced-exchange policy if, and only if, each district's
admissions rule respects the initial matching and is rationed.

\subsection{Diversity}\label{sec:diversity}
\label{sec:diversity} The third policy goal we consider is that of
diversity. More specifically, we are interested in how
to ensure that there is enough diversity across districts so that the student composition in
terms of demographics does not vary too much from district to district.

We are mainly motivated by a program that is used in the state of Minnesota.
State law in Minnesota identifies racially isolated (relative to one of
their neighbors) school districts and requires them to be in the \emph{%
Achievement and Integration (AI) Program}. The goal is to increase the
racial parity between neighboring school districts. We first introduce a
diversity policy in the spirit of this program: Given a constant $
\alpha \in [0,1]$, we say that a mechanism satisfies the
\textbf{\textit{$\alpha$-diversity policy}} if for all preferences, districts $d$ and $d'$, and
type $t$, the difference between the ratios
of type-$t$ students in districts $d$ and $d'$ is not more than $\alpha $.
We interpret $\alpha$ to be the maximum ratio difference tolerated
under the diversity policy; for instance, $\alpha=0.2$ for Minnesota.

We study admissions rules such that
 SPDA satisfies the $\alpha$-diversity
policy when there is interdistrict school choice.
Since this policy restricts the number of students across districts, a natural
starting point is to have type-specific ceilings at the district level.
However, it turns out that type-specific ceilings at the district level may
yield district admissions rules resulting in no stable matchings
(see Theorem \ref{thm:imposswithschools} in Appendix \ref{sec:districtlevel}).

Since there is an incompatibility between district-level type-specific
ceilings and the existence of a stable matching, we
impose type-specific ceilings at the school level as follows.

\begin{definition}
A district admissions rule $Ch_{d}$ has a \textbf{
\textit{school-level type-specific ceiling}} of $q_{c}^{t}$ at school $c$ for type-$t$ students if the number of type-$t$
students admitted cannot exceed this ceiling. 
More formally, for any matching $X$ that is feasible for students,
\begin{center}
$|\{x\in Ch_{d}(X)|\tau (s(x))=t,c\left( x\right) =c\}|\leq q_{c}^{t}$.
\end{center}
\end{definition}

Note that district admissions rules typically violate acceptance  once school-level type-specific
ceilings are imposed. This is because a student can be rejected
from a set that is feasible for students even when the number of applicants to each school
is smaller than its capacity and the number of applicants to the district is
smaller than the number of students from that district. Given this, we define a
weaker version of the acceptance assumption as follows.

\begin{definition}
A district admissions rule $Ch_{d}$ that has school-level type-specific ceilings is \textbf{weakly acceptant} if, for
any contract $x$ associated with a type-$t$ student and district $d$ and
matching $X$ that is feasible for students, if $x$ is rejected from $X$,
then at $Ch_{d}(X)$,

\begin{itemize}
\item the number of students assigned to school $c(x)$ is equal to $q_{c(x)}$, or

\item the number of students assigned to district $d$ is at least $k_{d}$, or

\item the number of type-$t$ students assigned to school $c(x)$
is at least $q_{c}^{t}$.
\end{itemize}
\end{definition}

In other words, a student can be rejected from a set that is feasible
for students only when one of these three conditions is satisfied.

In SPDA, a student may be left unassigned due to school-level type-specific ceilings
even when district admissions rules are weakly acceptant. To make sure that
every student is matched, we make the following assumption.

\begin{definition}\label{def:accommodate}
A profile of district admissions rules $(Ch_d)_{d\in \mathcal{D}}$
\textbf{\textit{accommodates unmatched students}} if for any student $s$ and
feasible matching $X$ in which student $s$ is unmatched, there exists
$x=(s,d,c) \in \mathcal{X}$ such that $x \in Ch_d(X\cup \{x\})$.
\end{definition}

When a profile of district admissions rules accommodates unmatched students, for any
feasible matching in which a student is unmatched, there exists a school
such that the district associated with the school would admit that student
if she applies to that school. For example, when each admissions rule respects the initial
matching, the profile of district admissions rules accommodates unmatched students because an
unmatched student's application to her initial school is always accepted.
Lemma \ref{lem:match} in Appendix \ref{sec:proofs} shows that when a profile of district admissions rules
accommodates unmatched students, every student is matched to a school in SPDA.

In general, accommodation of unmatched students may be in conflict with
type-specific ceilings because there may not be enough space for a student
type when ceilings are small for this type. To avoid this, we assume that type-specific
ceilings are high enough so that $(Ch_d)_{d\in \mathcal{D}}$ accommodates unmatched
students.\footnote{For instance, ignoring integer problems, $q_{d}^{t}\geq k_{d}\frac{k^{t}}
{\sum_{t'\in \mathcal{T}}k^{t'}}$ for all $t,d$, would make
ceilings compatible with this property as it would be possible to assign the
same percentage of students of each type to all districts.}

Our assumptions on district admissions rules allow us to
control the distribution of the SPDA outcome. In particular, the SPDA
outcome satisfies the following   conditions: (i)  $\sum_t\xi _{d}^{t}(X)=k_{d}$ for all
$d\in \mathcal{D}$, (ii) $\sum_{c\in \mathcal{C}}\xi _{c}^{t}(X)=k^{t}$ for
all $t\in \mathcal{T}$, (iii) $\sum_{t\in \mathcal{T}}\xi _{c}^{t}(X)\leq
q_{c}$ for all $c\in \mathcal{C}$, and (iv) $\xi _{c}^{t}(X)\leq q_{c}^{t}$ for
all $t\in \mathcal{T}$ and $c\in \mathcal{C}$.
We call any matching $X$ satisfying these conditions \textbf{legitimate}.

In this framework, type-$t$ ceilings of schools in district $d$ may
result in a floor of another type $t'$ in this district in the sense
that the number of type-$t'$ students in the district should be at
least a certain number. Moreover, this may further impose a ceiling for type
$t'$ in another district $d'$. To see this, suppose, for
example, that (i) there are two districts $d$ and $d^{\prime }$, (ii) in
each district, there is one school and 100 students, (iii) 100 students are
of type $t$ and 100 students are of another type $t^{\prime }$, and (iv)
each school has a type-$t$ ceiling of 60 and a type-$t^{\prime }$ ceiling of
70. In a legitimate matching, each district needs to have at least 40 type-$%
t^{\prime }$ students (because, otherwise, the number of type-$t$ students
in that district would have to be more than 60). Moreover, this would mean
that there cannot be more than 60 type-$t^{\prime }$ students in any
district (because, otherwise, there would need to be more than 40 type-$%
t^{\prime }$ students in the other district, contradicting the floor we just
calculated). Hence, in this example, in effect we have a floor of 40 and a
(further restricted) ceiling of 60 for type-$t^{\prime }$ students for each
district.

Faced with this complication, our approach is to find the tightest lower and
upper bounds induced by these constraints. For this purpose, a certain
optimization problem proves useful. More specifically, consider a
linear-programming problem where for each type $t$ and district $d$, we seek
the minimum and maximum values of $\sum_{c: d(c)=d} y_{c}^{t}$
subject to (i) $\sum_{t'\in \mathcal{T}} \sum_{c:d(c)=d'} y_{c}^{t^{\prime }}=k_{d^{\prime }}$ for all $d^{\prime }\in \mathcal{D}$,
(ii) $\sum_{c\in \mathcal{C}}y_{c}^{t^{\prime }}=k^{t^{\prime }}$ for all $%
t^{\prime }\in \mathcal{T}$, (iii) $\sum_{t^{\prime }\in \mathcal{T}%
}y_{c}^{t^{\prime }}\leq q_{c}$ for all $c\in \mathcal{C}$, and (vi) $%
y_{c}^{t^{\prime }}\leq q_{c}^{t^{\prime }}$ for all $t^{\prime }\in
\mathcal{T}$ and $c\in \mathcal{C}$.
Let $\hat{p}_{d}^{t}$ and $\hat{q}_{d}^{t}$ be the solutions to the minimization and maximization
problems, respectively.

Both of these optimization problems belong to a special class of
linear-programming problems called a minimum-cost flow problem, and many
computationally efficient algorithms to solve it are known in the literature.\footnote{
To see that our problem is a minimum-cost flow problem, note that we can
take $(k_{d})_{d\in \mathcal{D}}$ as the \textquotedblleft
supply,\textquotedblright\ $(k^{t})_{t\in \mathcal{T}}$ as the
\textquotedblleft demand,\textquotedblright\  $(q_{d}^{t})_{d\in \mathcal{%
D},t\in \mathcal{T}}$ as the \textquotedblleft arc capacity
bounds,\textquotedblright\ and the objective functions for $\hat{p}_{d}^{t}$
and $\hat{q}_{d}^{t}$ to be $\min y_{d}^{t}$ and $\min -y_{d}^{t}$,
respectively. These problems have an \textquotedblleft integrality
property\textquotedblright\ so that if the supply, demand, and bounds are
integers, then all the solutions are integers as well. As already mentioned,
many algorithms have been proposed to solve different objective functions
for these problems. For instance, the capacity scaling algorithm of
\cite{edmonds1972} gives the solutions in polynomial time. For more information,
see Chapter 10 of \cite{ahuja2017network}. We are grateful to Fatma
Kilinc-Karzan for helpful discussions.} A straightforward but important
observation is that $\hat{p}_{d}^{t}$ (resp. $\hat{q}_{d}^{t}$) is exactly
the lowest (resp. highest) number of type-$t$ students who can be matched to
district $d$ in a legitimate matching (Lemma \ref{lem:lower} in Appendix \ref{sec:proofs}). Given
this observation, we call $\hat{p}_{d}^{t}$ the \textbf{\textit{implied floor}}
and $\hat{q}_{d}^{t}$ the \textbf{\textit{implied ceiling}}.

Now we are ready to state the main result of this section.

\begin{theorem} \label{thm:diversity}
Suppose that each district admissions rule has school-level type-specific
ceilings and is rationed and weakly acceptant. Moreover,
suppose that the district admissions rule profile accommodates unmatched students.
Then, SPDA satisfies the $\alpha$-diversity policy if, and only if,
$\hat{q}_{d}^{t}/k_{d}-\hat{p}_{d^{\prime }}^{t}/k_{d^{\prime }}\leq
\alpha$ for every type $t$ and districts $d,d^{\prime }$ such that $d\neq
d^{\prime }$.
\end{theorem}

The proof of this theorem, given in Appendix \ref{sec:proofs}, is based on a
number of steps. First, as mentioned above, we note that $\hat{p}_{d}^{t}$
and $\hat{q}_{d}^{t}$ are the lower and upper bounds, respectively, of the
number
 of
type-$t$ students who can be matched  with district $d$ in any
legitimate matching. This observation immediately establishes the ``if'' part
of the theorem. Then, we further establish that the implied floors and
ceilings
can be achieved
simultaneously in the sense that, for any pair of districts $d$ and $%
d^{\prime }$ with $d\neq d^{\prime }$, there exists a legitimate matching
that assigns exactly $\hat{q}_{d}^{t}$ type-$t$ students to district $d$ and
exactly $\hat{p}_{d^{\prime }}^{t}$ type-$t$ students to district $d^{\prime
}$ (Lemma \ref{lem:suff}). In other words, we establish that the implied
ceiling and floor are achieved in two different districts, and they are
achieved at \textit{one} legitimate matching simultaneously. We complete the
proof of the theorem by constructing student preferences such that the outcome of SPDA
achieves these bounds. In Appendix \ref{subsec:diversityexample}, we provide
an example that illustrates Theorem \ref{thm:diversity}.
In Appendix \ref{ex:reserve}, we provide a fairly general class of district
admissions rules that satisfies our assumptions in this result.

The analysis in this section characterizes conditions under which different policy goals are achieved under SPDA. One of the facts worth mentioning in this context is that achieving multiple policies can be overly demanding. To see this point, we note that individual rationality and $\alpha$-diversity policy are often incompatible with one another. For example, consider a  problem such that each student's most preferred school is her initial school and a constant $\alpha$ such that the initial matching
does not satisfy the $\alpha$-diversity policy. Indeed, in this case, no mechanism can simultaneously satisfy
individual rationality and the $\alpha$-diversity policy because the initial matching is the
unique individually rational matching, but it fails the $\alpha$-diversity policy.

\section{Achieving Policy Goals with Efficient Outcomes}\label{sec:TTC}
In this section, we turn our focus to efficiency. More specifically, we study the existence of a mechanism that satisfies a given policy goal on the distribution of agents, constrained efficiency, strategy-proofness, and individual rationality. We first consider a policy goal with type-specific ceilings at the district level. In this setting, we establish an impossibility result.

\begin{theorem}\label{thm:eff_impos}
There exist a problem and ceilings $(q_d^t)_{t\in \T, d\in \D}$ such that the the initial matching $\tilde{X}$ satisfies
the policy goal $\Xi \equiv \{\xi | q_d^t \geq \xi_d^t \text{ for all } d \text{ and } t, q_c \geq \sum_t \xi_c^t \text{ for all } c \}$, while there exists no mechanism that satisfies the policy goal $\Xi$, constrained efficiency, individual rationality, and strategy-proofness.
\end{theorem}

We show this result using the following example.

\begin{example}\label{ex:notconvex}
Consider the following problem with districts $d_1$ and $d_2$. District $d_1$ has schools
$c_1$, $c_2$, and $c_3$ and district $d_2$ has schools $c_4$, $c_5$, and $c_6$. All schools
have a capacity of one. There are six students: students $s_1$ and $s_4$ have type $t_1$,
students $s_2$ and $s_5$ have type $t_2$, and students $s_3$ and $s_6$ have type $t_3$.
Both districts have a ceiling of one for types $t_1$ and $t_2$: $q^{t_1}_{d_1}=q^{t_2}_{d_1}=1$
and $q^{t_1}_{d_2}=q^{t_2}_{d_2}=1$. Initially, student $s_i$ is matched with school $c_i$, for $i=1,\ldots,6$,
so the initial matching satisfies the policy goal $\Xi$.
Student preferences are as follows:

\[
\begin{tabular}
[c]{llllll}%
$\underline{s_{1}}$ & $\underline{s_{2}}$ & $\underline{s_{3}}$ &
$\underline{s_{4}}$ & $\underline{s_{5}}$ & $\underline{s_{6}}$\\
$c_6$ & $c_6$ & $c_5$ & $c_{3}$ & $c_3$ & $c_1$\\
$c_1$ & $c_2$ & $c_4$ & $c_4$ & $c_{5}$ & $c_2$\\
\vdots & \vdots & $c_3$ &  \vdots & \vdots  & $c_6$\\
          &            & \vdots &             &            & \vdots
\end{tabular}
\]
\smallskip

\noindent
where the dots in the table mean that the corresponding parts of the
preferences are arbitrary.

In this example, there are two matchings that satisfy the policy goal $\Xi$, constrained efficiency, and individual rationality:
\begin{align*}
X & =\{(  s_{1},c_6)  ,(s_{2},c_2)
,(s_{3},c_{4}),(s_{4},c_{3}),(s_{5}
,c_{5})  ,(s_{6},c_1)\}, \text{ and} \\
X'  & =\{(s_{1},c_{1}),(s_{2}
,c_6),(s_{3},c_5),(s_{4},c_4)
,(s_{5},c_3),(s_{6},c_2)\}.
\end{align*}

If a mechanism satisfies the desired properties, then its outcome at the above student preference profile must be either matching $X$ or $X'$.

Consider the case where the mechanism produces matching $X$ at the above student preference profile. Suppose student $s_3$ misreports her preference by ranking $c_5$ first and $c_3$ second (while the ranking of other schools is arbitrary). Under the new report, the mechanism produces matching $X'$ because it is the only matching that satisfies the policy goal $\Xi$, constrained-efficiency, and individual rationality. Since student $s_3$ strictly prefers her school in $X'$ to her school in $X$, she has a profitable deviation.

Similarly, consider the case where the mechanism produces matching $X'$ at the above student preference profile. Suppose student $s_6$ misreports her preference by ranking $c_1$ first and $c_6$ second (while the ranking of the other schools is arbitrary). In this case, the mechanism produces matching $X$ because it is the only matching that satisfies the policy goal $\Xi$, constrained-efficiency, and individual rationality. Since student $s_6$ strictly prefers her school in $X$ to her school in $X'$, she has a profitable deviation.

In both cases, there exists a student who benefits from misreporting, so the desired conclusion follows.
\qed
\end{example}

This example also shows that there is no mechanism that satisfies the $\alpha$-diversity policy goal for $\alpha=0$ introduced in Section \ref{sec:diversity}, constrained efficiency, individual rationality, and strategy-proofness. Consequently, without any assumptions, a policy goal may not be implemented with the desirable properties. To establish a positive result, we consider distributional policy goals that satisfy the following notion of discrete convexity, which is studied in the mathematics and operations research literatures \citep{Murota:SIAM:2003}.


\begin{definition}
Let $\chi_{c,t}$ denote the distribution where there is one type-$t$ student at school $c$ and there are no other students. A set of distributions $\Xi$ is \df{M-convex} if whenever $\xi,\tilde{\xi} \in \Xi$ and $\xi_c^t>\tilde{\xi}_c^t$ for some school $c$  and type $t$ then there exist school $c'$ and type $t'$ with $\xi_{c'}^{t'}<\tilde{\xi}_{c'}^{t'}$ such that $\xi-\chi_{c,t}+\chi_{c',t'}\in \Xi$ and $\tilde{\xi}+\chi_{c,t}-\chi_{c',t'} \in \Xi$.\footnote{The letter M in the term M-convex set comes from the word \emph{matroid}, a closely related and well-studied concept in discrete mathematics.}
\end{definition}

To illustrate this concept, suppose that a set of distributions $\Xi$ is M-convex. Consider two distributions $\xi$ and $\tilde{\xi}$ in this set such that there are more type-$t$ students in school $c$ at $\xi$ than at $\tilde{\xi}$. Then there exist school $c'$ and type $t'$ such that there are more type-$t'$ students in school $c'$ at $\tilde{\xi}$ than $\xi$ with the following two properties. First, removing one type-$t$ student from school $c$ and adding one type-$t'$ student to school $c'$ in $\xi$ produces a distribution in $\Xi$. Second, removing one type-$t'$ student from school $c'$ and adding one type-$t$ student to school $c$ in $\tilde{\xi}$ gives a distribution in $\Xi$ (see Figure \ref{fig:mconvex}). Intuitively, from each of these two distributions we can move closer to the other distribution in an incremental manner within $\Xi$, a property analogous to the standard convexity notion but adapted to a discrete setting.  We illustrate this concept with the following example.

\begin{figure}[htb]
\includegraphics[scale=.7]{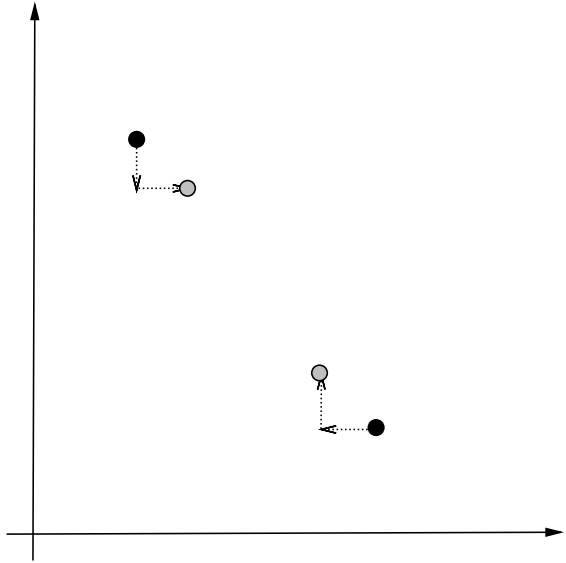}
    \caption{Illustration of M-convexity}
    \label{fig:mconvex}
\centering
\end{figure}

\begin{example}
Consider the problem and the set of distributions $\Xi$ 
defined in Example \ref{ex:notconvex}. We show that $\Xi$ is not M-convex. Recall matchings $X$ and $X'$ in that example. By construction, both $X$ and $X'$ satisfy the policy goal $\Xi$. Furthermore, $\xi_{c_3}^{t_1}(X)=1>0=\xi_{c_3}^{t_1}(X')$ because (i) school $c_3$ is matched with student $s_4$ at $X$, whose type is $t_1$, while (ii) school $c_3$ is matched with student $s_5$ at $X'$, whose type is $t_2 \neq t_1$. If the set of distributions $\Xi$ is M-convex, there exist a school $c$ and a type $t$ such that $\xi_c^t(X)<\xi_c^t(X')$ and $\xi(X)-\chi_{c_3,t_1}+\chi_{c,t}$ is in $\Xi$. Because each school's capacity is one, and at matching $X$ all schools have filled their capacities, this means that the only candidate for $(c,t)$ satisfying the above condition is such that $c=c_3$. But the only nonzero $\xi_{c_3}^t(X')$ is for $t=t_2$ (because $s_5$ is the unique student matched with $c_3$ at $X'$), and $\xi(X)-\chi_{c_3,t_1}+\chi_{c_3,t_2}$ does not satisfy the policy goal because district $d_1$'s ceiling for type $t_2$ is violated (note $\xi_{c_2}^{t_2}(X)=1$ because student $s_2$ is matched with $c_2$ at $X$.)

The above argument implies that $\Xi \cap \Xi^0$ is not M-convex either. To see this, note that both $\xi(X)$ and $\xi(X')$ are in  $\Xi \cap \Xi^0$ because all students are matched. Because we have shown that no distribution of the form $\xi(X)-\chi_{c_3,t_1}+\chi_{c,t}$ is in $\Xi$, by set inclusion relation $\Xi \cap \Xi^0 \subseteq \Xi$, there is no distribution of the form $\xi(X)-\chi_{c_3,t_1}+\chi_{c,t}$ in $\Xi \cap \Xi^0$ either.
\qed
\end{example}


Now we introduce a mechanism that achieves the desirable properties whenever the policy goal is M-convex. To do this, we first create a hypothetical matching problem. On one side of the market, there are school-type pairs $(c,t)$ where $c\in \C$ and $t\in \T$. On the other side, there are students from the original problem, $\S$. Given any student $s \in \S$ and a preference order $P_s$ of $s$ in the original problem, define preference order $\tilde P_s$ over school-type pairs in the hypothetical problem as follows: letting $t$ be the type of student $s$ and $c_0$ be her initial school in the original problem, $(s,c) \mathrel{P_s} (s,c') \iff (c,t) \mathrel{\tilde P_s} (c', t)$ for any $c,c' \in \C$, and $(c_0,t) \mathrel{\tilde P_s} (c,t')$ for any $c \in \C$ and
$t'\in \T$ such that $t' \neq t$. That is, $\tilde P_s$ is a preference order over school-type pairs that ranks the school-type pairs in which the type is $t$ in the same order as in $P_s$, while finding all school-type pairs specifying a different type as less preferred than the pair corresponding to her initial school. Furthermore, let $(c_0,t)$ be the initial school-type pair for $s$ in the hypothetical problem.

Next we define a priority ordering of students that school-type pairs use to rank students. For school-type pair $(c,t)$, students initially matched with $(c,t)$ have 
the highest priority, and then all other students have the second highest priority. This gives us two priority classes for students. Then, ties are broken according to a master priority list that every school-type pair uses.

We say that a type-$t$ student $s$ with  the initial school-type pair $(c,t)$ is \df{permissible} to school-type pair $(c',t')$ at matching $X$ if $\xi(X)+\chi_{c',t'}-\chi_{c,t}$ is in $\Xi$. Note that a type-$t$ student with initial school-type pair $(c,t)$ is always permissible to pair $(c,t)$ at matching $X$ whenever $\xi(X)$ is in $\Xi$.

The following is a generalization of Gale's top trading cycles mechanism \citep{shasca74}, building on its recent extension by \cite{suzuki17}.

\medskip

\paragraph{\textbf{Top Trading Cycles Algorithm}}
Consider a hypothetical problem.
\begin{description}
  \item[Step 1] Let $X^1 \equiv \tilde{X}$. Each school-type pair points to the permissible student at matching $X^1$ with the highest priority. If there exists no such student, remove the school-type pair from the market. Each student $s$ points to the highest ranked remaining school-type pair with respect to $\tilde{P}_s$. Identify and execute cycles. Any student who is part of an executed cycle is assigned the school-type pair she is pointing to and is removed from the market.
  \item[Step $\mathbf{n}$ ($\mathbf{n>1}$)] Let $X^n$ denote the matching consisting of assignments in the previous steps and initial assignments for all students who have not been processed in the previous steps. Each remaining school-type pair points to the unassigned student who is permissible at matching $X^n$ with the highest priority. If there exists no such student, remove the school-type pair from the market. Each unassigned student $s$ points to the highest ranked remaining school-type pair with respect to $\tilde{P}_s$. Identify and execute cycles. Any student who is part of an executed cycle is assigned  the school-type pair she is pointing to and is removed from the market.
\end{description}

This algorithm terminates in the first step such that no student remains to be
processed. The outcome is defined as the matching induced by the outcome of the
hypothetical problem at this step. The top trading cycles mechanism (TTC) takes a profile of student preferences
as input and produces the outcome of this algorithm at the reported student preference profile. Note that the
definition of permissibility and, hence, the definition of TTC, depend on the policy goal. Nevertheless, we do
not explicitly state the policy goal under consideration when it is clear from the context.

The main result of this section is as follows.
\begin{theorem}\label{thm:ttc}
Suppose that the initial matching satisfies the policy goal $\Xi$. If $\Xi \cap \Xi^0$ is M-convex, then TTC satisfies the policy goal $\Xi$, constrained efficiency, individual rationality, and strategy-proofness.
\end{theorem}


The assumption that the initial matching satisfies the policy goal is necessary for the result:
Consider student preferences such that each student's highest-ranked school is her initial
school. Then the initial matching is the unique individually rational matching.
Therefore, if there exists a mechanism with the desired properties, then the outcome at this
preference profile has to be the initial matching. Hence, we need the assumption that the
initial matching satisfies the policy goal to have such a mechanism.

To see one of the implications of this theorem, suppose that the policy goal $\Xi$ is such
that no school is matched with more students than its capacity. In that case, if $\Xi$ is M-convex,
then TTC satisfies the desirable properties.

\begin{corollary}\label{corollary:TTC}
Suppose that the policy goal $\Xi$ is such that for every $\xi \in \Xi$ and $c\in \C$, $\sum_{t} \xi^t_c \leq q_c$.  Furthermore, suppose that the initial matching satisfies $\Xi$. If $\Xi$ is M-convex, then TTC satisfies the policy goal $\Xi$, constrained efficiency, individual rationality, and strategy-proofness.
\end{corollary}

In the proof of this corollary, we show that when $\Xi$ is M-convex and no
distribution in $\Xi$ assigns more students to a school than its capacity,
then $\Xi \cap \Xi^0$ is also M-convex. Therefore, the corollary follows
directly from Theorem \ref{thm:ttc}.

Next we illustrate TTC with an example.

\begin{example}\label{ex:TTC}
Consider a problem with two school districts, $d_1$ and $d_2$. District $d_1$ has
school $c_{1}$ with capacity three and school $c_{2}$ with capacity two. District $d_2$ has school $c_{3}$
with capacity two and school $c_{4}$ with capacity one. There are seven
students: students $s_{1}$, $s_{2}$, $s_{3}$, and $s_{4}$ are from district $
d_{1}$ and have type $t_{1}$, whereas students $s_{5}$, $s_{6}$, and $s_{7}$
are from district $d_{2}$ and have type $t_{2}$. The initial matching is $\{(s_1,c_1),(s_2,c_1),(s_3,c_2),(s_4,c_2),(s_5,c_3),(s_6,c_3),(s_7,c_4)\}$.
Student preferences are as follows.

\[
\begin{tabular}
[c]{lllllll}
$\underline{P_{s_{1}}}$ & $\underline{P_{s_{2}}}$ & $\underline{P_{s_{3}}}$ &
$\underline{P_{s_{4}}}$ & $\underline{P_{s_{5}}}$ & $\underline{P_{s_{6}}}$ &
$\underline{P_{s_{7}}}$\\
$c_{2}$ & $c_{3}$ & $c_{4}$ & $c_{2}$ & $c_{1}$ & $c_{4}$ & $c_{2}$\\
$c_3 $ & $c_{1}$ & $c_{2}$ & $c_{3}$ & $c_{2}$ & $c_1$ & $c_{3}$\\
 $\vdots$ & $\vdots$ & $\vdots$ & $c_{1}$ & $c_{3}$ & $c_3$ & $c_{1}$\\
&  &  & $c_4$ & $c_4$ & $c_2$ & $c_4$%
\end{tabular}
\]
\smallskip


In addition to the school capacities, there is only one additional constraint that school $c_1$ cannot have more than one type-$t_2$ student. As we show in the proof of Corollary \ref{corollary:convexdiv}, the set of distributions that satisfy this policy goal and the requirement that every student is matched is an M-convex set. Therefore, by Theorem \ref{thm:ttc}, TTC satisfies constrained efficiency, individual rationality, strategy-proofness, and the policy goal.

To run TTC, we use a master priority list. Suppose that the master priority list ranks students as follows: $s_1 \succ s_2 \succ s_3 \succ s_4 \succ s_5 \succ s_6 \succ s_7$.

At Step $1$, there are eight school-type pairs. Consider $(c_1,t_1)$. Initially, students $s_1$ and $s_2$ are matched with it, so they are both permissible to this pair. We use the master priority list to rank them, so $s_1$ gets the highest priority at $(c_1,t_1)$. Therefore, $(c_1,t_1)$ points to $s_1$. Now consider $(c_1,t_2)$. Initially, it does not have any students because there is no type-$t_2$ student assigned to $c_1$ in the original problem. Furthermore, $s_1$ is permissible to $(c_1,t_2)$ because she can be removed from $(c_1,t_1)$ and a type-$t_2$ student can be assigned to $(c_1,t_2)$ without violating the school capacities or the policy goal. Therefore, $(c_1,t_2)$ points to $s_1$ as well, who gets a higher priority than the other permissible students because of the master priority list. The rest of the pairs also point to the highest-priority permissible students. Each student points to the highest ranked school-type pair of the same type as shown in Figure \ref{fig:TTC}A. There is only one cycle: $s_7 \rightarrow (c_2,t_2) \rightarrow s_3 \rightarrow (c_4,t_1) \rightarrow s_7$. Therefore, $s_7$ is matched with $(c_2,t_2)$ and $s_3$ is matched with $(c_4,t_1)$.

\begin{figure}[htb]
\centering
  \subfloat[Step 1 of TTC]{
    \includegraphics[width=.45\textwidth]{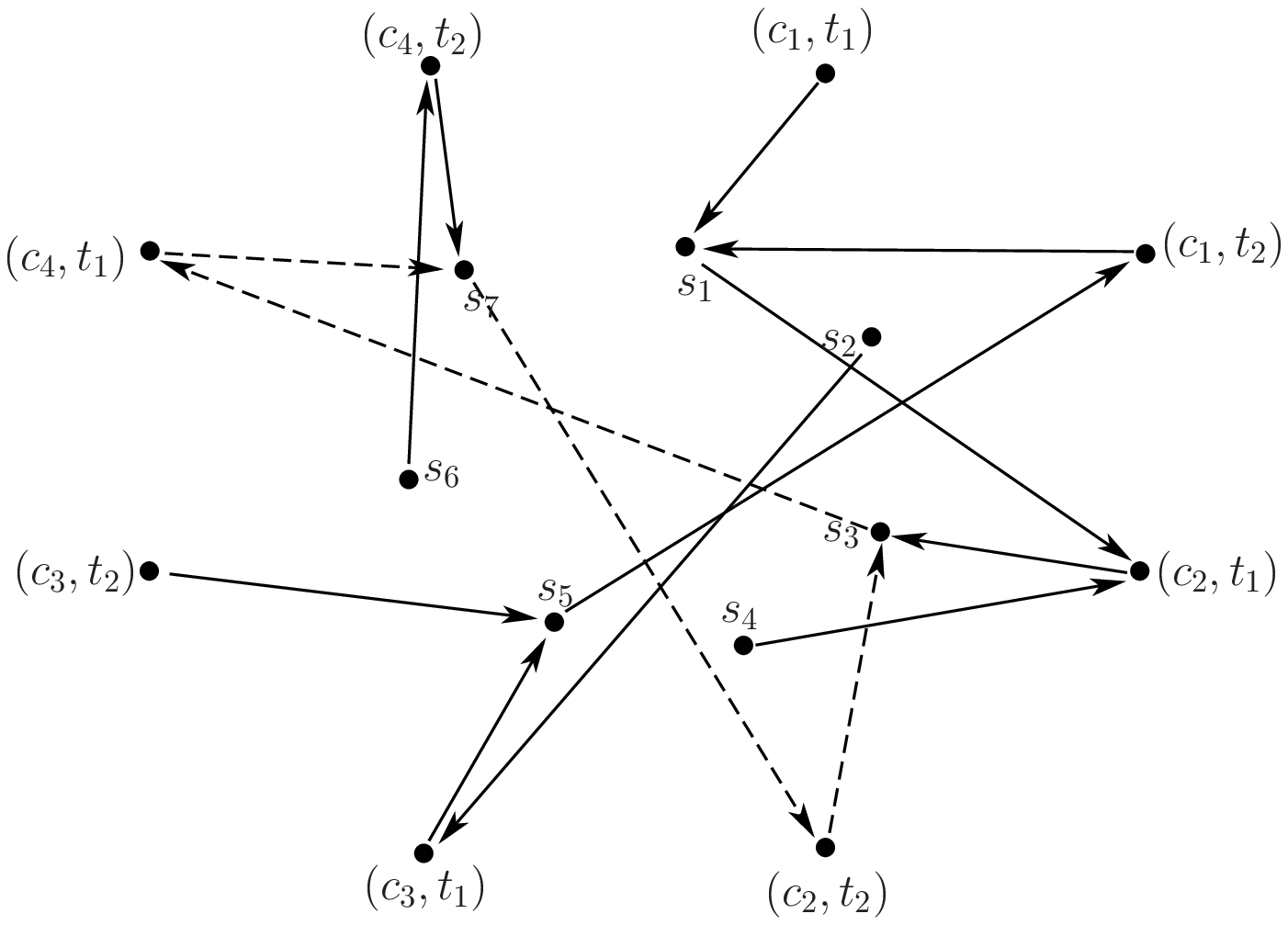}}\hfill
  \subfloat[Step 2 of TTC]{
    \includegraphics[width=.45\textwidth]{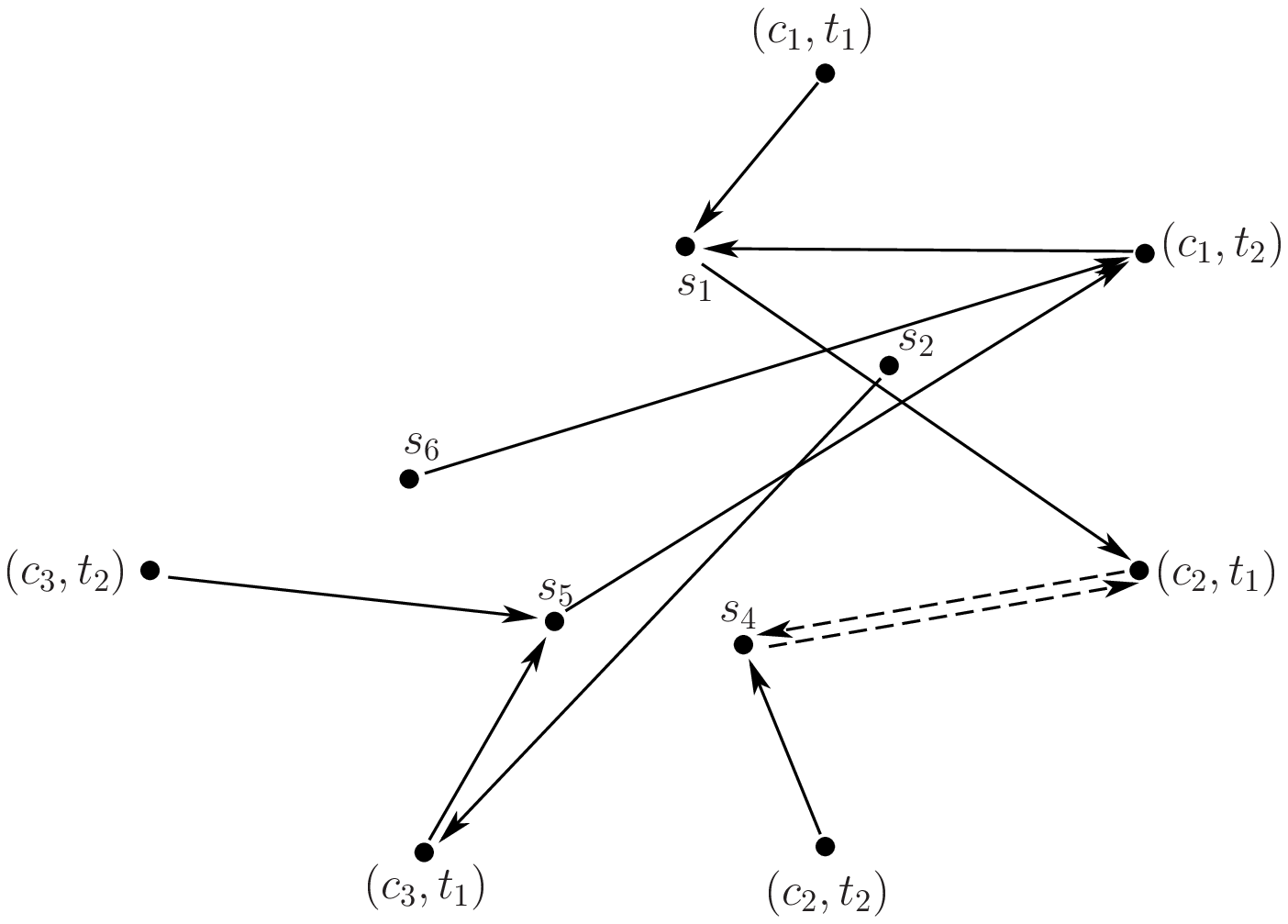}}\\
  \subfloat[Step 3 of TTC]{
    \includegraphics[width=.45\textwidth]{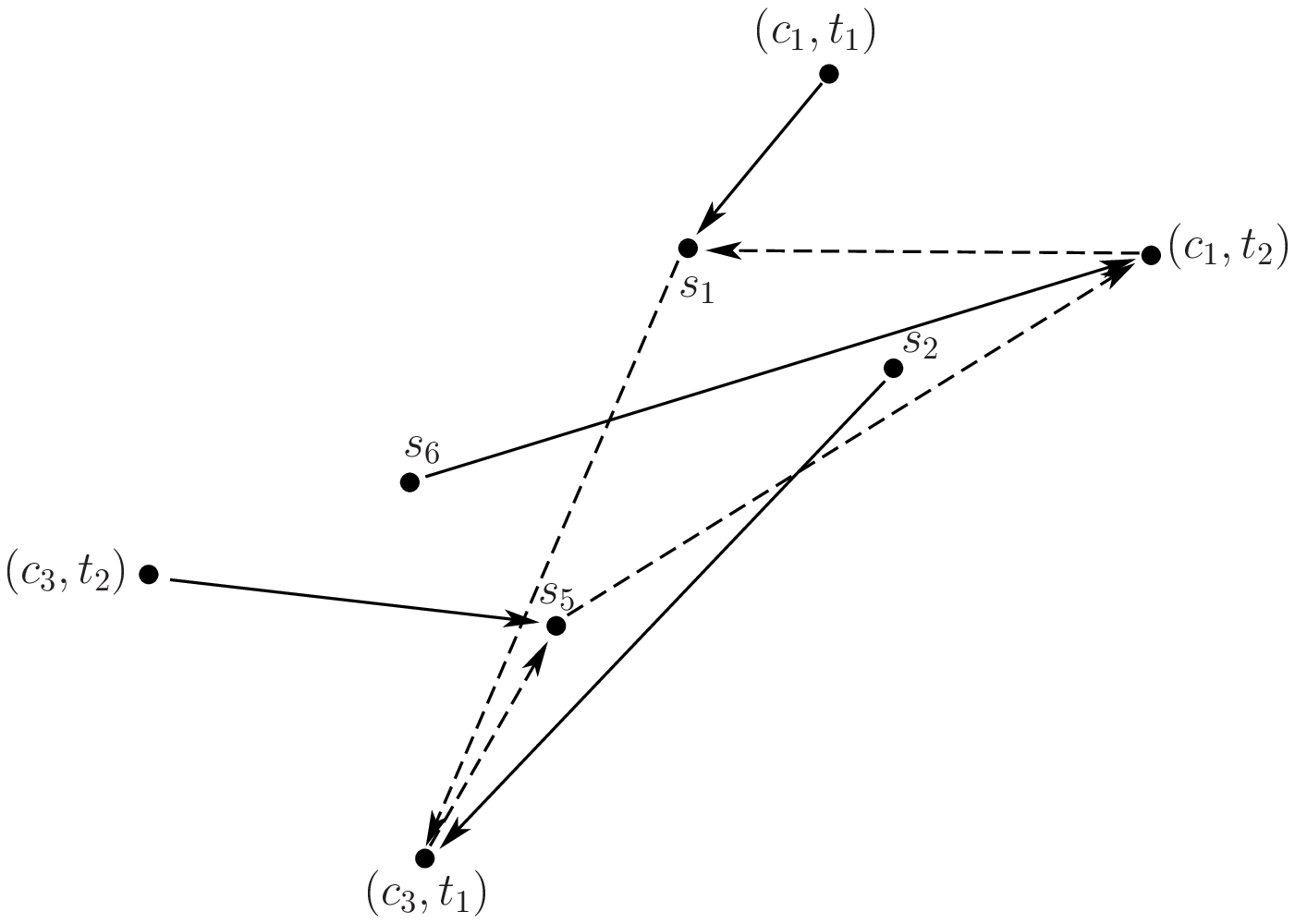}}\hfill
  \subfloat[Step 4 of TTC]{
    \includegraphics[width=.35\textwidth]{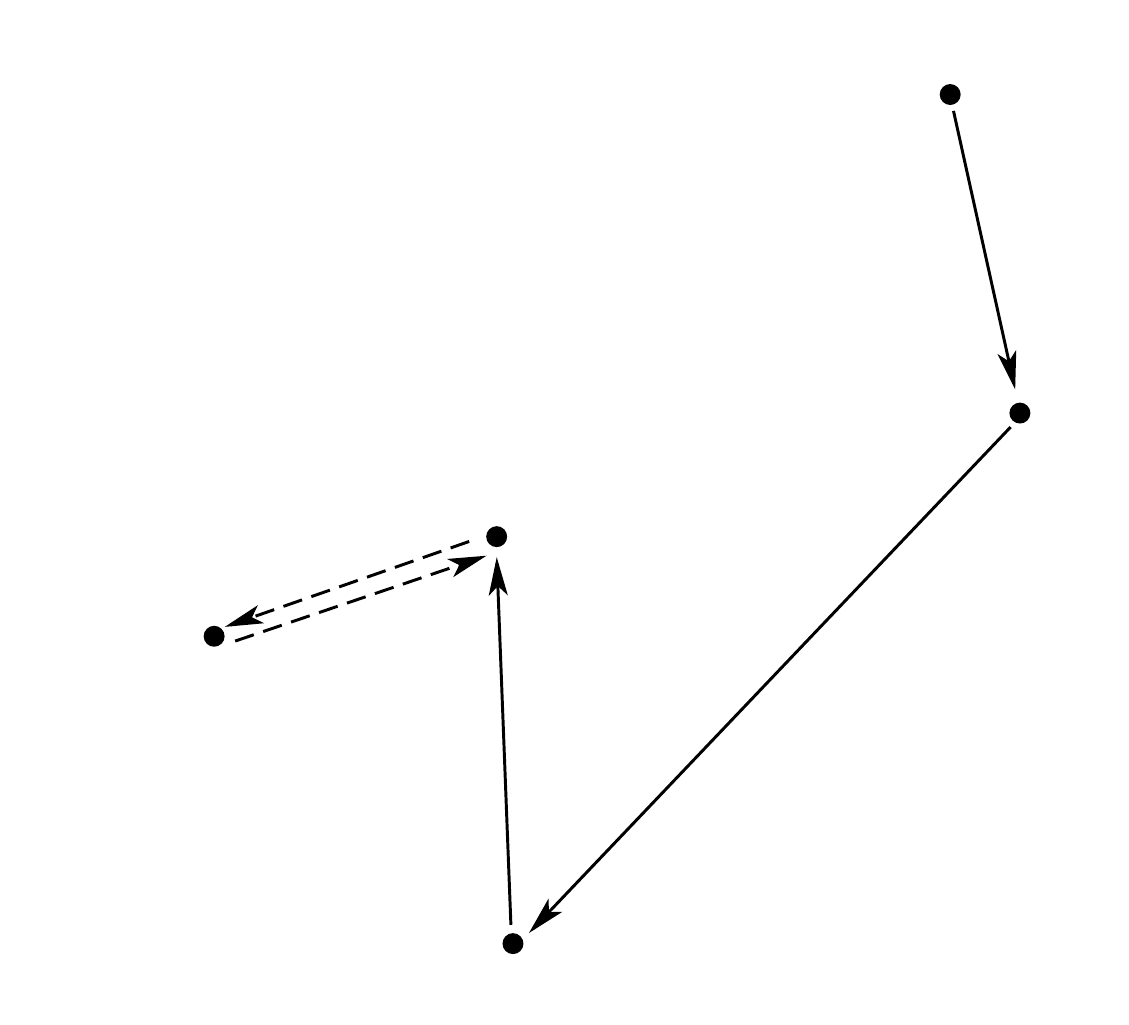}}
  \caption{The first four steps of TTC. In each step, there is only one cycle, which is represented by the dashed lines.}
  \label{fig:TTC}
\end{figure}

At Step $2$, there are six remaining school-type pairs: There are no permissible students for $(c_4,t_1)$ and $(c_4,t_2)$ because $c_4$ has a capacity of one and it is already assigned to $s_3$. Each remaining school-type pair points to the highest-ranked remaining permissible student. Each student points to the highest-ranked remaining school-type pair (see Figure \ref{fig:TTC}B). There is only one cycle: $s_4 \rightarrow (c_2,t_1) \rightarrow s_4$. Hence, $s_4$ is assigned to $(c_2,t_1)$.

The algorithm ends in five steps. Steps 3 and 4 are also shown in Figure \ref{fig:TTC}. In  Step 5, $s_2$ points to $(c_1,t_1)$, which points back to the student. The outcome is \[
\{(s_1,c_3),(s_2,c_1),(s_{3},c_{4}),(s_{4},c_{2}),(s_{5},c_{1}),(s_{6},c_{3}),(s_{7},c_{2})\}.
\]

It can be easily seen that the distribution associated with this matching satisfies the policy goal because no school has more students than its capacity and $c_1$ has only one type-$t_2$ student.
\qed
\end{example}

Sometimes it may be more convenient to describe a policy goal using a real-valued function rather
than a set of distributions. The interpretation is that the policy function measures how satisfactory the distribution
is in terms of the policy goal. To formalize this alternative approach let $f: \mathbb Z_+^{|\C|\times |\T|}
\rightarrow \mathbb{R}$ be a function on distributions such that $f(\xi) \geq f(\xi')$ means
that distribution $\xi$ satisfies the policy at least as well as distribution $\xi'$. Let $\lambda \in \mathbb{R}$ be
a constant. Consider the following $\boldsymbol{(f,\lambda)-}$\textbf{policy}:
$\Xi(f,\lambda) \equiv \{\xi \in  \mathbb Z_+^{|\C|\times |\T|} | f(\xi)\geq \lambda\}$.
Note that the initial matching $\tilde{X}$ satisfies the $(f,\lambda)$-policy if, and only if, $f(\xi(\tilde{X})) \geq \lambda$.

We introduce the following condition on functions, which plays a crucial role in the M-convexity of the
$(f,\lambda)$-policy.

\begin{definition}
A function $f$ is \textbf{pseudo M-concave}, if for every distinct $\xi,\tilde{\xi} \in \Xi_0$, there exist $(c,t)$
and $(c',t')$ with $\xi_c^t>\tilde{\xi}_c^t$ and $\xi_{c'}^{t'}<\tilde{\xi}_{c'}^{t'}$ such that
  \begin{equation*}
  \min \{f(\xi-\chi_{c,t}+\chi_{c',t'}), f(\tilde{\xi}+\chi_{c,t}-\chi_{c',t'})\} \geq \min \{f(\xi),f(\tilde{\xi})\}.
  \end{equation*}
\end{definition}

This is a notion of concavity for functions on a discrete domain. Lemma \ref{lem:f-diversity} shows
that pseudo M-concavity characterizes when upper contour sets are M-convex. It is stronger
than quasi M-concavity but not logically related to the M-concavity studied in the discrete
mathematics literature \citep{Murota:SIAM:2003}.

\begin{lemma}\label{lem:f-diversity}
$\Xi(f,\lambda) \cap \Xi^0$ is M-convex for every $\lambda$ if, and only if, $f$ is pseudo M-concave.
\end{lemma}

Therefore, we get the following result:

\begin{theorem}\label{thm:ttcf}
If $f$ is pseudo M-concave and $\lambda$ is such that $f(\xi(\tilde{X}))\geq \lambda$, then TTC satisfies
the $(f,\lambda)$-policy, constrained efficiency, individual rationality, and strategy-proofness.
\end{theorem}

To see why this theorem holds, recall that by Lemma \ref{lem:f-diversity}, $\Xi(f,\lambda) \cap \Xi^0$ is M-convex.
Furthermore, by assumption, the initial matching
satisfies the $(f,\lambda)$-policy. Therefore, the result follows from Theorem \ref{thm:ttc}.

Before we consider specific policy goals, we show that the set-based approach in Theorem \ref{thm:ttc} and the
function-based approach in Theorem \ref{thm:ttcf} are equivalent. For that purpose, note first that Lemma 1 already shows that the $(f,\lambda)$-policy for a pseudo M-concave function $f$ yields an M-convex policy set $\Xi(f,\lambda)\cap \Xi^0$. We establish a sense in which a converse result holds. 

\begin{theorem}\label{thm:ttcequ}
Suppose that $\Xi$ is a set of distributions. If $\Xi \cap \Xi^0$ is M-convex, then there exist
a pseudo M-concave function $f$ and a constant $\lambda \in \mathbb{R}$ such that $\Xi(f,\lambda) \cap \Xi^0 = \Xi \cap \Xi^0$.
\end{theorem}

Now that we have established general results based on M-convexity of the policy set or pseudo M-concavity of the policy function, we proceed to apply them to a variety of situations. To begin, consider the set $\Xi$ of distributions of all feasible matchings. In other words, consider a situation in which no policy goal is imposed other than $\sum_t \xi^t_c \le q_c$ for each $c$. Then it is rather straightforward to show that the set $\Xi\cap \Xi^0$ is an M-convex set. This implies that when there is no policy goal, TTC is  efficient, individually rational, and strategy-proof, a standard result in the literature \citep{abdulson03}.

Next we apply Theorems \ref{thm:ttc} and \ref{thm:ttcf} to a variety of policy goals.
These results turn out to be applicable to many specific cases, as  a wide variety of policy goals induce distributions that satisfy M-convexity or can be expressed by policy functions that are pseudo M-concave. To be more specific, first suppose that the policy goal $\Xi$ sets type-specific floors and ceilings at each school, i.e., $\Xi \equiv \{\xi | q_c^t \geq \xi_c^t \geq p_c^t \text{ for all } c \text{ and }t  \}$ where $q_c^t$ is the ceiling and $p_c^t$ is the floor for type $t$ at school $c$. Therefore, for each school, the number of students of a given type must be within the ceiling and floor of this type at the school. We call a policy goal $\Xi$ of this form a  \df{school-level diversity policy} and show that $\Xi \cap \Xi^0$ is an M-convex set. This finding, together with Theorem \ref{thm:ttc}, implies the following positive result.

\begin{corollary}\label{corollary:convexdiv}
Suppose that the initial matching satisfies a school-level diversity policy. Then
TTC satisfies the school-level diversity policy, constrained efficiency, individual rationality, and strategy-proofness.
\end{corollary}

We note a sharp contrast between this result and Theorem \ref{thm:eff_impos}. The latter result demonstrates that no mechanism is guaranteed to satisfy the policy goal and other desiderata such as constrained efficiency, individual rationality, and strategy-proofness if the floors or ceilings are imposed at the district level. Corollary \ref{corollary:convexdiv}, in contrast, shows that a mechanism
with the desirable properties exists if the floors and ceilings are imposed at the school level. Taken together, these results
inform policy makers about what kinds of diversity policies are compatible with the other desiderata.

One possible shortcoming of Corollary \ref{corollary:convexdiv} is that the result holds under the assumption that the
initial matching satisfies the school-level diversity policy. This may be undesirable given that often diversity policies are implemented because schools or districts
are regarded as insufficiently diverse, as in the case of the diversity law in Minnesota. In such a setting, a potential
diversity requirement can be that the diversity should not decrease as a result of interdistrict school choice according to a diversity measure $f$. Such a consideration can be formally described as the $(f,\xi(\tilde{X}))$-policy, $\Xi(f,\xi(\tilde{X}))$.
The next corollary establishes a positive result for a $\Xi(f,\xi(\tilde{X}))$-policy where the diversity is measured via the ``Manhattan distance'' to an ideal point.

\begin{corollary}\label{cor:fdiverse}
Let $\hat{\xi} \in \Xi_0$ be an ideal distribution and $f(\xi) \equiv -\sum_{c,t}|\xi_c^t-\hat{\xi}_c^t|$ be the policy function.
Then TTC satisfies $(f,\xi(\tilde{X}))$-policy, constrained efficiency, individual rationality, and strategy-proofness.
\end{corollary}

Note that the initial matching $\tilde{X}$ always satisfies $(f,\xi(\tilde{X}))$-policy. Furthermore, we show that
the policy function $f$ is pseudo M-concave. Therefore, this corollary follows from Theorem \ref{thm:ttcf}. More generally,
when the diversity is measured by a pseudo M-concave function, then the TTC outcome is as diverse as the initial matching.
Furthermore, TTC also satisfies the other desirable properties.

Next, we study the balanced-exchange policy introduced in Section \ref{sec:bal}. We establish that the balanced-exchange policy imposed on $\Xi^0$
  is represented by a distribution that satisfies M-convexity. This implies the following result.

\begin{corollary}\label{corollary:convexbal}
TTC satisfies the balanced-exchange policy, constrained efficiency, individual rationality, and strategy-proofness.
\end{corollary}

One of the advantages of our approach is that M-convexity of a set and pseudo M-concavity of a function are so general that a wide variety of policy goals satisfy them, and that it is likely to be applicable for policy goals that one may encounter in the future. To highlight this point, we consider imposing  the diversity and balanced-exchange policies at the same time. More specifically, define a set of distributions $\Xi \equiv \{\xi | q_c^t \geq \xi_c^t \geq p_c^t \text{ for all } c \text{ and } t \text{ and } \sum_t \sum_{c:d(c)=d} \xi_c^t =k_d \text{ for all } d\}$ and call it the \df{combination of balanced-exchange and school-level diversity policies}. This is the set of distributions that satisfy both the school-level floors and ceilings and the balanced-exchange requirement. We
establish $\Xi \cap \Xi^0$ is M-convex, implying the following result.

\begin{corollary}\label{corollary:mix}
Suppose that the initial matching satisfies the combination of balanced exchange and school-level diversity policies. Then
TTC satisfies the combination of balanced exchange and school-level diversity policies, constrained efficiency, individual rationality, and strategy-proofness.
\end{corollary}

 In general, the
intersection of two M-convex sets need not be M-convex.\footnote{Such
an example is available from the authors upon request.} Therefore, the proof of this result
does not follow from the proofs of Corollaries \ref{corollary:convexdiv} and \ref{corollary:convexbal}.

\section{Conclusion}\label{sec:conclusion}
Despite increasing interest in interdistrict school choice in the US, the scope of matching theory  has been limited to intradistrict choice. In this paper, we proposed a new framework to study interdistrict school choice that allows for interdistrict admissions, both from stability and efficiency perspectives. For stable mechanisms, we characterized
 conditions on district admissions rules that achieve  a variety of important policy goals, such as student diversity across districts. For efficient mechanisms, we showed that certain types of diversity policies are incompatible with desirable properties such as strategy-proofness, while alternative forms of diversity policies can be achieved by  a variation of the top trading cycles mechanism, which is strategy-proof. Overall, our analysis suggests that interdistrict school choice can help achieve desirable policy goals such as student diversity, but only with an appropriate design of constraints, admissions rules, and placement mechanisms.

We regard this paper as a first step toward formal analysis of interdistrict school choice based on tools of  market design. As such, we envision a variety of directions for future research. For example, it may be interesting to study cases in which the conditions for our results are violated. Although we already know the policy goals are not guaranteed to be satisfied for our stability results (our results provide necessary and sufficient conditions), the seriousness of the failure of the policy goals studied in the present paper is an open question. Quantitative measures or an approximation argument like those used in ``large matching market'' studies  (e.g., \citet{roth99}, \citet{kojpat09}, \citet{kojima2013matching}, and \citet{ashlagi2014stability}) may prove useful, although this is speculative at this point and beyond the scope of the present paper.

We studied policy goals that we regarded as among the most important ones, but they are far from being exhaustive. Other important policy goals may include a diversity policy requiring certain proportions of different student types in each district (see \citet{ngvoh17} for a related policy at the level of schools), as well as a balanced exchange policy requiring a certain bound on the difference in the numbers of students received from and sent to other districts (see \citet{dur2015two} for a related policy at the level of schools). 
Given that the existing literature has not studied interdistrict school choice, we envision that many policy goals await to be studied within our framework.

While our paper is primarily theoretical and aimed at proposing a general framework to study interdistrict school choice, the main motivation comes from applications to actual programs such as Minnesota's AI program. Given this motivation, it would be interesting to study interdistrict school choice empirically. For instance, evaluating how well the existing programs are doing in terms of balanced exchange, student welfare, and diversity, and how much improvement could be made by a conscious design based on theories such as the ones suggested in the present paper, are important questions left for future work. In addition, implementation of our designs in practice would be interesting. Doing so may, for instance, shed new light on the tradeoff between SPDA and TTC, which has been studied in the intradistrict school choice from a practical perspective (e.g., \citet{abpatrotson06}, \citet{abchepatroter:17}). We are only beginning to learn about the interdistrict school choice problem, and thus we expect that these and other questions could be answered as more researchers analyze it.


\bibliographystyle{aer}
\bibliography{matching}


\appendix

\section{Additional Results}
In this section, we provide some additional results. Before we proceed, we introduce two admissions rule properties. An admissions rule $Ch$ satisfies \df{path independence} if for every $X,Y\subseteq \X$, $Ch(X\cup Y)=Ch(X \cup Ch(Y))$. Path independence states that a set can be divided into not-necessarily disjoint subsets and the admissions rule can be applied to the subsets in any order so that the chosen set of contracts is always the same. An admissions rule $Ch$ satisfies \df{the irrelevance of rejected contracts} (IRC) if for every $X \subseteq \X$ and $x\notin Ch(X)$, $Ch(X\setminus \{x\})=Ch(X)$. The irrelevance or rejected contracts states that a rejected contract can be removed
from a set without changing the chosen set. Path independence is equivalent to substitutability and IRC \citep{aizmal81}. Furthermore, an admissions rule satisfies substitutability and LAD if, and only if, it satisfies path independence and LAD.\footnote{See \cite{aygson12a} for a study of IRC and \citet{chayen13} for a study of path independence in a matching context.}

\subsection{Improving Student Welfare for  Districts with Intradistct School Choice}\label{sec:cen}
 In Section \ref{sec:ir}, we studied when SPDA satisfies individual rationality, which requires that,
under interdistrict school choice, every student is matched with a school that is weakly more preferred
than her initial school. In this section, we consider an alternative setting where each district uses
SPDA to assign its students to schools when there is no interdistrict school choice. In other words,
the status quo is SPDA when there is only intradistrict school choice. More explicitly,
each student ranks schools in their home districts (or contracts associated with their home districts) and
SPDA is used between a district and students from that district. Note that we assume each student's
ranking over contracts associated with the home district is the same as the relative ranking in the
original preferences. Importantly, in this setting, we
compare SPDA outcomes in interdistrict and intradistrict school choice. In such a setting, we characterize district
admissions rules which guarantee that no student is hurt from interdistrict school choice.

The next property of district admissions rules plays a crucial role to achieve this policy.


\begin{definition}
A district admissions rule $\chd$ \df{favors own students} if for any matching $X$ that is feasible for students,
\begin{center}
$\chd(X)\supseteq \chd(\{x\in X|d(s(x))=d\})$.
\end{center}
\end{definition}

When a district admissions rule favors own students, any contract that is chosen from a set of contracts associated with students from a district is also chosen from a superset that includes additional contracts associated with students from the other districts. Roughly, this condition requires that, under interdistrict school choice, a district prioritizes its own students that it used to admit over students from the other districts (even though an out-of-district student can still be admitted when a student from the district is rejected).

The following result shows that this is exactly the condition which guarantees that interdistrict school choice weakly improves
the outcome for every student.

\begin{theorem}\label{thm:welfimprove}
Every student weakly prefers the SPDA outcome under interdistrict school choice to the
SPDA outcome under intradistrict school choice for all student preferences if, and only if, each district's admissions rule
favors own students.
\end{theorem}

In the proof, we show that in the intradistrict school choice the SPDA outcome can alternatively be produced by an
interdistrict school choice model where students rank contracts with all districts and districts have modified admissions
rules: For any set of contracts $X$, each district $d$ chooses the following contracts: $\chd(\{x\in X|d(s(x))=d\})$.
Since the original district admissions rules favor own students, the chosen set under the modified admissions rule is a subset of $\chd(X)$ when $X$ is feasible for students. Then the conclusion that students receive weakly more preferred outcomes in interdistrict school choice than in intradistrict school choice follows from a comparative statics property of SPDA that we show (Lemma \ref{lem:comp}).\footnote{We cannot use the comparative statics result of \cite{yen14} because in our setting $Ch_d(X)\supseteq Ch'_d(X)$ only when $X$ is feasible for students, whereas \cite{yen14} requires this property for all sets of contracts $X$.} To show
the ``only if'' part, when there exists a district admissions rule that fails to favor own students, we construct
student preferences such that interdistrict school choice makes at least one student strictly worse off than
intradistrict school choice.

\subsection{District-level Type-specific Ceilings}\label{sec:districtlevel}
In this section, we show the incompatibility of type-specific ceilings at the district level
with the existence of a stable matching.

\begin{definition}
A district admissions rule $Ch_{d}$ has a \textbf{\textit{district-level type-specific
ceiling}} of $q_{d}^{t}$ for type-$t$ students if the number
of type-$t$ students admitted from a matching that is feasible for students
cannot exceed this ceiling. More formally, for
any matching $X$ that is feasible for students,
\begin{center}
$|\{x\in Ch_d(X)| \tau(s(x))=t\}| \leq q^t_d$.
\end{center}
\end{definition}

Note that, as in the case of school-level type-specific ceilings, district admissions rules do not necessarily satisfy acceptance
once district-level type-specific ceilings are imposed. We define a weaker version of the acceptance
assumption as follows.

\begin{definition}
A district admissions rule $Ch_{d}$ that has district-level type-specific ceilings
is \textbf{d-weakly acceptant} if, for any contract $x$ associated with a
type-$t$ student and district $d$ and matching $X$ that is feasible for students,
if $x$ is rejected from $X$, then at $Ch_{d}(X)$,

\begin{itemize}
\item the number of students assigned to school $c(x)$ is equal to $q_{c(x)}$, or

\item the number of students assigned to district $d$ is at least $k_{d}$, or

\item the number of type-$t$ students assigned to district $d$ is at least $q_d^t$.
\end{itemize}
\end{definition}

This admissions rule property states that a student can be rejected only when one
of these three conditions is satisfied.

We establish that in an interdistrict school choice problem in which district admissions rules have
district-level type-specific ceilings that also satisfy some other desired properties,
there may exist no stable matching.

\begin{theorem}\label{thm:imposswithschools}
There exist districts, schools, students, and their types such that for every
admissions rule of a district with district-level type-specific ceilings
that satisfies $d$-weak acceptance and IRC, there
exist  admissions rules for the other districts
that satisfy substitutability and IRC and
student preferences such that no stable matching exists.
\end{theorem}

To show this result, we construct an environment such that a district
admissions rule with the desired properties cannot satisfy \emph{weak substitutability},
a necessary condition to guarantee the existence of a stable matching \citep{hatfield2008matching}.

\section{Examples of District Admissions Rules}\label{sec:examples}
In this section, we first provide a class of district admissions rules that are feasible and acceptant and, furthermore, have completions that satisfy substitutability and LAD. Then, based on this class, we identify admissions rules that also satisfy the properties stated in Theorems \ref{thm:welfimprove2}, \ref{thm:balance}, \ref{thm:diversity}, and \ref{thm:welfimprove}.

\subsection{An Example of a District Admissions Rule}\label{section:genex}
Consider a district $d$ with schools $c_1,\ldots,c_n$. Each school $c_i$ has an admissions rule $Ch_{c_i}$ such that, for any set of contracts $X$, $Ch_{c_i}(X)=Ch_{c_i}(X_{c_i})\subseteq X_{c_i}$. District $d$'s admissions rule $Ch_d$ is defined as follows. For any set of contracts $X$,
\begin{center}
$Ch_d(X)=Ch_{c_1}(X) \cup Ch_{c_2}(X\setminus Y_1) \cup \ldots \cup Ch_{c_n} (X\setminus Y_{n-1})$,
\end{center}
where $Y_i$ for $i=1,\ldots,n-1$ is the set of all contracts in $X$  associated with students who have contracts in $Ch_{c_1}(X)\cup \ldots \cup Ch_{c_i}(X\setminus Y_{i-1})$. In words, we order the schools and let schools choose in that order. Furthermore, if a student is chosen by some school, we remove all contracts associated with this student for the remaining schools.

We study when district admissions rule $\chd$ satisfies our assumptions.

\begin{claim}\label{claim:feasible}
Suppose that for every school $c_i$ and matching $X$, $\abs{Ch_{c_i}(X)}\leq q_{c_i}$. Then district admissions rule $Ch_d$ is feasible.
\end{claim}

\begin{proof}
Since every student-school pair uniquely defines a contract, for every matching $X$, every school $c_i$, and every student $s$, there is at most one contract  associated with $s$ in $Ch_{c_i}(X)$. In addition, whenever a student's contract with a school $c_i$ is chosen, her contracts with the remaining schools are included in $Y_j$ for every $j \ge i$ by the construction of $\chd$. Hence, for every $X$, $\chd(X)$ is feasible for students.
Furthermore, by assumption, $\abs{Ch_{c_i}(X)}\leq q_{c_i}$ for each $c_i$. Therefore, $\chd$ is feasible.
\end{proof}

\begin{claim}\label{claim:acceptant}
Suppose that for every school $c_i$ and matching $X$, $\abs{Ch_{c_i}(X)}=\min\{q_{c_i},\abs{X_{c_i}}\}$. Then district admissions rule $Ch_d$ is acceptant.
\end{claim}

\begin{proof}
Suppose that matching $X$ is feasible for students and $x\in X_d \setminus \chd(X)$. There exists $i\leq n$ such that $c_i = c(x)$. Since $X$ is feasible for students, $x\in X\setminus Y_{i-1}$ where $Y_{i-1}$ is as defined in the construction of $\chd$. Because $x\in X_d \setminus \chd(X)$, $x\notin Ch_{c_i}(X\setminus Y_{i-1})$. Then $\abs{Ch_{c_i}(X\setminus Y_{i-1})}=q_{c_i}$ by assumption, which implies that district admissions rule $\chd$ is acceptant.
\end{proof}

Next we study when district admissions rule $Ch_d$ has a completion that satisfies substitutability and LAD. Consider the following district admissions rule $Ch'_d$: For any set of contracts $X$,
\begin{center}
$Ch'_d(X)= Ch_{c_1}(X) \cup \ldots \cup Ch_{c_n}(X)$.
\end{center}

\begin{claim}\label{claim:completion}
Suppose that for every school $c_i$, $Ch_{c_i}$ satisfies substitutability and LAD. Then district admissions rule $Ch'_d$ is a completion of $Ch_d$, and it satisfies substitutability and LAD.
\end{claim}

\begin{proof}
To show that $Ch'_d$ is a completion of $Ch_d$, suppose that $X$ is a set of contracts such that $Ch'_d(X)$ is feasible for students. By mathematical induction, we show that $Ch_{c_i}(X)=Ch_{c_i}(X\setminus Y_{i-1})$ for $i=1,\ldots,n$, where $Y_i$ is defined as above for $i>1$ and $Y_0=\emptyset$. The claim trivially holds for $i=1$. Suppose that it also holds for $1,\ldots,i-1$. We show the claim for $i$. Since $Ch'_d(X)$ is feasible for students, $Ch_{c_i}(X)$ and $Ch_{c_1}(X)\cup \ldots \cup Ch_{c_{i-1}}(X)$ do not have any contracts associated with the same student. Therefore, $Ch_{c_i}(X) \cap Y_{i-1} = \emptyset$. Since $Ch_{c_i}$ satisfies IRC, $Ch_{c_i}(X)=Ch_{c_i}(X\setminus Y_{i-1})$. As a result, $Ch_d(X)=Ch'_d(X)$, which completes the proof that $Ch'_d$ is a completion of $Ch_d$.

Since all school admissions rules satisfy substitutability and LAD, so does $Ch'_d$.
\end{proof}

All of the assumptions on school admissions rules stated in Claims \ref{claim:feasible}, \ref{claim:acceptant}, and \ref{claim:completion} are satisfied when school admissions rules are \df{responsive}: each school has a ranking of contracts associated with itself and, from  any given set of contracts, each school chooses contracts  with the highest rank until the capacity of the school is full or there are no more contracts left. Responsive admissions rules satisfy substitutability and LAD. Furthermore, for every school $c_i$, $\abs{Ch_{c_i}(X)}=\min\{q_{c_i},\abs{X_{c_i}}\}$.\footnote{See \citet{chayen17} for a characterization of responsive admissions rules using substitutability.} By the claims stated above, when school admissions rules are responsive, district admissions rule $\chd$ is feasible and acceptant, and it has a completion that satisfies substitutability and LAD.

Based on these results, we provide examples of district admissions rules that further satisfy the additional assumptions considered in different parts of our paper.

\subsection{District Admissions Rules Satisfying the Assumptions in Theorem \ref{thm:welfimprove2}}\label{ex:respecting}

We use the district admissions rule construction above and we further specify each school's admissions rule. Each school has a responsive admissions rule. If a student is initially matched with a school, then her contract with this school is ranked higher than contracts of students who are not initially matched with the school. 
As before, district admissions rule $\chd$ is feasible and acceptant, and it has a completion that satisfies substitutability and LAD.

\begin{claim}
District admissions rule $\chd$ respects the initial matching.
\end{claim}

\begin{proof}
Let $c$ be the initial school of student $s$ and $x=(s,d,c)$. By construction, for any matching $X$ that is feasible for students, $x\in X$ implies $x\in \chd(X)$ because $c$ chooses $x$ from any set of contracts and $s$ does not have any other contract in $X$. Therefore, $\chd$ respects the initial matching.
\end{proof}

\subsection{District Admissions Rules Satisfying the Assumptions in Theorem \ref{thm:balance}}\label{ex:rationed}

We modify the district admissions rule construction in Appendix \ref{section:genex}.
Each school has a ranking of contracts associated with itself.
When it is the turn of a school, it accepts contracts
that have the highest rank until the capacity of the school is full, or the number
of contracts chosen by the district is $k_d$, or there are no more contracts left.
The remaining contracts of a chosen student are removed.

District admissions rule $\chd$ is feasible because no school admits more students than its capacity and no student is admitted to more than one school.

\begin{claim}
District admissions rule $Ch_d$ is acceptant.
\end{claim}

\begin{proof}
To show acceptance, suppose that matching $X$ is feasible for students and $x\in X_d \setminus \chd(X)$.
There exists $i\leq n$ such that $c_i = c(x)$. Since $X$ is feasible for students, $x\in X\setminus Y_{i-1}$
where $Y_{i-1}$ is the set of all contracts in $X$ associated with students who are chosen by schools $c_1,\ldots,c_{i-1}$.
Because $x\in X_d \setminus \chd(X)$, $x$ is not chosen by $c_i$. Then, by construction,
either $c_i$ fills its capacity or the district admits $k_d$ students, which implies that $\chd$ is acceptant.
\end{proof}

\begin{claim}
District admissions rule $Ch_d$ has a completion that satisfies substitutability and LAD.
\end{claim}

\begin{proof}
First, we construct a completion of $\chd$. Define the following district admissions rule: given a set of contracts $X$, when it is the turn
of a school, it chooses from all the contracts in $X$. Each school chooses contracts using the same priority order until the school
capacity is full, or the district has $k_d$ contracts, or there are no more contracts left. Denote this admissions rule by $Ch'_d$. Suppose
that $Ch'_d(X)$ is feasible for students. Then, by construction, $Ch'_d(X)=\chd(X)$. Therefore, $Ch'_d$ is a completion of $\chd$.

Next, we show that $\chd'$ satisfies LAD. 
Suppose that $Y\supseteq X$. 
Every school $c_{i}$ chooses weakly more contracts from $Y$ than $X$ unless the number of contracts chosen from $Y$ by the district reaches $k_d$. Since the number of chosen contracts from $X$ cannot exceed $k_d$ by construction, $\chd'$ satisfies LAD.

Finally, we show that $\chd'$ satisfies substitutability. Suppose that $x\in X \subseteq Y$ and $x\in \chd'(Y)$. Therefore, the number of
contracts chosen from $Y$ by schools preceding $c(x)$ is strictly less than $k_d$. This implies that the number of contracts
chosen from $X$ by schools preceding $c(x)$ is weakly less than this number as weakly more contracts are chosen by schools preceding school $c(x)$ in $Y$ than $X$.
As a result, for school $c(x)$, weakly more contracts can be chosen from $X$ than $Y$.

The ranking of contract $x$ among $Y$ in the ranking of school $c(x)$ is high enough that it is chosen from set $Y$. Therefore, the ranking of contract $x$ among $X$ in the ranking of school $c(x)$ must be high enough to be chosen from set $X$ because weakly more contracts are chosen from $X$ than $Y$ for school $c(x)$.
\end{proof}

Furthermore, by construction, district admissions rule $\chd$ never chooses more than $k_d$ students. Therefore, it is also rationed.

\subsection{District Admissions Rules Satisfying the Assumptions in Theorem \ref{thm:diversity}}\label{ex:reserve}

A profile of district admissions rules can accommodate unmatched students by \emph{reserving} seats for different types of students:

\begin{definition}
Let $c$ be a school in district $d$. A district admissions rule $\chd$ has a \df{reserve} of $r_c^t$ for type-$t$ students at school $c$ if, for any feasible matching $X$ that does not have any contract associated with type-$t$ student $s$, if $\abs{\{x\in X_c| \tau(s(x))=t\}}<r^t_c$, then $x=(s,d,c)$ satisfies $x \in \chd(X\cup \{x\})$.
\end{definition}

A reserve for a student type at a school $c$ guarantees space for this type at school $c$. 
Therefore, when a student is unmatched at a feasible matching and the reserve for her type is not yet filled at a school, the district will accept this student at that school if she applies to it.

\begin{claim}\label{claim:reserve}
Suppose that districts have admissions rules with reserves such that $\sum_{c} r_c^t=k^t$ for every type $t$. Then the profile of
district admissions rules accommodates unmatched students.
\end{claim}

\begin{proof}
Suppose that student $s$ is unmatched at a feasible matching $X$. Let
$t$ be the type of student $s$. Then there exists a school $c$ such that the number of
type-$t$ students in $c$ at $X$ is strictly less than $r_{c}^t$ because $\sum_c r_c^t=k^t$.
By definition of reserves, $x=(s,c)$ satisfies $x\in Ch_{d(c)}(X \cup \{x\})$.
\end{proof}

A district can have type-specific reserves at its schools in different ways. In the rest of this subsection, we use school
admissions rules with reserves introduced by \cite{hayeyi13}
to construct a fairly general
 example in which a district has schools with type-specific reserves.  Let $r_{c}^t$ be the number of seats reserved by school $c$ for type-$t$ students. Suppose that the type-specific ceilings for schools are given and that they satisfy the assumptions in Section \ref{sec:diversity}. 
Assume that, for every district $d$, $\sum_{c: d(c)=d} \sum_t r_c^t=k_d$, $\sum_c r_c^t=k^t$ and, for every type $t$ and school $c$, $r_c^t\leq q_c^t$. Furthermore, assume that $\sum_{t} r_{c}^t \leq q_{c}$ for every school $c$.

Consider the following district admissions rule for district $d$. Schools are ordered as  $c_1,c_2,\dots, c_n$.
Each school has a ranking over contracts associated with it
and a linear order over student types. First, all schools choose contracts for their reserve seats according to the order
$c_1,c_2,\dots, c_n$. When it is the turn of school $c_i$, all
contracts associated with students whose contracts were previously chosen are
removed. School $c_i$ chooses contracts for its reserved seats 
so that, for every type, either reserved seats
are filled or there are no more contracts associated with  students of that type remaining. 
Then all schools choose contracts for their empty seats in order.
When it is the turn of school $c_i$, all contracts of previously chosen students are
removed. School $c_i$ chooses from the remaining contracts in order. When a contract of a
type-$t$ student is considered, this contract is chosen unless the school's capacity is filled or
its type-$t$ ceiling is filled or the district has $k_d$ contracts. Denote this district
admissions rule by $\chd$.

District admissions rule $\chd$ is feasible because a student cannot have more than one contract
and a school cannot have more contracts than its capacity at any chosen set of contracts.
It is also weakly acceptant and rationed by construction. Furthermore, for every type $t$ and
school $c$, the district cannot admit more than $q_c^t$ type-$t$ students at $c$, so it has a school-level type-specific ceiling of $q_c^t$ for type-$t$ students and school $c$.

\begin{claim}
District admissions rule $\chd$ has a completion that satisfies substitutability and LAD.
\end{claim}

\begin{proof}
For any set of contracts $X$, school $c$, and type $t$, let $X_c^t$ denote the set of all contracts in $X$ that are associated with school $c$ and  type-$t$ students.

Consider the construction of $\chd$ above, but modify it by not removing  contracts of students who are chosen
previously. Denote this district admissions rule by $\chd'$. To show that $\chd'$ is a completion of $\chd$, consider a set of contracts $X$ and suppose that
$\chd'(X)$ is feasible for students. Since the only difference in the constructions of $\chd$ and
$\chd'$ is the removal of contracts of previously chosen students, it must be that $\chd'(X)=\chd(X)$. Therefore,
$\chd'$ is a completion of $\chd$.

To prove substitutability of $\chd'$, suppose, for contradiction, that there exist sets of contracts $X$ and $Y$ with $X \subseteq Y$
and a contract $x \in X$ such that  $x \in \chd'(Y) \setminus \chd'(X)$. Let $s$ and $c$ be such that $x=(s,c)$ and $t=\tau(s)$.
First, note that $|X^t_c|>r_c^t$ because $x \not\in \chd'(X)$. Since $Y \supseteq X$,
$|Y^t_c| \ge |X^t_c|>r_c^t$ is implied. Therefore, it is after all schools in $d$ have chosen contracts based on their reserves in the
algorithm describing $\chd'$ that contract $x$ is chosen by $\chd'$ given $Y$.  Let $n(X)$ and $n(Y)$ be the numbers of
contracts that have been chosen by all schools before the step (call it step $\kappa_c$) at  which school $c$ chooses students
beyond its reserve under $X$ and $Y$, respectively. Because $x \in \chd'(Y)$, it follows that $n(Y)<k_d$. Therefore, for each
school $c'$, the number of contracts chosen by $c'$ before step $\kappa_c$ under $Y$ is weakly larger than those under $X$,
which we prove as follows:

\begin{itemize}
\item Suppose that school $c'$ is processed after school $c$ in the algorithm deciding $\chd$. Then, by step $\kappa_c$, $c'$ is matched with students of each type $t'$ only up to its type-$t'$ reserve. More formally, the numbers of type-$t'$ students matched to $c'$ are equal to $\min\{r^{t'}_{c'}, |X^{t'}_{c'}|\}$ and $\min\{r^{t'}_{c'}, |Y^{t'}_{c'}|\}$ under $X$ and $Y$, respectively. Obviously, the latter expression is no smaller than the former expression.
\item Suppose that school $c'$ is processed before school $c$ in the algorithm deciding $\chd$. Recall that $n(Y)<k_d$. Therefore, for school $c'$, it is either (i) as many as $q_{c'}$ students are matched to $c'$ under $\chd(Y)$, or (ii) for each type $t'$, the number of type-$t'$ students matched to $c'$ in $Y$ is $\min\{q^{t'}_{c'},|Y^{t'}_{c'}|\}$. In case (i), the desired conclusion follows trivially because, given any set of contracts, the number of students matched to $c'$ cannot exceed $q_{c'}$. For case (ii), under $X$, the number of type-$t'$ students matched to $c'$ cannot exceed $\min\{q^{t'}_{c'},|X^{t'}_{c'}|\} \le \min\{q^{t'}_{c'},|Y^{t'}_{c'}|\}$. Summing up across all types, we obtain the desired conclusion.
\end{itemize}

Thus $n(X) \le n(Y)$, so $k_d-n(X) \ge k_d -n(Y)$. Now, in step $\kappa_c$, school $c$ will choose all the applications until either
the total number of contracts chosen reaches $k_d$, or the total number of contracts chosen at $c$ reaches $q_c$, or the
number of  contracts chosen at $c$ that are associated with type $t$ students reaches $q^t_c$. Given the previous fact
that $k_d-n(X) \ge k_d -n(Y)$, the fact that $Y \supseteq X$, and the fact that $x$ is chosen by $c$ in step $\kappa_c$ under
$Y$, it has to be the case that $x$ is also chosen by $c$ under $X$ in step $\kappa_c$ or before. We prove this as follows:

\begin{itemize}
\item If $|X^t_c| \le r^t_c$, then $x$ is chosen in the reserve stage by construction.

\item  Let $|X^t_c| > r^t_c$. First note that at step $\kappa_c$ under $X$ and $Y$, for each type $t$, there are fewer contracts
associated with school $c$ and type-$t$ students that remain to be processed under $X$ than under $Y$ ($X \subseteq Y$, and there is no contract in $X=X \cap Y$ that is processed in the reserve stage under $Y$ but not under $X$), so the subset of $X$ that should be processed in $\kappa_c$ is a subset of the corresponding subset of $Y$. Moreover, the remaining number of contracts to
be chosen before reaching the ceiling at $c$ for each type $t$ in step $\kappa_c$ is weakly larger at $X$ than at $Y$ by the definition of the reserve stage. Finally, as argued above, the total number of students in the district who can still be chosen at $\kappa_c$ is weakly larger under $X$ than at $Y$, so whenever $x$ is chosen under $Y$ in this stage, $x$ is chosen under $X$ in this stage or the reserve stage.
\end{itemize}

This is a contradiction to the assumption that $x\notin Ch'_d(X)$.

To show that $\chd'$ satisfies LAD, suppose, for contradiction, that there exist two sets of contracts $X, Y$ with $X\subseteq Y$ and $|\chd'(Y)|<|\chd'(X)|$.
Then, because $\chd'(Y)=\bigcup_{c: d(c)=d} (\chd'(Y) \cap Y_c)$ and $\chd'(X)=\bigcup_{c: d(c)=d} (\chd'(X) \cap X_c)$, there exists a school $c$ with $d(c)=d$ such that
\begin{align}
|\chd'(Y) \cap Y_c|<|\chd'(X) \cap X_c|.\label{load-inequality1}
\end{align}
 Fix such $c$ arbitrarily. Next, note that
\begin{align*}
|\chd'(Y)|<|\chd'(X)| \le k_d,\\
|\chd'(Y) \cap Y_c|<|\chd'(X) \cap X_c| \le q_c,
\end{align*}
where the first line follows because $\chd'$ is rationed by construction, and the second line also holds by construction of $\chd'$.
Therefore,
\begin{align}
|\chd'(Y) \cap Y_c^t| & = \min \{|Y_c^t|,q^t_c\} \notag \\
& \ge \min \{|X_c^t|,q^t_c\} \notag \\
& = |\chd'(X) \cap X_c^t|, \label{load-inequality2}
\end{align}
 for each type $t \in \T$. Because $\chd'(Y) \cap Y_c = \cup_{t \in \T} (\chd'(Y) \cap Y_c^t)$ and $\chd'(X) \cap X_c= \cup_{t \in \T}
 (\chd'(X) \cap X_c^t)$, inequality (\ref{load-inequality2}) and the fact $Y_c^t \cap Y_c^{t'}=X_c^t \cap X_c^{t'}=\emptyset$ for any pair of types $t, t'$ with $t \neq t'$ imply
\begin{align*}
|\chd'(Y) \cap Y_c| \ge |\chd'(X) \cap X_c|,
\end{align*}
which contradicts inequality (\ref{load-inequality1}).
\end{proof}

\subsection{District Admissions Rules Satisfying the Assumptions in Theorem \ref{thm:welfimprove}}\label{ex:favor}
Consider the district admissions rule construction in Appendix \ref{section:genex}. In this example, let each school use a priority ranking in such a way that all contracts of students from district $d$ are ranked higher than the other contracts.

\begin{claim}
District admissions rule $\chd$ favors own students.
\end{claim}

\begin{proof}
Suppose that $X$ is feasible for students. When it is the turn of school $c_i$, it considers $X_{c_i}$. Therefore, $\chd(X)=Ch_{c_1}(X_{c_1})\cup \ldots \cup Ch_{c_k}(X_{c_k})$. Furthermore, $Ch_{c_i}(X_{c_i})\supseteq Ch_{c_i}(\{x\in X_{c_i}|d(s(x))=d\})$ by construction. Taking the union of all subset inclusions yields $Ch_d(X)\supseteq Ch_{d}(\{x\in X_d|d(s(x))=d\})$. Therefore, $\chd$ favors own students.
\end{proof}


\section{An Example for Section \ref{sec:diversity}} \label{subsec:diversityexample}
In this section, we provide an example in which the conditions on the admissions rules
stated in Theorem \ref{thm:diversity} are satisfied and, therefore, SPDA satisfies the diversity policy.

Consider a problem with two school
districts, $d_{1}$ and $d_{2}$. District $d_{1}$ has school $c_{1}$ with
capacity three and school $c_{2}$ with capacity two. District $d_{2}$ has
school $c_{3}$ with capacity two and school $c_{4}$ with capacity one. There
are seven students: students $s_{1}$, $s_{2}$, $s_{3}$, and $s_{4}$ are from
district $d_{1}$, whereas students $s_{5}$, $s_{6}$, and $s_{7}$ are from
district $d_{2}$. Students $s_{1}$, $s_{5}$, $s_{6}$, and $s_{7}$ have type
$t_{1}$ and $s_{2}$, $s_{3}$, and $s_{4}$ have type $t_{2}$. To construct
district admissions rules that satisfy the properties stated in Theorem
\ref{thm:diversity}, let us first specify type-specific ceilings and calculate
implied floors and implied ceilings. Suppose that%
\begin{eqnarray*}
q_{c_{1}}^{t_{1}} =1,\text{ }q_{c_{1}}^{t_{2}}=2,\text{ }
q_{c_{2}}^{t_{1}}=1,\text{ }q_{c_{2}}^{t_{2}}=1, \\
\text{ }q_{c_{3}}^{t_{1}} =2,\text{ }q_{c_{3}}^{t_{2}}=1,\text{ }
q_{c_{4}}^{t_{1}}=1,\text{ }q_{c_{4}}^{t_{2}}=1.
\end{eqnarray*}

\noindent
These yield {the following implied floors},{\footnote{
To see this, note that there cannot be zero type-$t_{1}$ students in
$d_{1}$ (otherwise not all type-$t_{1}$ students can be matched
since there are only three spaces available for type-$t_{1}$ students in
$d_{2}$). If there is one type-$t_{1}$ student in $d_{1}$,
there has to be three type-$t_{1}$ students in $d_{2}$, which
implies there cannot be any type-$t_{2}$ students in $d_{2}$, and
this implies there will be three type-$t_{2}$ students in $d_{1}$.
If there are two type-$t_{1}$ students in $d_{1}$, there have to be
two type-$t_{2}$ students in $d_{2}$, which implies there is one
type-$t_{2}$ student in $d_{2}$, and this implies there will be two
type-$t_{2}$ students in $d_{1}$. By noting these minimum and
maximum numbers, we obtain the implied reserves and implied ceilings
accordingly. These bounds are achievable because it is feasible to
have (i) one type-$t_{1}$ student in $d_{1}$,
three type-$t_{1}$ students in $d_{2}$, zero type-$t_{2}$ students in $d_{2}$,
and three type-$t_{2}$ students in $d_{1}$, and (ii) two type-$t_{1}$
students in $d_{1}$, two type-$t_{2}$ students in $d_{2}$, one type-$t_{2}$
student in $d_{2}$, and two type-$t_{2}$ students in $d_{1}$.}}
\begin{equation*}
\hat{p}_{d_{1}}^{t_{1}}=1,\text{ }\hat{p}_{d_{1}}^{t_{2}}=2,\text{ }\hat{p}%
_{d_{2}}^{t_{1}}=2,\text{ }\hat{p}_{d_{2}}^{t_{2}}=0,
\end{equation*}%
and {implied ceilings }
\begin{equation*}
\hat{q}_{d_{1}}^{t_{1}}=2,\text{ }\hat{q}_{d_{1}}^{t_{2}}=3,\text{ }\hat{q}%
_{d_{2}}^{t_{1}}=3,\text{ }\hat{q}_{d_{2}}^{t_{2}}=1.
\end{equation*}

For any type $t$ and two districts $d$ and $d^{\prime }$, denote {$\hat{q}%
_{d}^{t}/k_{d}-\hat{p}_{d^{\prime }}^{t}/k_{d^{\prime }}$ } by $\Delta
_{d,d^{\prime }}^{t}$. Using the implied floors and ceilings above, we get:
\begin{align*}
\Delta _{d_{1},d_{2}}^{t_{1}}& =2/4-2/3=-1/6, \\
\Delta _{d_{2},d_{1}}^{t_{1}}& =3/3-1/4=3/4, \\
\Delta _{d_{1},d_{2}}^{t_{2}}& =3/4-0/3=3/4,\text{ and} \\
\Delta _{d_{2},d_{1}}^{t_{2}}& =1/3-2/4=-1/6.
\end{align*}%
Hence, these type-specific ceilings satisfy the condition stated in Theorem %
\ref{thm:diversity} that $\Delta _{d,d^{\prime }}^{t}\leq \alpha $ for $%
\alpha =0.75$.

We construct district admissions rules that have type-specific ceilings, and are rationed and weakly acceptant.
Furthermore, the profile of district admissions rules accommodates unmatched students. As in
Appendix \ref{ex:reserve}, we consider type-specific reserves (as detailed below,
we first fill in the reserves while applying the district admissions rule that uses type-specific reserves). Let us
consider the reserves for schools as follows:

\begin{equation*}
r_{c_{4}}^{t_{2}}=0,\text{ and }r_{c}^{t}=1\text{ for all other }c,t.
\end{equation*}

Consider the following district admissions rule. 
 For each district, schools
and student types are ordered and each school has a linear order over students.
First, schools choose contracts for their reserved seats following the master
priority list until the reserves are filled or all the applicants of the
relevant type are processed. Then, following the given order over schools
and student types, schools choose from the remaining contracts following the
linear order over students in order to fill the rest of their seats until
the school capacity is filled, or the district has $k_{d}$ contracts, or
district type-specific ceilings are filled, or there are no more remaining
contracts.\footnote{
In Appendix \ref{ex:reserve}, we provide a class of admissions rules that
include the one we consider here. These admissions rules satisfy all of the
assumptions that we make in this section.}

To give a more concrete example, suppose that the linear order over students
for each school is as follows: $s_{1}\succ s_{2}\succ s_{3}\succ s_{4}\succ
s_{5}\succ s_{6}\succ s_{7}$ and schools and types are ordered from the
lowest index to the highest. Then, for example, we have the following:
\begin{equation*}
Ch_{d_{1}}(\{(s_{1},c_{1}),(s_{2},c_{1}),\left( s_{3},c_{1}\right) ,\left(
s_{4},c_{1}\right)
,(s_{5},c_{2}),(s_{6},c_{2})\})=%
\{(s_{1},c_{1}),(s_{2},c_{1}),(s_{3},c_{1}),(s_{5},c_{2})\}.
\end{equation*}%
Let us elaborate on how we determine the chosen set of contracts in the
above case. School $c_{1}$ considers contracts with students $s_{1}$, $s_{2}$, $s_{3}$,
and $s_{4}$. Among these students, $c_{1}$ accepts $s_{1}$ for its
reserve for type $t_{1}$, and $s_{2}$ for its reserve for type $t_{2}$.
Moreover, school $c_{2}$ considers contracts with students $s_{5}\ $and $%
s_{6}$. Among these students students, $c_{2}$ accepts $s_{5}$ for its
reserve for type $t_{2}$. For the remainder of seats, $s_{3}$ is accepted by
$c_{1}$ since (i) $c_{1}$'s type $t_{2}$ ceiling is not full, (ii) $c_{1}$'s
capacity is not full, and (iii) district $d_{1}$ has only three accepted
contracts at this point. Next, $s_{4}$ and $s_{5}$ are rejected since $d_{1}$
has accepted four contracts at this point. This results in the chosen set of
contracts presented above.

To illustrate the SPDA outcomes, consider
student preferences given by the following table.

\begin{equation*}
\begin{tabular}{lllllll}
$\underline{P_{s_{1}}}$ & $\underline{P_{s_{2}}}$ & $\underline{P_{s_{3}}}$
& $\underline{P_{s_{4}}}$ & $\underline{P_{s_{5}}}$ & $\underline{P_{s_{6}}}$
& $\underline{P_{s_{7}}}$ \\
$c_{2}$ & $c_{3}$ & $c_{4}$ & $c_{1}$ & $c_{1}$ & $c_{1}$ & $c_{2}$ \\
$\vdots $ & $c_{1}$ & $c_{2}$ & $c_{3}$ & $c_{2}$ & $c_{4}$ & $c_{3}$ \\
& $\vdots $ & $\vdots $ & $c_{2}$ & $c_{3}$ & $c_{3}$ & $\vdots $ \\
&  &  & $c_{4}$ & $c_{4}$ & $c_{2}$ &
\end{tabular}
\
\end{equation*}

SPDA results in the following outcome:
\begin{equation*}
\{(s_{1},c_{2}),(s_{2},c_{3}),(s_{3},c_{2}),(s_{4},c_{1}),(s_{5},c_{1}),(s_{6},c_{4}),(s_{7},c_{3})\}.
\end{equation*}

District $d_{1}$ is assigned two students of both types and
district $d_{2}$ is assigned two type-$t_{1}$ students and one type-$t_{2}$
student. As a result, the ratio difference for type-$t_{1}$ students between
these districts
is roughly $0.17$, and the ratio
difference for type-$t_{2}$ students
is roughly $0.17$.
This example illustrates that the
actual ratio differences can be significantly lower than the
one given by Theorem \ref{thm:diversity} ($0.17$ versus $0.75$).


\section{Omitted Proofs}\label{sec:proofs}

In this section, we provide the omitted proofs.

\begin{proof}[Proof of Theorem \ref{thm:welfimprove2}]
First, to show the ``if'' part, suppose that all district admissions rules respect the initial matching. In SPDA, each student $s$ goes down in her preference order, and either SPDA ends before student $s$ reaches her initial school (which is a preferred outcome over the initial school), or student $s$ reaches her initial school. In the latter case, she is matched with her initial school because the district's admissions rule respects the initial matching and the district always considers a set of contracts that is feasible for students at any step of SPDA. From this step on, the district accepts this contract, so student $s$ is matched with her initial school. Therefore, SPDA satisfies individual rationality.

To prove the ``only if'' part, suppose that there exists a district $d$ with an admissions rule that fails to respect the initial matching.  Hence, there exists a matching $X$, which is feasible for students, that includes $x=(s,d,c)$ where school $c$ is the initial school of student $s$ and $x\notin Ch_d(X)$. Now, consider student preferences such that every student associated with a contract in $X_d$ prefers that contract the most and all other students prefer a contract associated with a different district the most. Then, at the first step of SPDA, district $d$ considers matching $X_d$ and tentatively accepts $Ch_d(X_d)$. Since $x\notin Ch_d(X_d)$, contract $x$ is rejected at the first step. Therefore, student $s$ is matched with a strictly less preferred school than her initial school, which implies that SPDA does not satisfy individual rationality.
\end{proof}
\medskip

\begin{proof}[Proof of Theorem \ref{thm:balance}]
We first prove that if each district admissions rule is rationed,
then SPDA satisfies the balanced-exchange policy. Let $X$ be the matching
produced by SPDA for a given preference profile.

We begin by showing that each student must be matched with a school in $X$. Suppose, for contradiction, that student $s$ is unmatched. Since $X$ is a stable matching, every contract $x=(s,d,c)$ associated with the student is rejected by the corresponding district, i.e., $x\notin \chd(X \cup \{x\})$. Otherwise, student $s$ and district $d$ would like to match with each other using contract $x$, contradicting the stability of matching $X$. Since $X \cup \{x\}$ is feasible for students, acceptance implies that, for each district $d$, either every school in the district is full or that the district has at least $k_d$ students at matching $X$. Both of them imply that the district has at least $k_d$ students in matching $X$ since the sum of the school capacities in district $d$ is at least $k_d$. But this is a contradiction to the assumption that student $s$ is unmatched since the existence of an unmatched student implies that there is at least one district $d$ such that the number of students in $X_d$ is less than $k_d$. Therefore, all students are matched in $X$.

Because $X$ is the outcome of SPDA, it is feasible for students. Therefore, because district admissions rules are rationed,
the number of students in district $d$ cannot be strictly more than $k_d$ for any district $d$. Furthermore, since every student is matched, the number of students in district $d$ must be exactly $k_d$ (because, otherwise, at least one student would have been unmatched.) As a result, SPDA satisfies the balanced-exchange policy.

Next, we prove that if at least one district's admissions rule fails to be rationed, then there exists a student preference profile under which SPDA does not satisfy the balanced-exchange policy. Suppose that there exist a district $d$ and a matching $X$, which is feasible for students, such that $\abs{\chd(X)}> k_d$. Consider a feasible matching $X'$ such that (i) all students are matched, (ii) $X'_d =\chd(X)$, and (iii) for every district $d'\neq d$, $\abs{X'_{d'}}\leq k_{d'}$. The existence of such $X'$ is guaranteed since every district has enough capacity to serve its students (i.e., for every district $d'$, $\sum_{c:d(c)=d'} q_c \geq k_{d'}$), and $\abs{\chd(X)} > k_{d}$. Now, consider any student preferences, where every student likes her contract in $X'$ the most.

We show that SPDA stops in the first step. For district $d'\neq d$, $X'_{d'}$ is feasible and the number of students matched to $d'$ at $X'_{d'}$ is weakly less than $k_{d'}$. Since $Ch_{d'}$ is acceptant, $Ch_{d'}(X'_{d'})=X'_{d'}$. For district $d$, we need to show that $Ch_{d}(X'_{d})=X'_{d}$, which is equivalent to $\chd(\chd(X))=\chd(X)$. Let $Ch'_d$ be a completion of $\chd$ that satisfies path independence. Because $X$ and $\chd(X)$ are feasible for students, $Ch'_d(X)=\chd(X)$ and $Ch'_d(Ch'_d(X))=\chd(\chd(X))$. Furthermore, since $Ch'_d$ is path independent, $Ch'_d(Ch'_d(X))=Ch'_d(X)$, which implies $\chd(\chd(X))=\chd(X)$. As a result, $\chd(X'_d)=X'_d$. Therefore, SPDA stops at the first step since no contract is rejected.

Since SPDA stops at the first step, the outcome is matching $X'$. But $X'$ fails the balanced-exchange policy because $\abs{X'_d}=\abs{\chd(X)}>k_d$.
\end{proof}
\medskip

\begin{proof}[Proof of Theorem \protect\ref{thm:diversity}]
To prove this result, we provide the following lemmas.

\begin{lemma}\label{lem:match} 
If a profile of district admissions rules accommodates unmatched
students, every student is matched to a school in SPDA.
\end{lemma}

\begin{proof}[Proof of Lemma \protect\ref{lem:match}]
Let $X$ be the outcome of SPDA for some preference profile. Suppose, for
contradiction, that student $s$ is unmatched. Since $X$ is a stable matching
and student $s$ prefers any contract $x=(s,d,c)$ to being unmatched,
$x\notin Ch_d(X \cup \{x\})$. But this is a contradiction to the assumption
that the profile of district admissions rules accommodates unmatched students.
\end{proof}

\begin{lemma}\label{lem:lower}
For each type $t$, district $d$, and legitimate matching $X$,
we have $\hat{q}_{d}^{t} \geq \xi_{d}^{t} (X) \geq \hat{p}_{d}^{t}$.
Moreover, for each type $t$ and district $d$, there exist legitimate
matchings $X$ and $X'$ such that $\xi_{d}^{t}(X) = \hat{p}_{d}^{t}$ and
$\xi_{d}^{t}(X') = \hat{q}_{d}^{t}$.
\end{lemma}

\begin{proof}[Proof of Lemma \protect\ref{lem:lower}]
Observe that for every legitimate matching $X$, the induced distribution satisfies
the constraints of the linear program. Therefore, the first part follows from the definition
of the implied floors and ceilings. For the second part, note that there exists a solution
to the linear program such that the ceiling and the floor are attained. Furthermore,
every solution $y=(y^t_c)_{c \in \C, t \in \T}$ of the linear program can be
supported by a legitimate matching $X$ such that $y_c^t=\xi_c^t(X)$ for every $c$ and $t$.
\end{proof}

\begin{lemma}
\label{lem:suff} For each $t \in \mathcal{T}$ and $d ,d^{\prime }\in
\mathcal{D}$ with $d \neq d^{\prime }$, there exists a legitimate matching $%
X $ such that $\xi_{d}^{t} (X) = \hat{q}_d^t$ and $\xi_{d^{\prime }}^{t} (X)
= \hat{p}_{d^{\prime }}^t$.
\end{lemma}

\begin{proof}[Proof of Lemma \protect\ref{lem:suff}]

Let $\hat{X}$ be a legitimate matching such that $\xi _{d}^{t}(\hat{X})=\hat{%
q}_{d}^{t}$ and $\mathcal{M}_{0}$ be the set of all legitimate matchings. Let%
\begin{equation*}
\mathcal{M}_{1}\equiv\{X\in \mathcal{M}_{0}|\xi _{d^{\prime }}^{t}(X)=\hat{p}%
_{d^{\prime }}^{t}\}.
\end{equation*}%
$\mathcal{M}_{1}$ is nonempty due to Lemma \ref{lem:lower}.
Next, let
\begin{equation*}
\mathcal{M}_{2}\equiv\{X\in \mathcal{M}_{1}|\sum_{\t,\c}\mid \xi _{\c}^{\t}(X)-\xi
_{\c}^{\t}(\hat{X})\mid \leq \sum_{\t,\c}\mid \xi _{\c}^{\t}(X^{\prime })-\xi
_{\c}^{\t}(\hat{X})\mid \text{ for every $X^{\prime }\in \mathcal{M}_{1}$}\}.
\end{equation*}%
$\mathcal{M}_{2}$ is  nonempty because $\mathcal M_1$ is a finite set.
We will show that for any $X\in \mathcal{M}_{2}$, $\xi _{d}^{t}(X)=\xi
_{d}^{t}(\hat{X})=\hat{q}_{d}^{t}$.

To prove the above claim, assume for contradiction that there exists $X\in \mathcal{M}_{2}$ such that $%
\xi _{d}^{t}(X)\neq \xi _{d}^{t}(\hat{X})$.
By Lemma \ref{lem:lower}, $\xi _{d}^{t}(X)\neq \xi _{d}^{t}(\hat{X})$
implies that $\xi _{d}^{t}(X)<\xi _{d}^{t}(\hat{X})$. Then there exists $c$ with $d(c)=d$ such that
$\xi _{c}^{t}(X)<\xi _{c}^{t}(\hat{X})$.  Consider the following procedure.

\begin{description}
\item[Step 0] Initialize by setting $(t_1,c_1):=(t,c)$. Note that $\xi _{c_1}^{t_1}(X)<\xi _{c_1}^{t_1}(\hat{X})$ by definition of $c$.

\item[Step $i \ge 1$] Given  sequences of type-school pairs $((t_j,c_j))_{1 \le j \le i}$ and $((t_{j+1},c^*_j))_{1 \le j < i}$, proceed as follows. We begin with $( t_i,c_i)$. Note that (by assumption for $i=1$, and as shown later for $i \ge 2$), $\xi _{c_i}^{t_i}(X)<\xi _{c_i}^{t_i}(\hat{X})$.
 Denote $d_i=d(c_i)$. Now,
\begin{enumerate}
\item Suppose that
 there exists $i'<i$ such that either (i) $c^*_{i'}=c_i$ or (ii) $\xi^{t_i}_{c_i}(X)<q_{c_i}$ and $d(c^*_{i'})=d(c_i)$.
If such an index $i'$ exists, then set $(t_{i+1},c^*_i):=(t_{i'+1},c^*_{i'})$.

\item Suppose not. Then, if there exists $t' \in \T$ such that  $\xi _{c_i}^{t'}(X)>\xi _{c_i}^{t'}(\hat{X})$, then set $(t_{i+1},c^*_i):=(t',c_i)$.
\item If not, then note that $\sum_{\t \in \T} \xi^{\t}_{c_i}(X) <q_{c_i}$.\footnote{A proof of this fact is as follows. By an earlier argument, $\xi _{c_i}^{t_i}(X)<\xi _{c_i}^{t_i}(\hat{X})$. Moreover, by assumption $\xi _{c_i}^{\t}(X) \le \xi _{c_i}^{\t}(\hat{X})$ for every $\t \in \T$. Therefore, $\sum_{\t \in \T} \xi _{c_i}^{\t}(X)<\sum_{\t \in \T} \xi _{c_i}^{\t}(\hat{X}) \le q_{c_i}$.} Also note that
there exists a type-school pair $(t',c')$ with $c' \neq c_i$ such that
$\xi _{c'}^{t'}(X)>\xi _{c'}^{t'}(\hat{X})$
 and $d(c') =d_i$ because $\sum_{\c: d(\c)=d_i,\t \in \T} \xi^{\t}_{\c}(X)=\sum_{\c: d(\c)=d_i,\t \in \T} \xi^{\t}_{\c}(\hat X)=k_{d_i}$.
\begin{enumerate}
\item If $t'=t_i$, then let $\bar X$
be a matching such that
\begin{align*}
\xi^{\t}_{\c}(\bar X)= \begin{cases} \xi^{t_i}_{c_i}(X)+1 & \mbox{ for $(\t,\c)=(t_i,c_i)$}, \\
\xi^{t'}_{c'}(X)-1 & \mbox{ for $(\t,\c)=(t_i,c')$}, \\
\xi^{\t}_{\c}(X) & \mbox{ otherwise.}
\end{cases}
\end{align*}
Note that $\bar X \in \mathcal M_1$.\footnote{A proof  of this fact is as follows.  Because $\sum_{\t \in \T} \xi^{\t}_{c_i}(X) <q_{c_i}$, $\sum_{\t \in \T} \xi^{\t}_{c_i}(\bar X) = \sum_{\t \in \T} \xi^{\t}_{c_i}(X) +1 \le q_{c_i}$. For every $\c \neq c_i$, $\sum_{\t \in \T} \xi^{\t}_{\c}(\bar X) \le \sum_{\t \in \T} \xi^{\t}_{\c}(X) \le q_{\c}$. Thus, all school capacities are satisfied. For all $\c,\t$, $\xi^{\t}_{\c}(\bar X) \le \max \{ \xi^{\t}_{\c}(X),\xi^{\t}_{\c}(\hat X) \} \le q^{\t}_{\c}$ by construction, so all type-specific ceilings are satisfied. And $\sum_{\t \in \T,\c \in \C} \xi^{\t}_{\c}(\bar X) = \sum_{\t \in \T} \xi^{\t}_{\c}(X)$ by definition of $\bar X$, so $\bar X$ is a legitimate matching.
Finally, $\xi^{\t}_{\d}(\bar X)=\xi^{\t}_{\d}(X)$ for every $\t$ and $\d$, so $\bar X \in \mathcal M_1$.} Also, by construction, $\sum_{\t,\c}\mid \xi _{\c}^{\t}(\bar X)-\xi
_{\c}^{\t}(\hat{X}) \mid =\sum_{\t,\c}\mid \xi _{\c}^{\t}(X)-\xi
_{\c}^{\t}(\hat{X}) \mid -2 < \sum_{\t,\c}\mid \xi _{\c}^{\t}(X)-\xi
_{\c}^{\t}(\hat{X}) \mid$, which contradicts the assumption that $X \in \mathcal{M}_2$.
\item Therefore, suppose that $t' \neq t_i$ and set $(t_{i+1},c^*_i):=(t',c')$.
\end{enumerate}
\item The pair $(t_{i+1},c^*_i)$ created above satisfies $\xi_{c^*_i}^{t_{i+1}}(X)>\xi _{c^*_i}^{t_{i+1}}(\hat{X})$, so there exists $c' \in \C$ such that
$\xi_{c'}^{t_{i+1}}(X)<\xi _{c'}^{t_{i+1}}(\hat{X})$. Set $c_{i+1}=c'$. Note that $\xi _{c_{i+1}}^{t_{i+1}}(X)<\xi _{c_{i+1}}^{t_{i+1}}(\hat{X})$.
\end{enumerate}

\end{description}

We follow the procedure above to define $(t_1,c_1), (t_2,c^*_1), (t_2,c_2), (t_3,c^*_2), (t_3,c_3)$, and so forth. 
Because
 $\T$ is a finite set, we have $i$ and $j>i$ with $t_i=t_j$. Consider the smallest $j$  with this property (note that given such $j$, $i$ is uniquely identified). Now,
let $\bar X$ 
  be a matching such that
\begin{align*}
\xi^{\t}_{\c}(\bar X)= \begin{cases} \xi^{t_k}_{c_k}(X)+1 & \mbox{ for $(\t,\c)=(t_k,c_k)$ for any $k \in \{i,i+1, \dots,j-1\}$}, \\
\xi^{t_{k+1}}_{c^*_k}(X)-1 & \mbox{ for $(\t,\c)=(t_{k+1},c^*_k)$ for any $k \in \{i,i+1, \dots,j-1\}$}, \\
\xi^{\t}_{\c}(X) & \mbox{ otherwise.}
\end{cases}
\end{align*}

We will show $\bar X \in \mathcal{M}_{1}$. To do so, by construction of $\bar X$, first note
that  $\sum_{\t \in \T} \xi^{\t}_{\c}(\bar X) \le  \sum_{\t \in \T} \xi^{\t}_{\c}(X) +1 \le q_{\c}$ for any $\c \in \{c_i,\dots, c_{j-1}\}$ such that $\sum_{\t \in \T} \xi^{\t}_{\c}(X) <q_{\c}$.
Next, by construction of $\bar X$, $\sum_{\t \in \T} \xi^{\t}_{\c}(\bar X) =   \sum_{\t \in \T} \xi^{\t}_{\c}(X) = q_{\c}$ for every $\c \in \{c_i,\dots, c_{j-1}\}$ such that  $\sum_{\t \in \T} \xi^{\t}_{\c}(X) =q_{\c}$.
 Moreover,  $\sum_{\t \in \T} \xi^{\t}_{\c}(\bar X) \le   \sum_{\t \in \T} \xi^{\t}_{\c}(X) = q_{\c}$ for every $\c \in \{c^*_i,\dots, c^*_{j-1}\}$.
Finally, for every $\c \in \C \setminus\{c_i,\dots,c_{j-1},c^*_i,\dots,c^*_{j-1}\}$, $\sum_{\t \in \T} \xi^{\t}_{\c}(\bar X) = \sum_{\t \in \T} \xi^{\t}_{\c}(X) \le q_{\c}$. Thus, all school capacities are satisfied by $\bar X$. Also by construction of $\bar X$, for each $\d \in \D$, $\sum_{\c: d(\c)=\d} \xi^{\t}_{\c}(\bar X)=\sum_{\c: d(\c)=\d}  \xi^{\t}_{\c}(X)=k_{\d}$, so $\bar X$ is rationed.
Furthermore, for every $\c \in \C$ and $\t \in \T$, $\xi^{\t}_{\c}(\bar X) \le \max \{ \xi^{\t}_{\c}(X),\xi^{\t}_{\c}(\hat X) \}$ by construction, so all type-specific ceilings are satisfied.
 Moreover, by construction of $\bar X$, for
each $\t \in \T$,  either $\xi^{\t}_{\c}(\bar X)=\xi^{\t}_{\c}(X)$ for every $\c \in \C$ or there exists exactly one pair of schools $\c'$ and $\c''$ in $\C$ such that $\xi^{\t}_{\c'}(\bar X)=\xi^{\t}_{\c'}(\bar X)+1$, $\xi^{\t}_{\c''}(\bar X)=\xi^{\t}_{\c''}(\bar X)-1$, and $\xi^{\t}_{\c}(\bar X)=\xi^{\t}_{\c}(X)$ for every $\c \in \C \setminus \{\c', \c''\}$. Thus,  $\t \in \T$, $\sum_{\c \in \C}\xi^{\t}_{\c}(\bar X)=\sum_{\c \in \C}\xi^{\t}_{\c}(X)$ for every $\t \in \T$. Therefore, $\bar X$ is legitimate.

By construction of $\bar X$, either $%
\xi _{d'}^{t}(\bar X)=\xi _{d'}^{t}(X)$ or $\xi _{d'}^{t}(\bar{X})=\xi
_{d'}^{t}(X)-1$. This implies that $\bar{X}\in \mathcal{M}_{1}$.
Furthermore, $\sum_{\t,\c}\mid \xi _{\c}^{\t}(\bar{X})-\xi
_{\c}^{\t}(\hat{X})\mid < \sum_{\t,\c}\mid \xi _{\c}^{\t}(X)-\xi _{\c}^{\t}(\hat{X%
})\mid$, since while creating the $\xi _{\c}^{\t}(%
\bar{X})$ entries, we add $1$ to some entries of $X$ that satisfy $\xi _{\c}^{\t}(X)<\xi _{\c}^{\t}(\hat{X})$ and subtract $1$  from some entries of $X$ that
satisfy $\xi _{\c}^{\t}(X)>\xi _{\c}^{\t}(\hat{X})$.
These lead to a contradiction to the assumption that $X \in \mathcal M_2$, which completes the proof.
\end{proof}

Now we are ready to prove the theorem. The ``if'' part follows from Lemmas
\ref{lem:match} and \ref{lem:lower}. Specifically, by Lemma \ref{lem:match},
SPDA produces a legitimate matching. Therefore, by Lemma \ref{lem:lower}, we
have $\hat{p}_{d}^{t}\leq \xi _{d}^{t}\left( X\right) \leq \hat{q}_{d}^{t}$
for every $t\in \mathcal{T}$ and $d\in \mathcal{D}$. For each school
district $d$, hence, the maximum proportion of type-$t$ students that can be
admitted is $\hat{q}_{d}^{t}/k_{d}$ and the minimum proportion of type $t$
students that can be admitted is $\hat{p}_{d}^{t}/k_{d}$. Therefore, the
ratio difference of type-$t$ students in any two districts is at most $%
\underset{d\neq d^{\prime }}{\max }\{\hat{q}_{d}^{t}/k_{d}-\hat{p}%
_{d^{\prime }}^{t}/k_{d^{\prime }}\}$. We conclude that the $\alpha $%
-diversity policy is achieved when $\hat{q}_{d}^{t}/k_{d}-\hat{p}_{d^{\prime
}}^{t}/k_{d^{\prime }}\leq \alpha $ for every $t$, $d$, and $d^{\prime }$
with $d\neq d^{\prime }$.

The ``only if'' part of the theorem follows from Lemma \ref{lem:suff}. Suppose
that $\hat{q}_{d}^{t}/k_{d}-\hat{p}_{d^{\prime }}^{t}/k_{d^{\prime }}>\alpha
$ for some $t$, $d$, and $d^{\prime }$ with $d\neq d^{\prime }$. From Lemma %
\ref{lem:suff}, we know the existence of a legitimate matching $X$ such that
$\xi _{d}^{t}\left( X\right) =\hat{q}_{d}^{t}$ and $\xi _{d^{\prime
}}^{t}\left( X\right) =\hat{p}_{d^{\prime }}^{t}$. Consider a student
preference profile where each student prefers her contract in $X$ the most.
Then, since the admissions rules are weakly acceptant, SPDA ends at the
first step as all applications are accepted. Thus $X$ is the outcome
of SPDA and, therefore, the $\alpha$-diversity policy
is not satisfied.
\end{proof}

\medskip

\begin{proof}[Proof of Theorem \ref{thm:ttc}]
\cite{suzuki17} study a setting in which each student is initially matched with a school and there are no constraints associated with student types, that is, when there is just one type. In that setting, they show that if the distribution is M-convex, then their mechanism, called TTC-M, satisfies the policy goal, constrained efficiency, individual rationality, and strategy-proofness. To adapt their result to our setting, consider the hypothetical matching problem that we have introduced before
the definition of TTC in which each student is matched with a school-type pair and each student has strict preferences over all school-type pairs. It is straightforward to verify that this hypothetical problem satisfies all the conditions assumed by \cite{suzuki17}. In particular, M-convexity of  $\Xi \cap \Xi^0$ holds by assumption. Therefore, TTC-M in this market satisfies the policy goal, constrained efficiency, individual rationality,
and strategy-proofness.

We note that the outcome of our TTC is isomorphic to the outcome of TTC-M in the hypothetical problem in the following sense. Student $s$ is allocated to contract $(s,c)$ under preference profile $P=(P_s)_{s \in \S}$ at the outcome of TTC if, and only if, student $s$ is allocated to the school-type pair $(c,t)$ under preference profile $\tilde P=(\tilde P_s)_{s \in \S}$ at TTC-M in the hypothetical problem. The rest of the proof is devoted to showing that our TTC satisfies the desired properties in the original problem.

The result that TTC satisfies the policy goal follows from the result in \cite{suzuki17} that the distribution
corresponding to the TTC-M outcome is in $\Xi \cap \Xi^0$.

To show constrained efficiency, let $X$ be the outcome of TTC and, for each student $s \in \S$, let $(s,c_s)$ be the contract associated with student $s$ at matching $X$. Suppose, for contradiction, that there exists a feasible matching $X'$ with $\xi(X')\in \Xi$ that Pareto dominates matching $X$. Denoting $X'_s=(s,c'_s)$ for each student $s \in \S$, this implies $(s,c'_s) \mathrel{R_s} (s,c_s)$ for every student $s \in \S$, with at least one relation being strict. Then, by the construction of preferences $\mathrel{\tilde R_s}$ in the hypothetical problem, we have $(c'_s, \tau(s)) \mathrel{\tilde R_s} (c_s,\tau(s))$ for every student $s \in \S$, with at least one relation being strict. Moreover, because matching $X'$ is feasible in the original problem, 
$Y'=\{(c'_s, \tau(s))|(s,c'_s) \in X'\}$ is feasible in the hypothetical problem, and $Y=\{(c_s, \tau(s))|(s,c_s) \in X\}$ is the result of TTC-M. This is a contradiction to the result in \cite{suzuki17} that TTC-M is constrained efficient.

To show individual rationality, let matching $X$ be the outcome of TTC and, for each student $s \in \S$, let $X_s=(s,c_s)$ be the contract associated with student $s$ at matching $X$. Additionally, let $Y=\{(c_s, \tau(s))|(s,c_s) \in X\}$ be the result of TTC-M in the hypothetical problem. \cite{suzuki17} establish that TTC-M is individually rational, so $(c_s, \tau(s)) \mathrel{\tilde R_s} (c_0(s), \tau(s))$ for every $s \in \S$, where $c_0(s)$ denotes the initial school of student $s$. By the construction of $\tilde R_s$, this relation implies  $(s, c_s) \mathrel{R_s} (s,c_0(s))$ for every student $s \in \S$, which means $X$ is individually rational in the original problem.

To show strategy-proofness, in the original problem, let $s$ be a student, $t$ her type, $P_{-s}$ the preference profile of students other than student $s$,  $P_s$ the true preference of student $s$, and $P'_s$ a misreported preference of student $s$. Furthermore, let $c$ and $c'$ be schools assigned to student $s$ under $(P_s,P_{-s})$ and $(P_s',P_{-s})$ for TTC, respectively. Note that the previous argument establishes that, in the hypothetical problem, student $s$ is allocated to $(c,t)$ and $(c',t)$ under $(\tilde P_s, \tilde P_{-s})$ and  $(\tilde P'_s, \tilde P'_{-s})$, respectively.  Because TTC-M is strategy-proof, it follows that $(c,t) \mathrel{\tilde P_s} (c',t)$ or $c=c'$. By the construction of $\tilde P_s$, this relation implies $(s, c) \mathrel{P_s} (s,c')$ or $ (s, c)=(s,c')$, establishing strategy-proofness of TTC in the original problem.
\end{proof}
\medskip

\begin{proof}[Proof of Corollary \ref{corollary:TTC}]
Assume that $\sum_{t} \xi^t_c \leq q_c$  for every $\xi \in \Xi$ and $c\in \C$.
Under this presumption, we show that when the policy goal $\Xi$ is M-convex,
so is $\Xi \cap \Xi^0$. Then the result follows immediately from Theorem
\ref{thm:ttc} because the initial matching satisfies $\Xi$.

Suppose that $\xi,\tilde{\xi} \in \Xi \cap \Xi^0$ such that $\xi^t_c>\tilde{\xi}^t_c$ for some school $c$ and type $t$. Since $\Xi$ is M-convex and $\xi,\tilde{\xi} \in \Xi$, there exist school $c'$ and type $t'$ with $\xi^{t'}_{c'}<\tilde{\xi}^{t'}_{c'}$
such that $\xii \equiv \xi-\chi_{c,t}+ \chi_{c',t'}\in \Xi$ and $\xij \equiv \tilde{\xi}+\chi_{c,t}-\chi_{c',t'} \in \Xi$. We will  show that $\xii,\xij \in \Xi^0$. 

Because $\xi \in \Xi^0$,
  $\sum_{\c \in \C,\t \in \T}\xii^{\t}_{\c}=\sum_{\c \in \C,\t \in \T}\xi^{\t}_{\c}=\sum_d k_d$. Furthermore, for every $\c,\t$, by definition of $\xii$, we have $\xii^{\t}_{\c} \le \max\{\xi^{\t}_{\c},\tilde \xi^{\t}_{\c}\} \le q^{\t}_{\c}$. These two properties imply that $\xii \in \Xi^0$. A similar argument shows $\xij \in \Xi^0$.
\end{proof}
\medskip

\begin{proof}[Proof of Lemma \ref{lem:f-diversity}]
Note that $\Xi(f,\lambda) \cap \Xi^0 = \{ \xi \in \Xi_0 | f(\xi) \geq \lambda \}$.

\noindent{\textbf{The ``if'' direction:}} Suppose that $\xi \in \Xi(f,\lambda) \cap \Xi^0$ and $\tilde{\xi} \in \Xi(f,\lambda) \cap \Xi^0$ are distinct. Therefore, $f(\xi),f(\tilde{\xi}) \geq \lambda$.
By assumption, there exist $(c,t)$ and $(c',t')$ with $\xi_c^t>\tilde{\xi}_c^t$ and $\xi_{c'}^{t'}<\tilde{\xi}_{c'}^{t'}$ such that
  \begin{center}
  $\min \{f(\xi-\chi_{c,t}+\chi_{c',t'}), f(\tilde{\xi}+\chi_{c,t}-\chi_{c',t'})\} \geq \min \{f(\xi),f(\tilde{\xi})\}$.
  \end{center}

This implies $f(\xi-\chi_{c,t}+\chi_{c',t'}),f(\tilde{\xi}+\chi_{c,t}-\chi_{c',t'}) \geq \lambda$. Furthermore,
$\xi-\chi_{c,t}+\chi_{c',t'}, \tilde{\xi}+\chi_{c,t}-\chi_{c',t'} \in \Xi^0$ since the sum of coordinates
is equal to $\sum_d k_d$ and no school is assigned more students than its capacity.
Therefore, $\xi-\chi_{c,t}+\chi_{c',t'},\tilde{\xi}+\chi_{c,t}-\chi_{c',t'}\in \Xi(f,\lambda) \cap \Xi^0$, so, by Theorem 4.3. of Murota (2003), $\Xi(f,\lambda) \cap \Xi^0$ is M-convex.

\noindent{\textbf{The ``only if'' direction:}} Suppose that the function $f$ is not pseudo M-concave, so that there exist distinct
$\xi,\tilde{\xi} \in \Xi_0$ such that for all $(c,t)$ and $(c',t')$ with $\xi_c^t>\tilde{\xi}_c^t$ and $\xi_{c'}^{t'}<\tilde{\xi}_{c'}^{t'}$ we have
  \begin{center}
  $\min \{f(\xi-\chi_{c,t}+\chi_{c',t'}), f(\tilde{\xi}+\chi_{c,t}-\chi_{c',t'})\} < \min \{f(\xi),f(\tilde{\xi})\}$.
  \end{center}
Let $\lambda \equiv \min\{f(\xi),f(\tilde{\xi})\}$. The above condition implies that
$\Xi(f,\lambda) \cap \Xi^0$ is not M-convex.
\end{proof}
\medskip

\begin{proof}[Proof of Theorem \ref{thm:ttcequ}]
Let $f(\xi)=1$ when $\xi \in \Xi \cap \Xi^0$ and $f(\xi)=0$ otherwise.

First we show that $f$ is pseudo M-concave. Take two distinct $\xi,\tilde{\xi} \in \Xi_0$. If
$\min \{f(\xi),f(\tilde{\xi})\}=0$, then $\min \{f(\xi-\chi_{c,t}+\chi_{c',t'}), f(\tilde{\xi}+\chi_{c,t}-\chi_{c',t'})\}\geq 0$
for every $(c,t)$ and $(c',t')$, so the desired inequality holds. Suppose that $f(\xi)=f(\tilde{\xi})=1$. By
the construction of $f$, we have $\xi,\tilde{\xi} \in \Xi \cap \Xi^0$. Since $\Xi \cap \Xi^0$ is M-convex, there
exist $(c,t)$ and $(c',t')$ such that $\xi-\chi_{c,t}+\chi_{c',t'}, \tilde{\xi}+\chi_{c,t}-\chi_{c',t'} \in \Xi \cap \Xi^0$.
By the construction of $f$, $f(\xi-\chi_{c,t}+\chi_{c',t'})=f(\tilde{\xi}+\chi_{c,t}-\chi_{c',t'})=1$. Therefore, the
desired inequality also holds for this case, so $f$ is pseudo M-concave.

Next we show that $\Xi(f,\lambda) \cap \Xi^0 = \Xi \cap \Xi^0$ for $\lambda=1$. For any $\xi \in \Xi(f,1) \cap \Xi^0$,
$f(\xi)=1$, which implies that $\xi \in \Xi \cap \Xi^0$ by the construction of $f$. Therefore,
$\Xi(f,1) \cap \Xi^0 \subseteq \Xi \cap \Xi^0$. Now, let $\xi \in \Xi \cap \Xi^0$. Then, by the construction of $f$,
$f(\xi)=1$, so $\xi \in \Xi(f,1) \cap \Xi^0$. Therefore, $\Xi \cap \Xi^0 \subseteq \Xi(f,1) \cap \Xi^0$. We conclude
that $\Xi(f,1) \cap \Xi^0 = \Xi \cap \Xi^0$.
\end{proof}
\medskip

\begin{proof}[Proof of Corollary \ref{corollary:convexdiv}]
Suppose that $\Xi$ is a school-level diversity policy. We will first show that $\Xi \cap \Xi^0$ is an M-convex set. Recall that $\Xi=\{\xi | \forall c,t \text{ } q_c^t \geq \xi_c^t \geq p_c^t \}$ and $\Xi^0=\{\xi | \sum_{c,t}{\xi_c^t}=\sum_d k_d \text{ and } \forall c \text{ } q_c \geq \sum_t \xi_c^t\}$.

Suppose that there exist $\xi,\tilde{\xi}\in \Xi \cap \Xi^0$ such that $\xi_c^t>\tilde{\xi}_{c}^{t}$. To show M-convexity, we will  find school $c'$ and type $t'$ with $\xi_{c'}^{t'}<\tilde{\xi}_{c'}^{t'}$ such that (1) $\xii  \equiv \xi-\chi_{c,t}+ \chi_{c',t'}\in \Xi \cap \Xi^0$ and (2) $\xij \equiv \tilde{\xi}+\chi_{c,t}-\chi_{c',t'} \in \Xi \cap \Xi^0$.   
To show both conditions, we look at two possible cases depending on whether $c'=c$ or not.


\textbf{Case 1:}  First consider the case in which there exists type $t'$ such that $\xi_{c}^{t'}<\tilde{\xi}_{c}^{t'}$. We prove (1) for $c'=c$. First, by definition of $\xii$,  we have
$\sum_{\t \in \T, \c \in \C}\xii^{\t}_{\c}=\sum_{\t \in \T, \c \in \C}\xi^{\t}_{\c}=\sum_d k_d$.
Next, since $\sum_{\t \in \T}\xii^{\t}_c=\sum_{\t \in \T}\xi^{\t}_c$, we have $\sum_{\t \in \T}\xii^{\t}_c \le q_c$.
Therefore, $\xii \in \Xi^0$.

Next, we have
$\xii^t_c=\xi_c^t-1 \ge \tilde{\xi}_{c}^{t} \ge p_c^t$ (the equality comes from the definition of $\xii$, the first inequality comes from the assumption $\xi_c^t>\tilde{\xi}_{c}^{t}$, and the second inequality
comes from the assumption  $\tilde{\xi} \in \Xi$),  and $\xii^t_c =\xi_c^t-1 < \xi^t_c \le q_c^t$ (the equality comes from the definition of $\xii$, the first inequality is obvious, and the second inequality comes from the assumption
$\xi \in \Xi$).
Moreover, we have $\xii^{t'}_c =\xi_{c}^{t'}+1 > \xi_{c}^{t'} \ge p_c^{t'}$ (the equality comes from the definition of $\xii$, the first inequality is obvious, and the second inequality comes from the assumption $\xi \in \Xi$), and
$\xii_c^{t'} = \xi_{c}^{t'}+1 \le \tilde{\xi}_{c}^{t'} \le q_c^{t'}$
(the equality comes from the definition of $\xii$, the first inequality comes from the assumption $\xi^{t'}_c<\tilde \xi^{t'}_c$, and the second inequality comes from the assumption $\tilde \xi \in \Xi$). For any $\c,\t$ with $(\c,\t) \not\in \{(c,t),(c,t')\},$ we have $\xii^{\t}_{\c}=\xi^{\t}_{\c}$ by definition of $\xii$, so $p^{\t}_{\c} \le \xii^{\t}_{\c} \le q^{\t}_{\c}$.
 Therefore, $\xii \in \Xi$ and hence we conclude (1).

The proof that (1) is satisfied follows from the facts that $\xi_c^t>\tilde{\xi}_{c}^{t}$ and $\xi_{c}^{t'}<\tilde{\xi}_{c}^{t'}$. By changing the roles of $t$ with $t'$ and $\xi$ with $\tilde{\xi}$ in the preceding argument, we get the implication of (1) that $\xij \in \Xi \cap \Xi^0$. But this is exactly (2).

\textbf{Case 2:} Second, consider the case in which there exists no type $t'$ such that $\xi_{c}^{t'}<\tilde{\xi}_{c}^{t'}$. Then, $\xi_{c}^{t'}\geq \tilde{\xi}_{c}^{t'}$ for every $t'\neq t$. This in particular implies $\sum_{\t \in \T} \xi^{\t}_c>\sum_{\t \in \T} \tilde \xi^{\t}_c$. Because
$\sum_{\t \in \T, \c \in \C} \xi^{\t}_{\c}=\sum_{\t \in \T,\c \in \C} \tilde \xi^{\t}_{\c}$ by the assumption that $\xi, \tilde \xi \in \Xi^0$,
there exists a school $c' \neq c$ such that
$\sum_{\t \in \T} \xi^{\t}_{c'}<\sum_{\t \in \T} \tilde \xi^{\t}_{c'}$. In particular, there exists a type $t'$ such that  $\tilde{\xi}_{c'}^{t'}>\xi_{c'}^{t'}$.

Now we proceed to show condition (1) for this case. To do so first note that, by definition of $\xii$, we have
$\sum_{\t \in \T, \c \in \C}\xii^{\t}_{\c}=\sum_{\t \in \T, \c \in \C}\xi^{\t}_{\c}$. In addition,  the relation $\sum_{\t \in \T} \xi^{\t}_c>\sum_{\t \in \T} \tilde \xi^{\t}_c=\sum_{\t \in \T} \xii^{\t}_c$ and the assumption $\xi \in \Xi$ imply that $\sum_{\t \in \T} \xii^{\t}_c \le q_c$. Likewise,
$\sum_{\t \in \T} \xii^{\t}_{c'} = \sum_{\t \in \T} \xi^{\t}_{c'}+1 \le \sum_{\t \in \T} \tilde \xi^{\t}_{c'} \le q_{c'}$. Finally, for any $\c,\t$ with $(\c,\t) \not\in \{(c,t),(c',t')\},$ we have $\xii^{\t}_{\c}=\xi^{\t}_{\c}$ by definition of $\xii$, so $\sum_{\t \in \T} \xii^{\t}_{\c} \le q_{\c}$ for every $\c \neq c, c'$.
Thus, $\xii \in \Xi^0$.

Next, $\xii_c^t = \xi_c^t-1 \ge \tilde{\xi}_{c}^{t} \ge p_c^t$ (the first inequality follows from the assumption $\xi_c^t>\tilde{\xi}_{c}^{t}$ and the second from $\tilde{\xi} \in \Xi$), and
$
\xii_c^t=\xi_c^t-1<\xi_c^t  \le q_c^t
$
(the first inequality is obvious and the second inequality follows from $\xi \in \Xi$).
Moreover,
 $\xii_{c'}^{t'}= \xi_{c'}^{t'}+1 > \xi_{c'}^{t'} \ge p_{c'}^{t'}$ (the first inequality is obvious and the second follows from  $\xi \in \Xi$), and $\xii_{c'}^{t'}=\xi_{c'}^{t'}+1 \le \tilde{\xi}_{c'}^{t'} \le q_{c'}^{t'}$ (the first inequality follows from $\xi_{c'}^{t'}<\tilde{\xi}_{c'}^{t'}$ and the second follows from $\tilde{\xi} \in \Xi$). For any $\c,\t$ with $(\c,\t) \not\in \{(c,t),(c',t')\},$ we have $\xii^{\t}_{\c}=\xi^{\t}_{\c}$ by definition of $\xii$, so $p^{\t}_{\c} \le \xii^{\t}_{\c} \le q^{\t}_{\c}$.
Therefore, $\xii \in \Xi$ and hence we conclude (1).

The proof that (1) is satisfied follows from the facts that $\xi_c^t>\tilde{\xi}_c^t$, $\tilde{\xi}_{c'}^{t'}>\xi_{c'}^{t'}$, there are more students assigned to school $c$ at $\xi$ than $\tilde{\xi}$, and there are more students assigned to school $c'$ at $\tilde{\xi}$ than $\xi$. If we change the roles of $\xi$ with $\tilde{\xi}$, $c$ with $c'$, and $t$ with $t'$, then (1) would imply $\xij \in \Xi \cap \Xi^0$. But this is exactly (2). Therefore, $\Xi\cap \Xi^0$ is an M-convex set.

The desired conclusion then follows from the fact that $\Xi\cap \Xi^0$ is an M-convex set and Theorem \ref{thm:ttc}.
\end{proof}
\medskip

\begin{proof}[Proof of Corollary \ref{cor:fdiverse}]
We show that $f$ is pseudo M-concave. Let $\xi,\tilde{\xi} \in \Xi^0$ be distinct. Then $U \equiv \{(c,t)|\xi^c_t>\tilde{\xi}^c_t\}$ is a nonempty set. Partition this set into three subsets: $U_1 \equiv \{(c,t)|\hat{\xi}^c_t \geq \xi^c_t>\tilde{\xi}^c_t\}$, $U_2 \equiv \{(c,t)| \xi^c_t>\hat{\xi}^c_t>\tilde{\xi}^c_t\}$,
and $U_3 \equiv \{(c,t)| \xi^c_t>\tilde{\xi}^c_t \geq \hat{\xi}^c_t\}$. Likewise $V \equiv \{(c',t')|\tilde{\xi}^{c'}_{t'}>\xi^{c'}_{t'}\}$ is
a nonempty set that can be partitioned into three subsets: $V_1 \equiv \{(c',t') | \hat{\xi}^{c'}_{t'} \geq \tilde{\xi}^{c'}_{t'} > \xi^{c'}_{t'}\}$,
$V_2 \equiv \{(c',t')| \tilde{\xi}^{c'}_{t'}>\hat{\xi}^{c'}_{t'}>\xi^{c'}_{t'}\}$, and $V_3 \equiv \{(c',t')| \tilde{\xi}^{c'}_{t'} > \xi^{c'}_{t'} \geq \hat{\xi}^{c'}_{t'}\}$. We consider several cases.

\noindent{\textbf{Case 1: $U_2$ is nonempty.}} There exists $(c,t)$ such that $\xi_t^c> \hat{\xi}_t^c >\tilde{\xi}_t^c$. Since $\xi,\tilde{\xi}\in \Xi_0$, there exists
$(c',t')$ such that $\xi_{t'}^{c'}<\tilde{\xi}_{t'}^{c'}$. In this case, $f(\xi-\chi_{c,t}+\chi_{c',t'})\geq f(\xi)$ and
$f(\tilde{\xi}+\chi_{c,t}-\chi_{c',t'}) \geq f(\tilde{\xi})$ by definition of $f$. Therefore,
\[
\min \{f(\xi-\chi_{c,t}+\chi_{c',t'}), f(\tilde{\xi}+\chi_{c,t}-\chi_{c',t'})\} \geq \min \{f(\xi),f(\tilde{\xi})\}.
\]

\noindent{\textbf{Case 2: $V_2$ is nonempty.}} There exists $(c',t')$ such that $\tilde{\xi}_{t'}^{c'} > \hat{\xi}_{t'}^{c'} > \xi_{t'}^{c'}$. The proof of this case is similar to the proof of Case 1.

\noindent{\textbf{Case 3: $U_1$ and $V_1$ are nonempty.}} There exist $(c,t)$ and $(c',t')$ such that $\hat{\xi}_t^c \geq \xi_t^c>  \tilde{\xi}_t^c$ and
$\hat{\xi}_{t'}^{c'} \geq \tilde{\xi}_{t'}^{c'} > \xi_{t'}^{c'}$. In this case, $f(\xi-\chi_{c,t}+\chi_{c',t'}) = f(\xi)$ and
$f(\tilde{\xi}+\chi_{c,t}-\chi_{c',t'}) = f(\tilde{\xi})$ by definition of $f$. Then,
\[
\min \{f(\xi-\chi_{c,t}+\chi_{c',t'}), f(\tilde{\xi}+\chi_{c,t}-\chi_{c',t'})\} = \min \{f(\xi),f(\tilde{\xi})\}.
\]

\noindent{\textbf{Case 4: $U_3$ and $V_3$ are nonempty.}} There exist $(c,t)$ and $(c',t')$ such that $\xi_t^c>  \tilde{\xi}_t^c \geq \hat{\xi}_t^c$ and
$\tilde{\xi}_{t'}^{c'} > \xi_{t'}^{c'} \geq \hat{\xi}_{t'}^{c'}$. The proof is similar to the proof of Case 3.

\noindent{\textbf{Case 5: $U=U_1$ and $V=V_3$.}} For every $(c,t)$ such that $\xi_t^c> \tilde{\xi}_t^c$, we have $\hat{\xi}_t^c \geq \xi_t^c  > \tilde{\xi}_t^c$,
and for every $(c',t')$ such that $\tilde{\xi}_{t'}^{c'} > \xi_{t'}^{c'}$, we have $\tilde{\xi}_{t'}^{c'} > \xi_{t'}^{c'} \geq \hat{\xi}_{t'}^{c'}$.
In this case, $f(\tilde{\xi})+2 \leq f(\xi)$, so $\min\{f(\tilde{\xi}), f(\xi)\}=f(\tilde{\xi})$. Furthermore, for any choice of $(c,t)$ and $(c',t')$
such that $\xi_c^t>\tilde{\xi}_c^t$ and $\xi_{c'}^{t'}<\tilde{\xi}_{c'}^{t'}$, we have $f(\xi-\chi_{c,t}+\chi_{c',t'})=f(\xi)-2$ and
$f(\tilde{\xi}+\chi_{c,t}-\chi_{c',t'})=f(\tilde{\xi})+2$.

Since $f(\xi)-2\geq f(\tilde{\xi})$ and $f(\tilde{\xi})+2>f(\tilde{\xi})$, we get the desired conclusion that
\[
\min \{f(\xi-\chi_{c,t}+\chi_{c',t'}), f(\tilde{\xi}+\chi_{c,t}-\chi_{c',t'})\} \geq f(\tilde{\xi}) = \min \{f(\xi),f(\tilde{\xi})\}.
\]

\noindent{\textbf{Case 6: $U=U_3$ and $V=V_1$.}} For every $(c,t)$ such that $\xi_t^c> \tilde{\xi}_t^c$, we have $\xi_t^c  > \tilde{\xi}_t^c \geq \hat{\xi}_t^c$,
and for every $(c',t')$ such that $\tilde{\xi}_{t'}^{c'} > \xi_{t'}^{c'}$, we have $\hat{\xi}_{t'}^{c'} \geq \tilde{\xi}_{t'}^{c'} > \xi_{t'}^{c'}$.
The proof is similar to the proof of Case 5.

We have considered all the possible cases: If $U_2$ or $V_2$ are nonempty, then we are done by Cases 1 and 2, respectively. Suppose that they are both empty. Therefore, $U=U_1 \cup U_3$ and $V=V_1 \cup V_3$ are both nonempty. If $U_1$ and $V_1$ are nonempty, then we are done by Case 3. If $U_3$ and $V_3$ are nonempty, then we are done by Case 4. If $U_1$ and $V_3$ are nonempty and one of $U_3$ or $V_1$ is nonempty, then we are done by Cases 3 or 4. Otherwise, if $U_3$ and $V_1$ are empty when $U_1$ and $V_3$ are nonempty, then $U=U_1$ and $V=V_3$, which is covered by Case 5.
If $U_3$ and $V_1$ are nonempty and one of $U_1$ or $V_3$ is nonempty, then we are done by Cases 3 or 4. Otherwise, if
$U_1$ and $V_3$ are empty when $U_3$ and $V_1$ are nonempty, then $U=U_3$ and $V=V_1$, which is covered by Case 6.

We conclude that $f$ is pseudo M-concave since in all possible cases we derive the desired inequality. Then the
proof follows from Theorem \ref{thm:ttcf}.
\end{proof}
\medskip

\begin{proof}[Proof of Corollary \ref{corollary:convexbal}]
Let the balanced-exchange policy be denoted by $\Xi$.
We first show that $\Xi\cap \Xi^0$ is M-convex.

Suppose that there exist $\xi,\tilde{\xi} \in \Xi \cap \Xi^0$ such that $\xi_c^t>\tilde{\xi}_{c}^{t}$. To show M-convexity, we need to
find school $c'$ and type $t'$ with $\xi_{c'}^{t'}<\tilde{\xi}_{c'}^{t'}$ such that (1) $\xii \equiv \xi-\chi_{c,t}+ \chi_{c',t'}\in \Xi \cap \Xi^0$ and
(2) $\xij \equiv \tilde{\xi}+\chi_{c,t}-\chi_{c',t'} \in \Xi \cap \Xi^0$.

If there exists $t'$ such that $\tilde{\xi}_{c}^{t'}>\xi_{c}^{t'}$, then
$\sum_{\t \in \T} \xii_{d}^{\t}=\sum_{\t \in \T} \xij_{d}^{\t}=\sum_{\t \in \T} \xi_{d}^{\t} =k_d$ for every $d$  and  $\sum_{\t \in \T} \xii_c^{\t}=\sum_{\t \in \T} \xij_c^{\t}=\sum_{\t \in \T} \xi_c^{\t} \le q_c$ for every $c \in C$,  so both (1) and (2) are satisfied.

Now suppose $\tilde{\xi}_{c}^{t'}\leq \xi_{c}^{t'}$ for every type $t'\neq t$. Therefore, $\sum_{\t \in \T} \tilde{\xi}_{c}^{\t} < \sum_{\t \in \T} \xi_{c}^{\t}$. Because $\sum_{\t \in \T, \c: d(\c)=d} \tilde{\xi}_{\c}^{\t} = \sum_{\t \in \T,\c: d(\c)=d} \xi_{\c}^{\t}$, where $d\equiv d(c)$,
there exists another school $c'$ in district $d$ such that $\sum_{\t \in \T} \tilde{\xi}_{c'}^{\t} > \sum_{\t \in \T} \xi_{c'}^{\t}$. In particular, there exists a type $t'$ such that  $\tilde{\xi}_{c'}^{t'}>\xi_{c'}^{t'}$.

We first show (1). To do so, first note that since both schools $c$ and $c'$ are in district $d$,
$\sum_{\t \in \T, \c: d(\c)=d} \xii_{\c}^{\t} = \sum_{\t \in \T,\c: d(\c)=d} \xi_{\c}^{\t}=k_d$. Moreover, for any $\d \neq d$, $\sum_{\t \in \T, \c: d(\c)=\d} \xii_{\c}^{\t} = \sum_{\t \in \T,\c: d(\c)=\d} \xi_{\c}^{\t}=k_{\d}$ because $\xii^{\t}_{\c}=\xi^{\t}_{\c}$ for any $\t$ and $\c$ with $d(\c) = \d$ by definition of $\xii$. Thus, $\xii \in \Xi$.
Next we show $\xii \in \Xi^0$. To do so, first observe that $\sum_{\t \in \T} \tilde{\xi}_{c}^{\t} = \sum_{\t \in \T} \xi_{c}^{\t}-1 < q_c$.
Moreover, $ \sum_{\t \in \T} \xii_{c'}^{\t}= \sum_{\t \in \T} \xi_{c'}^{\t}+1 \le
\sum_{\t \in \T} \tilde{\xi}_{c'}^{\t} \le q_{c'}.$ Furthermore, for any $\c \neq c, c'$,
$ \sum_{\t \in \T} \xii_{\c}^{\t}= \sum_{\t \in \T} \xi_{\c}^{\t} \le q_{\c}$. Therefore, $\xii \in \Xi^0$ and hence (1) holds.

Note that the above argument relies on the facts $\xi_c^t>\tilde{\xi}_c^t$, $\xi_{c'}^{t'}<\tilde{\xi}_{c'}^{t'}$, and $d(c)=d(c')$. If we switch the roles of $c$ with $c'$ and $\xi$ with $\tilde{\xi}$, the implication of (1) is  (2).

The result then follows from Theorem \ref{thm:ttc} because $\Xi \cap \Xi^0$ is M-convex and the initial matching trivially satisfies the balanced-exchange policy.
\end{proof}
\medskip

\begin{proof}[Proof of Corollary \ref{corollary:mix}]

The proof is very similar to those of Corollary \ref{corollary:convexdiv} and Corollary \ref{corollary:convexbal}. We first show that $\Xi \cap \Xi^0$ is an M-convex set. Recall that $\Xi=\{\xi | \forall c,t \text{ } q_c^t \geq \xi_c^t \geq p_c^t \text{ and } \forall d \text{ } \sum_t \xi_d^t =k_d \}$ and $\Xi^0=\{\xi | \sum_{c,t}{\xi_c^t}=\sum_d k_d \text{ and } \forall c \text{ } q_c \geq \sum_t \xi_c^t\}$.

Suppose that there exist $\xi,\tilde{\xi}\in \Xi \cap \Xi^0$ such that $\xi_c^t>\tilde{\xi}_c^t$. To show M-convexity, we need to find school $c'$ and type $t'$ with $\xi_{c'}^{t'}<\tilde{\xi}_{c'}^{t'}$ such that (1) $\xii \equiv \xi-\chi_{c,t}+ \chi_{c',t'}\in \Xi \cap \Xi^0$ and (2) $\xij \equiv \tilde{\xi}+\chi_{c,t}-\chi_{c',t'} \in \Xi \cap \Xi^0$. Let $d\equiv d(c)$. To show both conditions, we look at two possible cases depending on whether $c'=c$ or not.

\textbf{Case 1:}  First consider the case in which there exists type $t'$ such that $\xi_{c}^{t'}<\tilde{\xi}_{c}^{t'}$. We prove (1) for $c'=c$.
First, by definition of $\xii$,  we have
$\sum_{\t \in \T, \c \in \C}\xii^{\t}_{\c}=\sum_{\t \in \T, \c \in \C}\xi^{\t}_{\c}=\sum_d k_d$.
Next, since $\sum_{\t \in \T}\xii^{\t}_c=\sum_{\t \in \T}\xi^{\t}_c$, we have $\sum_{\t \in \T}\xii^{\t}_c \le q_c$.
Therefore, $\xii \in \Xi^0$.

Next, we have
$\xii^t_c=\xi_c^t-1 \ge \tilde{\xi}_{c}^{t} \ge p_c^t$ (the equality comes from the definition of $\xii$, the first inequality comes from the assumption $\xi_c^t>\tilde{\xi}_{c}^{t}$, and the second inequality
comes from the assumption  $\tilde{\xi} \in \Xi$),  and $\xii^t_c =\xi_c^t-1 < \xi^t_c \le q_c^t$ (the equality comes from the definition of $\xii$, the first inequality is obvious, and the second inequality comes from the assumption
$\xi \in \Xi$).
Moreover, we have $\xii^{t'}_c =\xi_{c}^{t'}+1 > \xi_{c}^{t'} \ge p_c^{t'}$ (the equality comes from the definition of $\xii$, the first inequality is obvious, and the second inequality comes from the assumption $\xi \in \Xi$), and
$\xii_c^{t'} = \xi_{c}^{t'}+1 \le \tilde{\xi}_{c}^{t'} \le q_c^{t'}$
(the equality comes from the definition of $\xii$, the first inequality comes from the assumption $\xi^{t'}_c<\tilde \xi^{t'}_c$, and the second inequality comes from the assumption $\tilde \xi \in \Xi$). For any $\c,\t$ with $(\c,\t) \not\in \{(c,t),(c,t')\},$ we have $\xii^{\t}_{\c}=\xi^{\t}_{\c}$ by definition of $\xii$, so $p^{\t}_{\c} \le \xii^{\t}_{\c} \le q^{\t}_{\c}$. Finally, $\sum_{\t \in \T} \xii_{d}^{\t}=\sum_{\t \in \T} \xij_{d}^{\t}=\sum_{\t \in \T} \xi_{d}^{\t} =k_d$ for every $d$.
 Therefore, $\xii \in \Xi$ and hence we conclude (1).

The proof that (1) is satisfied follows from the facts that $\xi_c^t>\tilde{\xi}_{c}^{t}$ and $\xi_{c}^{t'}<\tilde{\xi}_{c}^{t'}$. By changing the roles of $t$ with $t'$ and $\xi$ with $\tilde{\xi}$ in the preceding argument, we get the implication of (1) that $\xij \in \Xi \cap \Xi^0$. But this is exactly (2).

\textbf{Case 2:} Second, consider the case in which there exists no type $t'$ such that $\xi_{c}^{t'}<\tilde{\xi}_{c}^{t'}$. Then, $\xi_{c}^{t'}\geq \tilde{\xi}_{c}^{t'}$ for every $t'\neq t$. This in particular implies $\sum_{\t \in \T} \xi^{\t}_c>\sum_{\t \in \T} \tilde \xi^{\t}_c$.
Because $\sum_{\t \in \T, \c: d(\c)=d} \tilde{\xi}_{\c}^{\t} = \sum_{\t \in \T,\c: d(\c)=d} \xi_{\c}^{\t}$, where $d\equiv d(c)$,
there exists another school $c'$ in district $d$ such that $\sum_{\t \in \T} \tilde{\xi}_{c'}^{\t} > \sum_{\t \in \T} \xi_{c'}^{\t}$.
 In particular, there exists a type $t'$ such that  $\tilde{\xi}_{c'}^{t'}>\xi_{c'}^{t'}$.

Now we proceed to show condition (1) for this case. To do so first note that, by definition of $\xii$, we have
$\sum_{\t \in \T, \c \in \C}\xii^{\t}_{\c}=\sum_{\t \in \T, \c \in \C}\xi^{\t}_{\c}$. In addition,  the relation $\sum_{\t \in \T} \xi^{\t}_c>\sum_{\t \in \T} \tilde \xi^{\t}_c=\sum_{\t \in \T} \xii^{\t}_c$ and the assumption $\xi \in \Xi$ imply that $\sum_{\t \in \T} \xii^{\t}_c \le q_c$. Likewise,
$\sum_{\t \in \T} \xii^{\t}_{c'} = \sum_{\t \in \T} \xi^{\t}_{c'}+1 \le \sum_{\t \in \T} \tilde \xi^{\t}_{c'} \le q_{c'}$. Finally, for any $\c,\t$ with $(\c,\t) \not\in \{(c,t),(c',t')\},$ we have $\xii^{\t}_{\c}=\xi^{\t}_{\c}$ by definition of $\xii$, so $\sum_{\t \in \T} \xii^{\t}_{\c} \le q_{\c}$ for every $\c \neq c, c'$.
Thus, $\xii \in \Xi^0$.

Next, $\xii_c^t = \xi_c^t-1 \ge \tilde{\xi}_{c}^{t} \ge p_c^t$ (the first inequality follows from the assumption $\xi_c^t>\tilde{\xi}_{c}^{t}$ and the second from $\tilde{\xi} \in \Xi$), and
$
\xii_c^t=\xi_c^t-1<\xi_c^t  \le q_c^t
$
(the first inequality is obvious and the second inequality follows from $\xi \in \Xi$).
Moreover,
 $\xii_{c'}^{t'}= \xi_{c'}^{t'}+1 > \xi_{c'}^{t'} \ge p_{c'}^{t'}$ (the first inequality is obvious and the second follows from  $\xi \in \Xi$), and $\xii_{c'}^{t'}=\xi_{c'}^{t'}+1 \le \tilde{\xi}_{c'}^{t'} \le q_{c'}^{t'}$ (the first inequality follows from $\xi_{c'}^{t'}<\tilde{\xi}_{c'}^{t'}$ and the second follows from $\tilde{\xi} \in \Xi$). For any $\c,\t$ with $(\c,\t) \not\in \{(c,t),(c',t')\},$ we have $\xii^{\t}_{\c}=\xi^{\t}_{\c}$ by definition of $\xii$, so $p^{\t}_{\c} \le \xii^{\t}_{\c} \le q^{\t}_{\c}$.
Next, note that since both schools $c$ and $c'$ are in district $d$,
$\sum_{\t \in \T, \c: d(\c)=d} \xii_{\c}^{\t} = \sum_{\t \in \T,\c: d(\c)=d} \xi_{\c}^{\t}=k_d$. Moreover, for any $\d \neq d$, $\sum_{\t \in \T, \c: d(\c)=\d} \xii_{\c}^{\t} = \sum_{\t \in \T,\c: d(\c)=\d} \xi_{\c}^{\t}=k_{\d}$ because $\xii^{\t}_{\c}=\xi^{\t}_{\c}$ for any $\t$ and $\c$ with $d(\c) = \d$ by definition of $\xii$.
Therefore, $\xii \in \Xi$ and hence we conclude (1).

The proof that (1) is satisfied follows from the facts that $\xi_c^t>\tilde{\xi}_c^t$, $\tilde{\xi}_{c'}^{t'}>\xi_{c'}^{t'}$, there are more students assigned to school $c$ at $\xi$ than $\tilde{\xi}$, and there are more students assigned to school $c'$ at $\tilde{\xi}$ than $\xi$. If we change the roles of $\xi$ with $\tilde{\xi}$, $c$ with $c'$, and $t$ with $t'$, then (1) would imply $\xij \in \Xi \cap \Xi^0$. But this is exactly (2). Therefore, $\Xi\cap \Xi^0$ is an M-convex set.

The result then follows from Theorem \ref{thm:ttc} because $\Xi \cap \Xi^0$ is M-convex.
\end{proof}
\medskip

\begin{proof}[Proof of Theorem \ref{thm:welfimprove}]
Suppose that district admissions rules favor own students. Fix a student preference profile. Recall that under
interdistrict school choice, students are assigned to schools by SPDA, where each student ranks all contracts associated
with her and each district $d$ has the admissions rule $\chd$. Under intradistrict school choice, students are assigned
to schools by SPDA where students only rank the contracts associated with their home districts and each district $d$
has the admissions rule $\chd$. We first show that the intradistrict SPDA outcome can be produced by SPDA when all
districts participate simultaneously and students rank all contracts, including the ones associated with the other
districts, by modifying admissions rules for the districts. Let $\chd'(X) \equiv \chd(\{x\in X|d(s(x))=d\})$ be the
modified admissions rule.

In SPDA, if district admissions rules have completions that satisfy path independence, then
SPDA outcomes are the same under the completions and the original admissions rules because in
SPDA a district always considers a set of proposals which is feasible for students. Furthermore,
SPDA does not depend on the order of proposals when district admissions rules are path independent.
As a result, SPDA does not depend on the order of proposals when district admissions rules have
completions that satisfy path independence. Therefore, the intradistrict SPDA outcome can be
produced by SPDA when all districts participate simultaneously and students rank all contracts
including the ones associated with the other districts and each district $d$ has the admissions
rule $\chd'$. The reason behind this is that when each district $d$ has admissions rule $\chd'$,
a student is not admitted to a school district other than her home district. Furthermore, because
$\chd$ favors own students, the set of chosen students under $\chd'$ is the same as that under
$\chd$ for any set of contracts of the form $\{x\in X|d(s(x))=d\}$ for any set $X$.

We next show that $\chd'$ has a path-independent completion. By assumption, for every district $d$,
there exists a path-independent completion $\widetilde{Ch}_d$ of $\chd$. Let $\widetilde{Ch}_d'(X) \equiv \widetilde{Ch}_d(\{x\in X|d(s(x))=d\})$. We show that $\widetilde{Ch}_d'$ is a path-independent completion of $Ch'_d$. To show that $\widetilde{Ch}_d'(X)$ is a completion, consider a set $X$ such that $\widetilde{Ch}_d'(X)$ is feasible for students. Let $X^*\equiv \{x\in X|d(s(x))=d\}$. Then we have the following:
\begin{equation*}
  \widetilde{Ch}_d'(X^*)=\widetilde{Ch}_d(X^*)=\chd(X^*)=Ch'_d(X^*),
\end{equation*}
where the first equality follows from the definition of $\widetilde{Ch}'_d$, the second equality follows from the fact that $\widetilde{Ch}_d$ is a completion of $\chd$, and the third equality follows from the definition of $Ch'_d$. Furthermore, because
$\widetilde{Ch}_d'(X)=\widetilde{Ch}_d'(X^*)$ and $Ch'_d(X^*)=Ch'_d(X)$, we get $\widetilde{Ch}_d'(X)=Ch'_d(X)$. Therefore, $\widetilde{Ch}_d'$ is a completion of $Ch'_d$.

To show that $\widetilde{Ch}_d'$ is path independent, consider two sets of contracts $X$ and $Y$. Let $X^*\equiv \{x\in X|d(s(x))=d\}$ and $Y^*\equiv \{x\in Y|d(s(x))=d\}$. Then we have the following:

\begin{alignat*}{2}
 \widetilde{Ch}_d'(X \cup \widetilde{Ch}_d'(Y))  &= \widetilde{Ch}_d'(X \cup \widetilde{Ch}_d(Y^*)) \\
 &=  \widetilde{Ch}_d(X^* \cup \widetilde{Ch}_d(Y^*)) \\
 & =  \widetilde{Ch}_d(X^* \cup Y^*) \\ & =  \widetilde{Ch}_d'(X \cup Y),
\end{alignat*}
where the first and second equalities follow from the definition of $\widetilde{Ch}_d'$, the third equality follows from path independence of $\widetilde{Ch}_d$, and the last equality follows from the definition of $\widetilde{Ch}_d'$. Therefore, $\widetilde{Ch}_d'$ is path independent.

Because $\chd$ favors own students, we have $\chd(X)\supseteq \chd'(X)$ for every $X$ that is feasible for students. Furthermore, for any such $X$, $\wt{Ch}_d(X)=\chd(X)$ and $\wt{Ch}'_d(X)=\chd'(X)$ because $\wt{Ch}_d$ is a completion of $\chd$ and $\wt{Ch}'_d$ is a completion of $\chd'$, respectively. Therefore, for any $X$ that is feasible for students, $\wt{Ch}_d(X)\supseteq \wt{Ch}'_d(X)$. We use this result to show the following lemma.

\begin{lemma}\label{lem:comp}
Every student weakly prefers the interdistrict SPDA outcome under $(\widetilde{Ch}_d)_{d\in \D}$ to the interdistrict SPDA outcome under $(\widetilde{Ch}'_d)_{d\in \D}$.
\end{lemma}

\begin{proof}
Let $\mu$ be the interdistrict SPDA outcome under $(\widetilde{Ch}_d)_{d\in \D}$ and $\mu'$ be the interdistrict SPDA outcome under $(\widetilde{Ch}'_d)_{d\in \D}$. If $\mu'$ is stable under $(\widetilde{Ch}_d)_{d\in \D}$, then the conclusion follows from the result that $\mu$ is the student-optimal stable matching under $(\widetilde{Ch}_d)_{d\in \D}$ because each $\widetilde{Ch}_d$ is path independent \citep{chayen13}.

Suppose that $\mu'$ is not stable under $(\widetilde{Ch}_d)_{d\in \D}$. Since $\mu'$ is stable under $(\widetilde{Ch}'_d)_{d\in \D}$, $\widetilde{Ch}'_d(\mu'_d)=\mu'_d$ for every district $d$. Furthermore, $\mu'_d$ is feasible for students, so $\wt{Ch}_d(\mu'_d)\supseteq \wt{Ch}'_d(\mu'_d)=\mu'_d$. By definition of admissions rules, $\mu'_d \supseteq \wt{Ch}_d(\mu'_d)$, so $\widetilde{Ch}_d(\mu'_d)=\mu'_d$. As a result, there must exist a blocking contract for matching $\mu'$ so that it is not stable under $(\widetilde{Ch}_d)_{d\in \D}$. Whenever there exists a blocking pair, we consider the following algorithm to improve student welfare. Let $d_1$ be a district associated with a blocking contract. Set $\mu^0 \equiv \mu'$.

\begin{description}
  \item[Step $\mathbf{n}$ ($\mathbf{n\geq 1}$)] Consider the following set of contracts associated with a district $d_n$ for which there exists an associated blocking contract: $X_{d_n}^n\equiv \{x=(s,d_n,c)| x \mathrel{P_s} \mu^{n-1}_s\}$. District $d_n$ accepts $\wt{Ch}_{d_n}(\mu^{n-1}_d \cup X_{d_n}^n)$ and rejects the rest of the contracts. Let $\mu^n_{d_n}\equiv \wt{Ch}_{d_n}(\mu^{n-1}_{d_n} \cup X_{d_n}^n)$ and $\mu^n_d \equiv \mu^{n-1}_d \setminus Y^n$ where $Y^n\equiv \{x\in \mu^{n-1}|\exists y \in \mu^n_{d_n} \text { s.t. } s(x)=s(y)\}$ for $d\neq d_n$. If there are no blocking contracts for matching $\mu^n$ under $(\widetilde{Ch}_d)_{d\in \D}$, then stop and return $\mu^n$, otherwise go to Step $n+1$.
\end{description}

We show that district $d_n$ does not reject any contract in $\mu_{d_n}^{n-1}$ by mathematical induction on $n$, i.e., $\mu^n_{d_n}\supseteq \mu^{n-1}_{d_n}$ for every $n\geq 1$. Consider the base case for $n=1$. Recall that $\mu^1_{d_1} = \wt{Ch}_{d_1}(\mu^{0}_{d_1} \cup X_{d_1}^1)=\wt{Ch}_{d_1}(\mu'_{d_1} \cup X_{d_1}^1)$. By construction, $\mu^1_{d_1}$ is a feasible matching. We claim that $\mu'_{d_1} \cup \mu^1_{d_1}$ is feasible for students. Suppose, for contradiction, that it is not feasible for students. Then there exists a student $s$ who has one contract in $\mu'_{d_1}$ and one in $\mu^1_{d_1}\setminus \mu'_{d_1}$. Call the latter contract $z$. By construction, $z \mathrel{P_s} \mu'_s$, and by path independence, $z\in \wt{Ch}_{d_1}(\mu'_{d_1} \cup \{z\})$. Furthermore, since student $s$ is matched with district $d_1$ in $\mu'$, $d(s)=d_1$. Therefore, $\wt{Ch}_{d_1}(\mu'_{d_1} \cup \{z\})=\wt{Ch}'_{d_1}(\mu'_{d_1} \cup \{z\})$ by definition of $\wt{Ch}'_{d_1}$ and construction of $\mu'$. Hence, $z\in \wt{Ch}'_{d_1}(\mu'_{d_1} \cup \{z\})$, which contradicts the fact that $\mu'$ is stable under $(\wt{Ch}'_d)_{d\in \D}$. Hence, $\mu'_{d_1} \cup \mu^1_{d_1}$ is feasible for students. Feasibility for students implies that $\wt{Ch}_{d_1}(\mu'_{d_1} \cup \mu^1_{d_1})\supseteq \wt{Ch}'_{d_1}(\mu'_{d_1} \cup \mu^1_{d_1})$. Path independence and construction of $\mu^1_{d_1}$ yield $\mu^1_{d_1}=\wt{Ch}_{d_1}(\mu'_{d_1}\cup \mu^1_{d_1})$. Furthermore, there exists no student $s$, such that $d(s)=d_1$ ,who has a contract in $\mu^1_{d_1}\setminus \mu'$, as this would contradict stability of $\mu'$ under $(\wt{Ch}'_d)_{d\in \D}$. This implies, by definition of $\wt{Ch}'_{d_1}$, that $\wt{Ch}'_{d_1}(\mu'_{d_1} \cup \mu^1_{d_1})=\wt{Ch}'_{d_1}(\mu'_{d_1})$, and, by stability of $\mu'$ under $(\wt{Ch}'_d)_{d\in \D}$, $\wt{Ch}'_{d_1}(\mu'_{d_1})=\mu'_{d_1}$. Therefore, $\mu^1_{d_1}=\wt{Ch}_{d_1}(\mu'_{d_1} \cup \mu^1_{d_1})\supseteq \wt{Ch}'_{d_1}(\mu'_{d_1} \cup \mu^1_{d_1})=\mu'_{d_1}=\mu^0_{d_1}$, which means that district $d_1$ does not reject any contracts.

Now consider district $d_n$ where $n>1$. There are two cases to consider. First consider the case when $d_n\neq d_i$ for every $i<n$. In this case, $\mu^{n-1}_{d_n}\subseteq \mu^0_{d_n}=\mu'_{d_n}$. We repeat the same arguments as in the previous paragraph. Stability of $\mu'$ under $(\wt{Ch}'_d)_{d\in \D}$ and path independence of $\wt{Ch}'_{d_n}$ implies that $\mu^n_{d_n}\cup \mu^{n-1}_{d_n}$ is feasible for students. Therefore, $\wt{Ch}_{d_n}(\mu^{n-1}_{d_n} \cup \mu^n_{d_n})\supseteq \wt{Ch}'_{d_n}(\mu^{n-1}_{d_n} \cup \mu^n_{d_n})$. Furthermore, there exists no student $s$, such that $d(s)=d_n$, who has a contract in $\mu^n_{d_n}\setminus \mu^{n-1}_{d_n}$. As a result, by definition of $\wt{Ch}'_{d_n}$ and by path independence, $\wt{Ch}'_{d_n}(\mu^{n-1}_{d_n} \cup \mu^n_{d_n})=\wt{Ch}'_{d_n}(\mu^{n-1}_{d_n})=\mu^{n-1}_{d_n}$. As in the previous paragraph, we conclude that $\mu^n_{d_n}=\wt{Ch}_{d_n}(\mu^{n-1}_{d_n} \cup \mu^n_{d_n})\supseteq \wt{Ch}'_{d_n}(\mu^{n-1}_{d_n} \cup \mu^n_{d_n})=\mu^{n-1}_{d_n}$.

The second case is when there exists $i<n$ such that $d_i=d_n$. Let $i^*$ be the last such step before $n$. Since student welfare improves at every step before $n$ by the mathematical induction hypothesis, $\mu_{d_n}^{i^*-1} \cup X_{d_n}^{i^*}\supseteq \mu_{d_n}^{n-1} \cup X_{d_n}^{n}$. By definition, $\mu_{d_n}^{i^*}=\wt{Ch}_{d_n}(\mu_{d_n}^{i^*-1} \cup X_{d_n}^{i^*})$, which implies by path independence that $\mu_{d_n}^{n-1}\subseteq \wt{Ch}_{d_n}(\mu_{d_n}^{n-1} \cup X_{d_n}^{n})=\mu_{d_n}^n$ since $\mu_{d_n}^{n-1}\subseteq \mu_{d_n}^{i^*}$.

Finally, we need to show that the improvement algorithm terminates. We claim that $\mu^n_{d_n}\neq \mu^{n-1}_{d_n}$. Suppose, for contradiction, that these two matchings are the same. Then, by path independence of $\wt{Ch}_{d_n}$, for every $x\in X^n_{d_n}$, $\wt{Ch}_{d_n}(\mu^{n-1}_{d_n}\cup \{x\})=\mu^{n-1}_{d_n}$. This is a contradiction because there exists at least one blocking contract associated with district $d_n$. Therefore, district $d_n$ gets at least one new contract at Step $n$. Hence, at least one student gets a strictly more preferred contract at every step of the algorithm while every other student gets a weakly more preferred contract. Since the number of contracts is finite, the algorithm has to end in a finite number of steps.

\end{proof}

Because the interdistrict SPDA outcome under $(\chd)_{d\in \D}$ is the same as the interdistrict SPDA outcome under $(\wt{Ch}_d)_{d\in \D}$ and the interdistrict SPDA outcome under $(\chd')_{d\in \D}$ is the same as the interdistrict SPDA outcome under $(\wt{Ch}')_{d\in \D}$, the lemma implies that every student weakly prefers the outcome of interdistrict SPDA under $(\chd)_{d\in \D}$ to the outcome of intradistrict SPDA (which is the same as the interdistrict SPDA outcome under $(\chd')_{d\in \D}$). This completes the proof of the first part.

To prove the second part of the theorem, we show that if at least one district's admissions rule fails to favor own students, then there exists a student  preference profile such that not every student is weakly better off under interdistrict SPDA than under
intradistrict SPDA. Suppose that for some district $d$, there exists a matching $X$, which is feasible for students, such that $Ch_d(X)$ is not a superset of $Ch_d(X^*)$, where $X^* \equiv \{x \in X|d(s(x))=d\}$. Now, consider a matching $Y$ where (i) all students from district $d$ are matched with schools in district $d$, (ii) $Y$ is feasible, and (iii) $Y\supseteq \chd(X^*)$. The existence of such a $Y$ follows from the fact that $\chd(X^*)$ is feasible and $k_{d'} \le \sum_{c:d(c)=d'} q_c$, for every district $d'$ (that is, there are enough seats in district $d'$ to match all students from district $d'$.) Because $Y$ is feasible and $\chd$ is acceptant, $\chd(Y_d)=Y_d$.

Now consider the following student preferences. First we consider students from district $d$. Each student $s$ who has a contract in $X^*$ ranks $X^*_s$ as her top choice. Note that doing so is well defined because $X^*$ is feasible for students. Each student $s$ who has a contract in $X^* \setminus \chd(X^*)$ ranks contract $Y_s$ as her second top choice. Note that, in this case, $Y_s$ cannot be the same as $X^*_s$ because $\chd(Y_d)=Y_d$ and $\chd$ is path independent. Each student $s$ who has a contract in $Y\setminus X^*$ ranks that contract as her top choice. Next we consider students from the other districts. Each student $s$ who has a contract in $X\setminus X^*$ ranks that contract as her top choice. Any other student ranks a contract not associated with district $d$ as her top choice. Complete the rest of the student preferences arbitrarily.

Consider SPDA for district $d$ in intradistrict school choice. At the first step, students who have a contract in $X^*$ propose that contract. The remaining students who have contracts in $Y\setminus X^*$ propose the associated contracts. Because $Y$ is feasible, $Y$ contains $\chd(X^*)$, and $\chd$ is acceptant, only contracts in $X^* \setminus \chd(X^*)$ are rejected. At the second step, these students propose their contracts in $Y_d$, and the set of proposals that the district considers is $Y_d$. Because $\chd(Y_d)=Y_d$, no contract is rejected, and SPDA stops and returns $Y_d$. In particular, every student who has a contract in $\chd(X^*)$ has the corresponding contract at the outcome.

In interdistrict SPDA, at the first step, each student who has a contract in $X$ proposes that contract and every other student proposes a contract associated with a district different from $d$. District $d$ considers $X$ (or $X_d$), and tentatively accepts $\chd(X)$. Because $\chd(X) \not \supseteq \chd(X^*)$ by assumption, at least one student who has a contract in $\chd(X^*)$ is rejected. Therefore, this student is strictly worse off under interdistrict school choice than under intradistrict school choice.
\end{proof}

\medskip

\begin{proof}[Proof of Theorem \ref{thm:imposswithschools}]

To show the result, we first introduce the following weakening of the substitutability condition \citep{hatfield2008matching}.
A district admissions rule $Ch_d$ satisfies \df{weak substitutability} if, for every
$x \in X\subseteq Y \subseteq \X$ with $x\in Ch_d(Y)$ and $|Y_s| \le 1$ for each $s \in \S$,
it must be that $x\in Ch_d(X)$.

Under weak substitutability, the following result is known (the statement is slightly modified for the present setting).

\begin{theorem}[\citet{hatfield2008matching}]
Let $d$ and $d'$ be two distinct districts. Suppose that $\chd$ satisfies IRC
but violates weak substitutability. Then, there exist student preferences and a path-independent admissions rule for $d'$ such that, regardless of the other districts' admissions rules, no stable matching exists.
\end{theorem}

Given this result, for our purposes it suffices to show the following.

\begin{customthm}{9'}\label{thm:imposswithschools'}
Let $d$ be a district. There exist a set of students, their types, schools in $d$, and
type-specific ceilings for $d$ such that
there is no district admissions rule of $d$ that has district-level type-specific ceilings,
is $d$-weakly acceptant, and
satisfies IRC and weak substitutability.
\end{customthm}

To show this result, consider a district $d$ with $k_{d}=2$. There are three schools $c_{1}$,
$c_{2}$, $c_{3}$ in the district, each with capacity one, and four students
$s_{1}$, $s_{2}$, $s_{3}$, $s_{4}$  of which two are from a different district. Students $s_{1}$ and $s_{2}$
are of type $t_1$ and students $s_{3}$ and $s_{4}$ are of type $t_2$. The
district-level type-specific
ceilings are as follows: $q_{d}^{t_1}=q_{d}^{t_2}=1$.

Suppose, for contradiction, that the district admissions rule has district-level type-specific ceilings,
is $d$-weakly acceptant, and satisfies IRC and weak substitutability.

Consider $Ch_{d}(\{(s_{1},c_{1}),(s_{2},c_{1}),(s_{3},c_{1}),(s_{4},c_{1}) \})$.
Since types are symmetric and two students
are symmetric within each type, without loss of generality, we can assume
$Ch_{d}(\{(s_{1},c_{1}),(s_{2},c_{1}),( s_{3},c_{1}) ,(
s_{4},c_{1}) \})=\{( s_{1},c_{1}) \}$ because $q_{c_1}=1$.

Next, consider $Ch_{d}(\{( s_{2},c_{2}) ,( s_{3},c_{2})
,( s_{4},c_{2}) \})$. Because $q_{c_2}=1$ and $Ch_d$ is $d$-weakly acceptant,
this is either equal to $\{(s_{2},c_{2}) \}$ or $\{( s_{3},c_{2}) \}$ (the case when it
is equal to $\{( s_{4},c_{2}) \}$ is symmetric to the case when
$\{( s_{3},c_{2}) \}$. We analyze these two cases separately.

\begin{enumerate}
\item Suppose $Ch_{d}(\{( s_{2},c_{2}) ,( s_{3},c_{2}),( s_{4},c_{2}) \})=\{( s_{2},c_{2}) \}$. Then, by IRC,
we conclude that $Ch_{d}(\{( s_{2},c_{2}) ,(s_{3},c_{2}) \})=\{( s_{2},c_{2}) \}$. Next, we argue that
$Ch_{d}(\{( s_{1},c_{1}) ,( s_{2},c_{2}) ,(s_{3},c_{2}) \})=\{( s_{2},c_{2}) \}$. This is because the
only two cases that satisfy d-weak acceptance and type-specific ceilings
are $\{( s_{2},c_{2}) \}$ and $\{( s_{1},c_{1}),( s_{3},c_{2}) \}$. The latter would violate weak
substitutability since in that case $( s_{3},c_{2})$ would be accepted in a larger set
$\{( s_{1},c_{1}) ,(s_{2},c_{2}) ,( s_{3},c_{2}) \}$ and rejected from a
smaller set $\{( s_{2},c_{2}) ,( s_{3},c_{2}) \}$.
Then, by IRC, $Ch_{d}(\{( s_{1},c_{1}) ,( s_{2},c_{2}),( s_{3},c_{2}) \})=\{( s_{2},c_{2}) \}$ implies
$Ch_{d}(\{( s_{1},c_{1}) ,( s_{2},c_{2}) \})=\{(s_{2},c_{2}) \}$. Then we note
that $Ch_{d} (\{( s_{1},c_{1}),( s_{2},c_{2}) ,( s_{3},c_{1})\})=\{(
s_{2},c_{2}) ,( s_{3},c_{1}) \}$ since by weak substitutability
$( s_{1},c_{1}) $ cannot be chosen, and
therefore $( s_{2},c_{2}) $ and $( s_{3},c_{1}) $ have
to be chosen due to $d$-weak acceptance. Next, again by weak substitutability,
we note that $Ch_{d}(\{( s_{1},c_{1}),( s_{2},c_{2}) ,( s_{3},c_{1}) \})=\{(
s_{2},c_{2}) ,( s_{3},c_{1}) \}$ implies $Ch_{d}(\{(
s_{1},c_{1}) ,( s_{3},c_{1}) \})=\{( s_{3},c_{1})
\}$. Finally, we note that this contradicts $Ch_{d}
( \{(s_{1},c_{1}),(s_{2},c_{1}),( s_{3},c_{1}) ,(
s_{4},c_{1}) \} )=\{( s_{1},c_{1}) \}$ and IRC.

\item Suppose $Ch_{d}(\{( s_{2},c_{2}) ,( s_{3},c_{2})
,( s_{4},c_{2}) \})=\{( s_{3},c_{2}) \}$. Consider $
Ch_{d}(\{( s_{2},c_{3}) ,( s_{4},c_{3}) \})$. Because $q_{c_3}=1$ and $Ch_d$ is $d$-weakly acceptant,
this is either $\{( s_{2},c_{3}) \}$ or $\{( s_{4},c_{3})\}$. We consider
these two possible cases separately. These two subcases
will follow similar arguments to Case (1) above and change the indices
appropriately in order to get a contradiction.

\begin{enumerate}
\item Suppose $Ch_{d}(\{( s_{2},c_{3}) ,( s_{4},c_{3})\})=\{( s_{2},c_{3}) \}$. Next, we argue that
$Ch_{d}(\{(s_{1},c_{1}) ,( s_{2},c_{3}) ,( s_{4},c_{3})
\}) =\{( s_{2},c_{3}) \}$. This is because the only two cases that
satisfy d-weak acceptance and type-specific ceilings are $\{(
s_{2},c_{3}) \}$ and $\{( s_{1},c_{1}) ,(
s_{4},c_{3}) \}$. The latter would violate weak substitutability since in that case $( s_{4},c_{3}) $ would be
accepted in a larger set $\{( s_{1},c_{1}) ,(
s_{2},c_{3}) ,( s_{4},c_{3}) \}$ and rejected from a
smaller set $\{( s_{2},c_{3}) ,( s_{4},c_{3}) \}$.
Then, by IRC, $Ch_{d}(\{( s_{1},c_{1}) ,( s_{2},c_{3})
,( s_{4},c_{3}) \})=\{( s_{2},c_{3}) \}$ implies $%
Ch_{d}(\{( s_{1},c_{1}) ,( s_{2},c_{3}) \})=\{(
s_{2},c_{3}) \}$. Then we note that $Ch_{d}(\{( s_{1},c_{1})
,( s_{2},c_{3}) ,( s_{4},c_{1}) \})=\{(
s_{2},c_{3}) ,( s_{4},c_{1}) \}$ since by weak substitutability $( s_{1},c_{1}) $ cannot to be chosen,
therefore $( s_{2},c_{3}) $ and $( s_{4},c_{1}) $ have
to be chosen due to d-weak acceptance. Next, again by weak substitutability, we note that $Ch_{d}(\{(
s_{1},c_{1}) ,( s_{2},c_{3}) ,( s_{4},c_{1})
\})=\{( s_{2},c_{3}) ,( s_{4},c_{1}) \}$ implies $%
Ch_{d}(\{
( s_{1},c_{1}) ,( s_{4},c_{1}) \})=\{(
s_{4},c_{1}) \}$. Finally, we note that this contradicts
$Ch_{d}(\{(s_{1},c_{1}),(s_{2},c_{1}),( s_{3},c_{1}) ,(
s_{4},c_{1}) \})=\{( s_{1},c_{1}) \}$ and IRC.

\item Suppose $Ch_{d}(\{( s_{2},c_{3}) ,( s_{4},c_{3})
\})=\{
( s_{4},c_{3}) \}$. Next, we argue that $Ch_{d}(\{(
s_{2},c_{3}) ,( s_{3},c_{2}) ,( s_{4},c_{3})
\})=\{( s_{4},c_{3}) \}$. This is because the only two cases that
satisfy d-weak acceptance and type-specific ceilings are $\{(
s_{4},c_{3}) \}$ and $\{( s_{2},c_{3}) ,(
s_{3},c_{2}) \}$. The latter would violate weak substitutability since in that case $( s_{2},c_{3}) $ would be
accepted in a larger set $\{( s_{2},c_{3}) ,(
s_{3},c_{2}) ,( s_{4},c_{3}) \}$ and rejected from a
smaller set $\{( s_{2},c_{3}) ,( s_{4},c_{3}) \}$.
Then, by IRC, $Ch_{d}(\{( s_{2},c_{3}) ,( s_{3},c_{2})
,( s_{4},c_{3}) \})=\{( s_{4},c_{3}) \}$ implies $%
Ch_{d}(\{( s_{3},c_{2}) ,( s_{4},c_{3}) \})=\{(
s_{4},c_{3}) \}$. Then we note that $Ch_{d}(\{( s_{2},c_{2})
,( s_{3},c_{2}) ,( s_{4},c_{3}) \})=\{(
s_{2},c_{2}) ,( s_{4},c_{3}) \}$ since by weak substitutability $( s_{3},c_{2}) $ cannot to be chosen,
therefore $( s_{4},c_{3}) $ and $( s_{2},c_{2}) $ have
to be chosen due to d-weak acceptance. Next, again by weak substitutability, we note that $Ch_{d}(\{(
s_{2},c_{2}) ,( s_{3},c_{2}) ,( s_{4},c_{3})
\})=\{
( s_{2},c_{2}) ,( s_{4},c_{3}) \}$ implies $%
Ch_{d}(\{( s_{2},c_{2}) ,( s_{3},c_{2}) \})=\{(
s_{2},c_{2}) \}$. Finally, we note that this contradicts
$Ch_{d}(\{( s_{2},c_{2}) ,( s_{3},c_{2}) ,(
s_{4},c_{2}) \})=\{( s_{3},c_{2}) \}$ and IRC.

\end{enumerate}
\end{enumerate}

\end{proof}

\end{document}